%% file: main.tex
\documentclass{article}
\usepackage[utf8]{inputenc}
\input{header.tex}
\input{newcommands.tex}

\title{Monotone Submodular Multiway Partition}
\author{
Richard Bi\thanks{UIUC, Email: \{rbi3, karthe\}@illinois.edu. Supported in part by NSF grant CCF-2402667.}
\and Karthekeyan Chandrasekaran\footnotemark[1] 
\and Soham Joshi\thanks{IIT Bombay, Email: sohamjoshi@cse.iitb.ac.in. Work done while visiting UIUC and was supported in part by NSF grant CCF-2402667.}
}
\begin{document}

\maketitle
\input{abstract}

\newpage
\tableofcontents
\newpage

\setcounter{page}{1}
\input{intro}
\input{mono-sub-mp}

\input{graph-coverage-MP}

%\section{Notes}
%Need to Cite this: \url{https://proceedings.neurips.cc/paper_files/paper/2015/file/dc960c46c38bd16e953d97cdeefdbc68-Paper.pdf}

\bibliographystyle{abbrv}
\bibliography{references} % Entries are in the references.bib file

\end{document}

%% file: header.tex
\usepackage[T1]{fontenc}
\usepackage{soul, comment} 
\usepackage{phfqit}
\usepackage{proba} 
\usepackage[sort,nocompress]{cite}
\usepackage{amssymb}
\usepackage{amsmath}
\usepackage{amsthm}
\usepackage{amsfonts}
\usepackage{mathtools}
\usepackage{xspace}
\usepackage{bm}
\usepackage{nicefrac}
\usepackage{commath}
\usepackage{bbm}
\usepackage{boxedminipage}
\usepackage{xparse}
\usepackage{xcolor}
\usepackage{float}
\usepackage{multirow}
\usepackage{makecell}
\usepackage{graphicx}
\usepackage{caption}
\usepackage{subcaption}
\usepackage{ifthen}
\usepackage{thm-restate}
\usepackage{algorithm}
\usepackage[noend]{algpseudocode}
\usepackage{colortbl} 
\usepackage{fancyhdr}   
\usepackage{hyperref}
    \hypersetup{ colorlinks=true, linkcolor=blue, citecolor=magenta, urlcolor=blue,}
\usepackage{fullpage}
\usepackage[utf8]{inputenc}
\usepackage{environ}
\usepackage{mdframed}
\usepackage{cleveref}
\usepackage[short]{optidef} 
\usepackage{enumitem}
\usepackage{tikz, ifthen}
\usetikzlibrary{shapes}
\usetikzlibrary{positioning}
\usetikzlibrary{fit}
\usetikzlibrary{graphs,graphs.standard}
\usepackage{enumitem}
\newtheorem{proposition}{Proposition}[section]
\newtheorem{lemma}{Lemma}[section]
\newtheorem{remark}{Remark}[section]
\newtheorem{theorem}{Theorem}[section]
\newtheorem{corollary}{Corollary}[section]

\newtheorem{definition}{Definition}[section]
\newtheorem{conjecture}{Conjecture}[section]
\newtheorem{observation}{Observation}[section]

\newtheorem{claim}{Claim}[section]

\newtheorem{question}{Question}
\newtheorem{case}{Case}
\usepackage{enumitem}

\newcommand{\eps}{\ensuremath{\varepsilon}}

\NewEnviron{problem}[1]{%
	\begin{center}\fbox{\parbox{6in}{%
				{\centering\scshape #1\par}%
				\parskip=1ex
				\everypar{\hangindent=1em}%
				\BODY
}}\end{center}}

%% file: newcommands.tex
\newcommand{\vp}{\mathbf{p}}
\newcommand{\vx}{\mathbf{x}}
\newcommand{\vy}{\mathbf{y}}

\DeclareMathOperator*{\argmin}{arg\,min}
\DeclareMathOperator*{\argmax}{arg\,max}
\newcommand{\obj}[1]{\ensuremath{\textsc{obj}\left(#1\right)}\xspace}
\newcommand{\swapi}[1]{\ensuremath{\text{swap}_i\left(#1\right)}\xspace}
\newcommand{\phin}[1]{\ensuremath{\phi_n\left(\left|#1\right|\right)}\xspace}
\newcommand{\phit}[1]{\ensuremath{\phi_t\left(\left|#1\right|\right)}\xspace}

\newcommand{\mlsymopt}{\textsc{ML-OPT}_\textsc{Sym}\xspace}
\newcommand{\mlopt}{\textsc{ML-OPT}\xspace}
\newcommand{\opt}{\textsc{OPT}\xspace}
\newcommand{\symopt}{\ensuremath{\textsc{OPT}_\textsc{Sym}}\xspace}
\newcommand{\symgap}{\ensuremath{\textsc{Sym-Gap}}\xspace}

\newcommand{\indicator}[1]{\ensuremath{\mathbbm{1}_{#1}}\xspace}

\newcommand{\submp}{\textsc{Sub-MP}\xspace}
\newcommand{\symsubmp}{\textsc{Sym-Sub-MP}\xspace}
\newcommand{\monosubmp}{\textsc{Mono-Sub-MP}\xspace}
\newcommand{\submprel}{\textsc{Sub-MP-Rel}\xspace}
\newcommand{\matrixmp}{\textsc{Matrix-MP}\xspace}
\newcommand{\matroidmp}{\textsc{Matroid-MP}\xspace}
\newcommand{\covmp}{\textsc{Coverage-MP}\xspace}
\newcommand{\gcovmp}{\textsc{Graph-Coverage-MP}\xspace}

\newcommand{\mwc}{\textsc{Graph-Multiway-Cut}\xspace}
\newcommand{\maxcut}{\textsc{Max-cut}\xspace}

\newcommand{\knote}[1]{{\color{red}[{\tiny Karthik: \bf #1}]\marginpar{\color{red}*}}}
\newcommand{\donote}[1]{{\color{blue}[{\tiny TODO: \bf #1}]\marginpar{\color{red}*}}}

\newcommand{\snote}[1]{{\color{red}[{\tiny Soham: \bf #1}]\marginpar{\color{red}*}}}

%% file: abstract.tex
\begin{abstract}
In submodular multiway partition (\submp), the input is a non-negative submodular function $f:2^V\rightarrow \R_{\ge 0}$ given by an evaluation oracle along with $k$ terminals $t_1, t_2, \ldots, t_k\in V$. The goal is to find a partition $V_1, V_2, \ldots, V_k$ of $V$ with $t_i\in V_i$ for every $i\in [k]$ in order to minimize $\sum_{i=1}^k f(V_i)$. In this work, we focus on \submp when the input function is monotone (termed \monosubmp). 
\monosubmp formulates partitioning problems over several interesting structures---e.g., matrices, matroids, graphs, and hypergraphs. 
\monosubmp is NP-hard since the graph multiway cut problem can be cast as a special case. We investigate the approximability of \monosubmp: we show that it admits a $4/3$-approximation and does not admit a $(10/9-\epsilon)$-approximation for every constant $\epsilon>0$. %Our approximation algorithm is based on a new rounding of the convex relaxation while our lower bound result is based on constructing an instance with large symmetry gap. 
%assuming that the algorithm makes polynomial number of function evaluation queries. 
Next, we study a special case of \monosubmp where the monotone submodular function of interest is the coverage function of an input graph, termed \gcovmp. 
%Next, we study a special case of \monosubmp, namely \gcovmp: the input here is an edge-weighted graph $(G=(V, E), w: E\rightarrow \R_{\ge 0})$ and $k$ terminals $t_1, t_2, \ldots, t_k\in V$. The goal is to find a partition $V_1, V_2, \ldots, V_k$ of $V$ with $t_i\in V_i$ for every $i\in [k]$ in order to minimize $\sum_{i=1}^k b(V_i)$, where $b: 2^V\rightarrow \R_{\ge 0}$ is the coverage function of the graph. 
\gcovmp is equivalent to the classic multiway cut problem for the purposes of exact optimization.
%, but the two problems do not have an approximation preserving reduction in either direction. 
We show that \gcovmp admits a $1.125$-approximation and does not admit a $(1.00074-\epsilon)$-approximation for every constant $\epsilon>0$ assuming the Unique Games Conjecture. 
These results separate \gcovmp from graph multiway cut in terms of approximability. 
%show a separation between graph multiway cut %\mwc 
%and \gcovmp in terms of approximability. 

%\keywords{Approximation Algorithms \and Multiway Partition  \and Monotone Submodular Functions.}
\end{abstract}

%% file: intro.tex
\section{Introduction}
A set function $f:2^V\rightarrow \R$ is \emph{submodular} if $f(A)+f(B)\ge f(A\cap B) + f(A\cup B)$ for every $A, B\subseteq V$.  
%We consider the submodular multiway partition problem for monotone submodular functions. 
%Several well-studied partitioning problems can be cast as special cases of submodular multiway partition. 
%Submodular multiway partitioning problems generalize several well-studied partitioning problems. 
We focus on the submodular multiway partition problem denoted \submp:  
%The submodular multiway partition problem, denoted \submp, was introduced by Chekuri and Ene: 
The input here is a non-negative monotone submodular function $f: 2^V\rightarrow \R_{\ge 0}$ given by an evaluation oracle\footnote{The evaluation oracle for a function $f:2^V\rightarrow \R$ takes a set $S\subseteq V$ as input and returns $f(S)$.} and a collection $T=\{t_1, t_2, \ldots, t_k\}\subseteq V$ of terminals. 
The goal is to find a partition $V_1, V_2, \ldots, V_k$ of $V$ with $t_i\in V_i$ for every $i\in [k]$ in order to minimize $\sum_{i=1}^k f(V_i)$, i.e., 
\[
\min\left\{\sum_{i=1}^k f(V_i): V_1, V_2, \ldots, V_k \text{ is a partition of }V \text{ such that } t_i\in V_i\ \forall\ i\in [k] \right\}. 
\]
Throughout this work, 
%we consider the ground set $V$ to contain $n$ elements labeled as $\{v_1, v_2, \ldots, v_n\}$ and the terminal set $T$ to contain $k$ elements. We also 
we assume that the input function $f$ is non-negative, i.e., $f(A)\ge 0$ for every $A\subseteq V$ and normalized, i.e., $f(\emptyset) = 0$. 
\submp was introduced by 
Zhao, Nagamochi, and Ibaraki \cite{ZNI05} as a unified generalization of several partitioning problems---e.g., edge multiway cut in graphs/hypergraphs and node multiway cut in graphs. 
%Chekuri and Ene identified that node multiway cut in graphs can also be viewed as a special case of \submp. 
%Chekuri and Ene as a unified generalization of several partitioning problems (e.g., edge and node weighted multiway cut in graphs and edge weighted multiway cut in hypergraphs). 
We recall that a set function $f: 2^V\rightarrow \R$  is \emph{symmetric} if $f(A) = f(V\setminus A)$ for every $A\subseteq V$ and is \emph{monotone} if $f(A)\le f(B)$ for every $A\subseteq B\subseteq V$. 
In \submp, if the input function is symmetric, then we call the problem as symmetric submodular multiway partition, denoted \symsubmp and if the input function is monotone, then we call the problem as monotone submodular multiway partition, denoted \monosubmp. 
The focus of this work is on \monosubmp. 

%In this work, we consider the Monotone Submodular Multiway Partition problem, denoted \monosubmp: the input here is a normalized non-negative monotone submodular function $f: 2^V\rightarrow \R_{\ge 0}$ given by an evaluation oracle\footnote{The evaluation oracle for a function $f:2^V\rightarrow \R$ takes a set $S\subseteq V$ as input and returns $f(S)$.} and a collection $T=\{t_1, t_2, \ldots, t_k\}$ of terminals. The goal is to find a partition $V_1, V_2, \ldots, V_k$ of $V$ with $t_i\in V_i$ for every $i\in [k]$ in order to minimize $\sum_{i=1}^k f(V_i)$. Throughout this work, we consider the ground set $V$ to contain $n$ elements labeled as $\{v_1, v_2, \ldots, v_n\}$.

\paragraph{Motivations and Connections.} 
%While \submp and \symsubmp have been well-studied in the literature \cite{ZNI05, CE, EVW}, \monosubmp has not witnessed much prior work. 
We motivate \monosubmp via three important special cases: 
(1) In Matrix Multiway Partition, denoted \matrixmp, the input is a matrix $M\in \R^{n\times m}$ along with $k$ specified rows $r_1, r_2, \ldots, r_k\in [n]$ and the goal is to partition the $n$ rows into $V_1, V_2, \ldots, V_k$ with $r_i\in V_i$ for every $i\in [k]$ in order to minimize the sum of the dimensions of the subspaces spanned by the parts, i.e., minimize $\sum_{i=1}^k \text{dimension}(V_i)$. This is a special case of \monosubmp since the dimension function is non-negative, normalized, monotone, and submodular. \matrixmp is a fundamental linear algebra problem and is likely to have applications in data analysis. 
(2) In Matroid Multiway Partition, denoted \matroidmp, the input is the  rank function $r: 2^V\rightarrow \Z_{\ge 0}$ of a matroid (given by an evaluation oracle) along with a collection $t_1, \ldots, t_k\in V$ of elements. The goal is to partition the ground set $V$ into $k$ parts $V_1, V_2, \ldots, V_k$ with $t_i\in V_i$ for every $i\in [k]$ in order to minimize $\sum_{i=1}^k r(V_i)$. \matroidmp is a special case of \monosubmp since the matroid rank function is non-negative, normalized, monotone, and submodular. We note that \matrixmp is a special case of \matroidmp. 
(3) In Coverage Multiway Partition, denoted \covmp, the input is a hypergraph $H=(V, E)$ along with $k$ terminals $t_1, t_2, \ldots, t_k \in V$ and the goal is to partition the vertex set into $k$ parts $V_1, V_2,\ldots, V_k$ in order to minimize the sum of the coverage values of the parts, i.e., minimize $\sum_{i=1}^k b(V_i)$, where $b(X)$ is the number of hyperedges in $H$ intersecting the vertices in $X$. \covmp is a special case of \monosubmp since the coverage function $b: 2^V\rightarrow \Z_{\ge 0}$ is non-negative, normalized, monotone, and submodular. We emphasize that all three special cases mentioned above are themselves fairly general and numerous other partitioning problems can be cast as special cases of these problems as well. Thus, our motivation behind studying \monosubmp is to understand whether we can get good approximations for all these problems simultaneously. 

%\gcovmp is closely related to the classic multiway cut problem in graphs (the objective in \gcovmp is an additive translation of the multiway cut objective)---we will discuss this connection later. 
%We also study a special case of \monosubmp, namely \gcovmp. 
In \covmp, if the input hypergraph is simply a graph, then we denote the resulting problem as \gcovmp. \gcovmp is an important special case of \monosubmp, so we will study this special case as well in this work. 
\gcovmp is closely related to the well-studied graph multiway cut problem (denoted \mwc). The input in both problems is a graph with edge weights $G=(V, E, w: E\rightarrow \R_+)$ and a collection of $k$ terminals $t_1, \ldots, t_k\in V$. 
The goal in both problems is to find a partition $V_1, \ldots, V_k$ of the vertex set with $t_i\in V_i$  for every $i\in [k]$ in order to minimize an objective function. 
The objective function in \gcovmp is $\sum_{i=1}^k b(S_i)$, where $b:2^V\rightarrow \R_{\ge 0}$ is the graph coverage function defined as $b(S):=\sum_{e\in E: e\cap S\neq \emptyset}w_e$ for all $S\subseteq V$. 
The objective function in \mwc is $(1/2)\sum_{i=1}^k d(S_i)$, where $d:2^V\rightarrow \R_{\ge 0}$ is the graph cut function defined as $d(S):=\sum_{e\in E: e\cap S \neq \emptyset, e\setminus S \neq \emptyset}w_e$ for all $S\subseteq V$. We recall that the graph coverage function is monotone submodular while the graph cut function is symmetric submodular. Moreover, the objective in \gcovmp is a translation of the objective in \mwc by an additive $w(E):=\sum_{e\in E}w_e$. This connection immediately implies that \gcovmp is NP-hard (since \mwc is known to be NP-hard). However,  there is no approximation preserving reduction between \mwc and \gcovmp (in either direction). This situation raises the following question: to what extent do the approximation and inapproximability techniques for \mwc extend to \gcovmp? Motivated by this question, we study the approximability of \gcovmp. We refer the reader to Figure \ref{fig:reductions} showing the relationships between various submodular partitioning problems discussed here. 

\begin{figure}[ht]
\centering
\includegraphics[width = 0.7\textwidth]{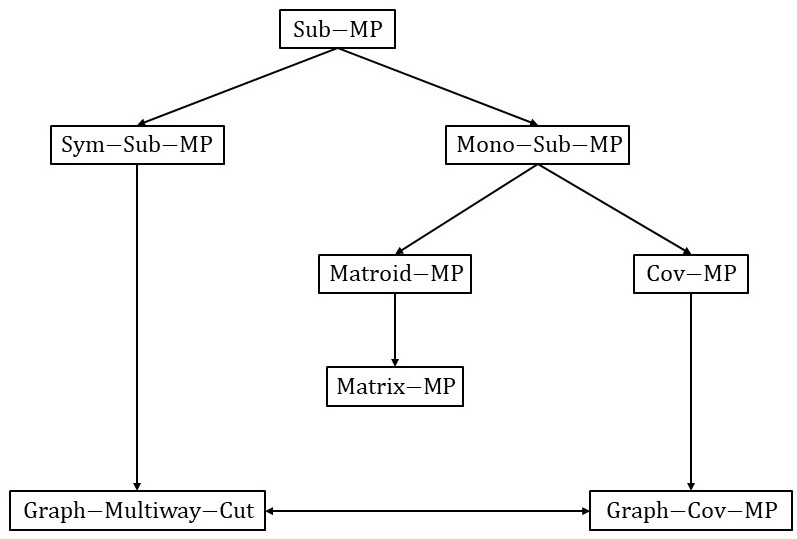}
\caption{Reductions between submodular multiway partitioning problems. Arrow from Problem A to Problem B implies that Problem B is a special case of Problem A. \mwc and \gcovmp are equivalent for the purposes of exact optimization but they do not admit approximation preserving reductions in either direction. }
\label{fig:reductions}
\end{figure}
%\knote{Problem names in the figure need to be consistent with the text.}

\paragraph{Approximability.} %Next we turn to the approximability status of submodular multiway partition problems. 
%Before discussing the approximability of \monosubmp, we briefly recall the status of \submp and \symsubmp. 
\submp, \symsubmp, and \monosubmp are all NP-hard since the classic graph multiway cut problem can be cast as a special case of all these problems. 
%As mentioned above, \gcovmp is NP-hard and consequently, \monosubmp is also NP-hard. 
Zhao, Nagamochi, and Ibaraki \cite{ZNI05} introduced a greedy splitting approach and showed that it implies a $(k-1)$-approximation for \submp and a $(2-2/k)$-approximation for both \symsubmp and \monosubmp. 
Chekuri and Ene \cite{CE} introduced the Lov\'{a}sz extension based convex programming relaxation for \submp. They designed rounding algorithms for this convex program that achieved a $2$-approximation for \submp and a $3/2$-approximation for \symsubmp---we note that their rounding algorithms are different for \submp and \symsubmp. 
%Chekuri and Ene showed that \submp admits a $2$-approximation and \symsubmp admits a $3/2$-approximation. 
%by designing rounding algorithms for a convex relaxation of the problem. 
In subsequent work, Ene, Vondr\'{a}k, and Wu \cite{EVW} improved the rounding algorithm for \submp and showed that it achieves a $(2-2/k)$-approximation. 
%relative to the convex relaxation. 
On the hardness side, Ene, Vondr\'{a}k, and Wu \cite{EVW} showed that for every constant $\epsilon>0$, (1) a $(2-2/k-\epsilon)$-approximation for \submp requires exponentially many function evaluation queries and (2) a $(1.268-\epsilon)$-approximation for \symsubmp requires exponentially many function evaluation queries. 
%We note that all algorithmic results mentioned above are by rounding an optimum solution to a convex relaxation that is based on the Lov\'{a}sz extension of the input function. 

We now discuss the complexity and approximability of \monosubmp. 
%\mwc, a well-studied NP-hard problem, is equivalent to \gcovmp for the purposes of exact optimization. We recall that \gcovmp is a special case of \monosubmp. Consequently, \monosubmp is NP-hard. 
As mentioned above, \gcovmp is NP-hard and consequently, \monosubmp is also NP-hard. 
As mentioned in the previous paragraph, two independent prior approaches for \monosubmp achieved a $(2-2/k)$-approximation: (i) the greedy splitting technique \cite{ZNI05} and (ii) rounding an optimum solution to the Lov\'{a}sz extension based convex program for \submp \cite{CE}. 
%Zhao, Nagamochi, and Ibaraki \cite{ZNI05} showed a greedy splitting approach to obtain a $(2-2/k)$-approximation for \monosubmp. 
We observe that there is a considerably simpler and faster (linear-time) algorithm for \monosubmp that achieves a $(2-1/k)$-approximation: order the terminal set as $t_1, t_2, \ldots, t_k$ where $f(\{t_1\})\le f(\{t_2\})\le \ldots \le f(\{t_k\})$ and return the partition $V_1:=\{t_1\}$, $V_2 :=\{t_2\}$, \ldots, $V_{k-1}:=\{t_{k-1}\}$, $V_k:=(V\setminus T)\cup \{t_k\}$. 
See Proposition \ref{prop:monosubmp-easy-2-approx} for the approximation factor analysis of this algorithm. 
Given this simple and linear-time $2$-approximation and the plethora of applications that can be cast as \monosubmp, we focus on understanding the approximability of \monosubmp in this work. Another technical motivation to study \monosubmp 
arises from the convex program that is used to approximate \submp. 
%comes from the following questions: what is the integrality gap of the convex program that is used to approximate \submp 
%if we additionally assume that the input function is monotone, then what is the integrality gap of this convex program? 
If the input function is monotone in addition to being submodular, then can we design a rounding algorithm for the convex program with better approximation guarantee (i.e, better than what is possible for submodular or symmetric submodular functions)? 
%Can we design a rounding algorithm for the convex program with better approximation guarantee for \monosubmp than what is possible for \submp and \symsubmp? 
%What is the inapproximability of \monosubmp? 

As mentioned earlier, we also study \gcovmp since it is a special case of \monosubmp and is closely related to the well-studied \mwc. \mwc is a fundamental graph partitioning problem and has been studied extensively in the literature. Its rich structure enables an approximation preserving reduction to a geometric problem, namely simplex  partitioning \cite{CKR00, KKSTY04, CCT06}. In spite of this reduction to a geometric problem, pinning down the approximability of \mwc has remained a tantalizing open question with several works devoted to narrowing the gap between the upper and lower bounds 
%with several works giving better rounding algorithms and improving the integrality gap of an associated linear program 
\cite{freund-karloff, MNRS08, BNS13, SV14, AMM17, BCKM20}. 
For \mwc, the current largest inapproximation factor is $1.200016$ and the current smallest approximation factor is $1.2965$. 
%, both of which are achieved via a linear program (rounding algorithms to this linear program lead to good approximation factors while integrality gap of this linear program gives the hardness of approximation). 
Although \gcovmp is equivalent to \mwc for the purposes of exact optimization, we will later see that they are not equivalent in terms of approximability. 
%In this work, we study the approximability of \gcovmp. 

%In this paper \knote{work}, we consider the Monotone Submodular Multiway Partition problem (Mono-Sub-MP) \knote{\monosubmp}, which is defined as follows. Let $f \colon 2^V \rightarrow \mathbb{R}_{\geq 0}$ be a non-negative monotone submodular function over the ground set $V = \{v_1, \ldots, v_n\}$. Let $T = \{t_1, \ldots, t_k\} \subseteq V$ be a set of $k \in \mathbb{N}$ terminals. In Mono-Sub-MP, we seek a partition $A_1, \ldots, A_k$ of $V$ such that for each $i \in [k]$, $t_i \in A_i$ and $\sum_{i = 1}^k f(A_i)$ is minimized.

%\donote{Would be good to include a figure showing \submp at the top, \symsubmp and \monosubmp as special cases of \submp; \matrixmp, \matroidmp, \covmp, \gcovmp as special cases of \monosubmp; \mwc as a special case of \symsubmp.}

\subsection{Results and Techniques}
As mentioned previously, there are three known $2$-approximations for \monosubmp---the greedy splitting approach, the convex program approach, and the simple linear-time greedy approach described above. Our first result improves the approximability of \monosubmp. 
\begin{theorem}\label{thm:mono-sub-mp-upper-bound}
    \monosubmp admits a $4/3$-approximation. 
\end{theorem}
This result is based on rounding an optimum solution to the Lov\'{a}sz extension based convex program for \submp introduced by Chekuri and Ene \cite{CE}. 
%We emphaize that our rounding algorithm and analysis are different from that of Chekuri and Ene. 
We use a threshold rounding algorithm similar to Chekuri and Ene. However, the interval for the threshold that we use is different from theirs. For \submp, Chekuri and Ene pick a uniform random threshold from the interval $[1/2,1]$ in order to ensure that sets obtained by rounding do not overlap and hence, give a partition. For \symsubmp, they pick a uniform random threshold from the interval $[0,1]$ which leads to overlapping sets, but these sets can be uncrossed to obtain a partition without increasing the objective value owing to the posimodularity property of symmetric submodular functions. For \monosubmp, we pick a uniform random threshold from the interval $[1/4,1]$. This again leads to overlapping sets, but we uncross these sets arbitrarily without increasing the objective by exploiting the monotonicity of the function. We emphasize that our approximation factor analysis differs considerably from that of Chekuri and Ene as well as the subsequent work of Ene, Vondr\'{a}k, and Wu \cite{EVW} (although we use some of their ingredients). Our analysis differs since our threshold interval is different and we need to exploit monotonicity to bound the approximation factor. Moreover, the  approximation factor that we achieve for \monosubmp is $4/3$ which is smaller than the known approximation factor for \submp and \symsubmp. 
We show that the approximation factor analysis of our algorithm is tight via an example. Hence, improving on the approximability requires new ideas. 
%and is not a simple exploitation of monotonicity in their results. 

Our next result shows a lower bound on the approximability assuming that the algorithm makes polynomial number of function evaluation queries. 

\begin{restatable}{theorem}{thmMonoSubMPLowerBound}\label{thm:mono-sub-mp-lower-bound}
For every constant $\epsilon > 0$, every algorithm for \monosubmp that achieves an approximation factor of $(\frac{10}{9} - \epsilon)$ requires $2^{\Omega(n)}$ function evaluation queries, where $n$ is the size of the ground set.

%For every constant $\epsilon>0$, there does not exist an algorithm for \monosubmp that achieves $(11/10-\epsilon)$-approximation using polynomial number of function evaluation queries. 
\end{restatable}

This lower bound result on the inapproximability of \monosubmp is based on the  \emph{symmetry gap} machinery introduced by Vondr\'{a}k \cite{Vondrak}. 
%Informally, the symmetry gap of a function relative to an optimization problem is the largest multiplicative gap between the objective values of optimum symmetric solution (under appropriately chosen notion of symmetry) and optimum solution to the problem. 
%Vondr\'{a}k introduced this notion and used it to show showed that large symmetry gap implies inapproximability factors for certain constrained submodular maximization problems in the valuation oracle input model (i.e., assuming that the algorithm makes only polynomial number of function evaluation queries). 
Vondr\'{a}k introduced the symmetry gap machinery to show 
%lower bounds on the approximability of 
inapproximability results 
in the valuation oracle input model (i.e., assuming that the algorithm makes only polynomial number of function evaluation queries) 
for constrained submodular \emph{maximization} problems. Subsequently, Ene, Vondr\'{a}k, and Wu \cite{EVW} adapted this machinery to show 
%lower bounds on the approximability of 
inapproximability results in the valuation oracle model for 
\submp and \symsubmp, which are \emph{minimization} problems related to submodular objectives. 
%Vondr\'{a}k's machinery (with mild adaption) immediately implies that the symmetry 
Informally, the symmetry gap of a function relative to an optimization problem is the largest ratio between the objective values of an optimum symmetric solution (under an appropriately chosen notion of symmetry) and an optimum solution to the problem. 
We adapt Vondr\'{a}k's \cite{Vondrak} and Ene, Vondr\'{a}k, and Wu's \cite{EVW} machinery in a straightforward way to conclude that the symmetry gap of \monosubmp is a lower bound on its inapproximability factor. %This boils down 
Given this status, the 
%non-trivial crux of 
technical challenge in using 
the symmetry gap machinery for inapproximability involves 
%the problem crux of the machinery involves 
constructing instances of the problem which have large symmetry gap. This is our key contribution: we construct an instance of \monosubmp that has symmetry gap at least $10/9$. We note that our analysis of the symmetry gap of our instance is tight. Thus, improving the inapproximability factor beyond $10/9$ via the symmetry gap machinery requires fundamentally new types of instances. 
We mention that the symmetry gap machinery allows query complexity based inapproximability results to be converted to computation based inapproximability results assuming NP $\neq$ RP \cite{DV12}. 
%via the techniques of Dobzinski and Vondr\'{a}k \cite{DV12}. 

%\knote{Is there a trivial $2$-approximation for \monosubmp? Monotonicity is rather powerful, so trivial algorithms tend to give good approximations. How about this algorithm: Let $N:=V-T$; now for each $i=1,\ldots, k$, consider the partition $(\{v_1\}, \{v_2\}, \ldots, \{v_{i-1}\}, \{v_i\}\cup N, \{v_{i+1}\}, \ldots, \{v_k\})$, compute its objective value, and return the cheapest. Is this a $2$-approximation? }

Next, we turn to \gcovmp which is a special case of \monosubmp. 
%closely related to \mwc. 
Its close relationship to \mwc inspires our results. We note that NP-hardness of \mwc implies the NP-hardness of  \gcovmp. 
%Next, we turn to the approximability of \gcovmp. 
It is easy to show that if $\mwc$ admits an $\alpha$-approximation, then $\gcovmp$ admits a $(1+\alpha)/2$-approximation (see Proposition \ref{prop:mwc-to-gcovmp-approximation}). Thus, the current best $1.2965$-approximation for \mwc implies a $1.14825$-approximation for \gcovmp. We improve the approximability of \gcovmp beyond this trivial guarantee. 
\begin{restatable}{theorem}{thmGcovMPApprox}\label{thm:gcovmp-approximation}
\gcovmp admits a $1.125$-approximation. 
\end{restatable}

Our algorithm for \gcovmp is based on rounding an optimum solution to a linear programming relaxation. 
%the \emph{exponential clocks rounding} for \mwc. 
The linear programming relaxation associated with \gcovmp has a different objective compared to that of \mwc, but the constraints are the same. Buchbinder, Naor, and Schwartz \cite{BNS13} designed the exponential clocks rounding scheme for the LP-relaxation of \mwc and showed that it achieves an approximation factor of $1.32388$. The analysis was subsequently improved by Sharma and Vondr\'{a}k \cite{SV14} to show that it achieves a $1.309017$-approximation for \mwc. Our main contribution is analyzing the approximation factor of the same algorithm for the objective of \gcovmp and bounding it by $1.125$. 

Next, we turn to the hardness of \gcovmp. It is known that \mwc is APX-hard via reduction from the max-cut problem \cite{cmplx-mwc}. We adapt the same reduction to conclude the following result. 

\begin{restatable}{theorem}{thmGCovMPInapprox}\label{thm:gcovmp-inapproximability}
%\gcovmp is APX-hard. Moreover, for 
For 
every constant $\epsilon>0$, \gcovmp does not admit a $(1.00036-\epsilon)$-approximation assuming $P\neq NP$ and does not admit a $(1.00074-\epsilon)$-approximation assuming the unique games conjecture. 
%\knote{Edit the APX-hardness section to pull out these factors based on the inapproximability factor of max-cut.}
\end{restatable}

Our lower and upper bounds for \gcovmp shows a separation of this problem from \mwc in terms of approximability although they are equivalent in terms of exact optimization. 
%We mention certain results for \mwc that do not seem to extend to \gcovmp. 
\iffalse
C\u{a}linescu, Karloff, and Rabani \cite{CKR00} formulated a linear programming relaxation for \mwc that is referred to in the literature as the CKR-relaxation. A similar LP-relaxation can also be formulated for \gcovmp (it has the same constraints but different objective). Firstly, the integrality gap of the CKR-relaxation for $k=3$ terminals has been determined exactly for \mwc \cite{kkrsty, CCT06}. However, we are unable to determine the integrality gap of the similar LP-relaxation for $k=3$ terminals for \gcovmp. This is because 
it is unclear if the worst-case integrality gap instance for the LP-relaxation of \gcovmp will be \emph{axis-aligned} (as is the case for the CKR-relaxation). 
%we are unable to prove the \emph{axis-aligned property} of the worst-case gap instance of the CKR-relaxation. 
%which was the key to reducing the problem to an axis-aligned simplex partitioning problem and reducing the search space. 
Secondly, the inapproximability factor of \mwc is known to coincide with the integrality gap of the CKR-relaxation \cite{MNRS08}. We are unable to show a similar result for \gcovmp. This is because we do not have a metric associated with edge costs in the coverage setting. The issue underlying our inability to extend both results seem related. Extending either of these results from \mwc to \gcovmp would be interesting and would reveal further similarities/differences between \gcovmp and \mwc. 
\fi
We conclude the section with a table of results showing the current best upper and lower bounds on the approximability of submodular partitioning problems that are of interest to this work (see Table \ref{table:results}). 
%We leave achieving either result mentioned above as direction for future research. 

%and leave it as a possible direction for future research. 
%A fundamental inapproximability result of \mwc under the unique games conjecture 

\iffalse

 \begin{table}[ht]
\centering
\begin{tabular}{|l | l | l | l |} 
 \hline
 \textbf{Problem} & \textbf{Upper Bound} & \textbf{Lower Bound} & \textbf{Integrality Gap} \\ [0.5ex] 
 \hline
 \submp & 2 \cite{CE} & 2 \cite{EVW} & 2 \cite{EVW} \\ 
 \hline 
 \symsubmp & 1.5 \cite{CE} & 1.268 \cite{EVW} & 1.268 \cite{EVW} \\
 \hline
 \monosubmp & 1.33333 (Thm. \ref{thm:mono-sub-mp-upper-bound}) & 1.1111 (Thm. \ref{thm:mono-sub-mp-lower-bound}) & 1.1111 (Thm. \ref{thm:mono-sub-mp-lower-bound} \& \cite{EVW})\\
 \hline
% \covmp & 1.33333 (Thm. \ref{thm:mono-sub-mp-upper-bound}) & ?? & 1.0714 (Thm. \ref{??})\\
% \hline
 \gcovmp & 1.125 (Thm. \ref{thm:gcovmp-approximation}) & ?? & ?? \\
 \hline 
 \mwc & 1.2965 \cite{SV14} & 1.20016 \cite{MNRS08, BCKM20} & 1.20016 \cite{BCKM20}\\
 \hline 
\end{tabular}
\caption{Upper bound and lower bound on the approximability of submodular partitioning problems and the integrality gap of the associated convex program. We have included known results for \submp, \symsubmp, and \mwc for the sake of comparison.}
\label{table:results}
\end{table}
\fi

\begin{table}[ht]
\centering
\begin{tabular}{|l | l | l |} 
 \hline
 \textbf{Problem} & \textbf{Lower Bound} & \textbf{Upper Bound}  \\ [0.5ex] 
 \hline
 \submp & 2 \cite{EVW} & 2 \cite{CE} \\ 
 \hline 
 \symsubmp & 1.268 \cite{EVW} & 1.5 \cite{CE} \\
 \hline
 \monosubmp & 1.1111 (Thm \ref{thm:mono-sub-mp-lower-bound}) & 1.33333 (Thm \ref{thm:mono-sub-mp-upper-bound}) \\
 \hline
% \covmp & 1.33333 (Thm \ref{thm:mono-sub-mp-upper-bound}) & ?? & 1.0714 (Thm \ref{??})\\
% \hline
 \gcovmp & 1.00074 (Thm \ref{thm:gcovmp-inapproximability}) & 1.125 (Thm \ref{thm:gcovmp-approximation}) \\
 \hline 
 \mwc & 1.20016 \cite{MNRS08, BCKM20} & 1.2965 \cite{SV14} \\
 \hline 
\end{tabular}
\caption{Upper bound and lower bound on the approximability of submodular partitioning problems. 
%and the integrality gap of the associated convex program. 
We have included known results for \submp, \symsubmp, and \mwc for the sake of comparison. The lower bounds for \submp, \symsubmp, and \monosubmp are assuming polynomial number of function evaluations. The lower bounds for \gcovmp and \mwc are assuming the unique games conjecture.}
\label{table:results}
\end{table}

%\knote{Lower bounds in the table - some are assuming $P\neq NP$ or UGC or polynomial oracle evaluation queries. Need to clarify.}
%\knote{We could drop the row for \covmp from the table.}

%% file: mono-sub-mp.tex
\section{Monotone Submodular Multiway Partition}
\label{sec:mono-sub-mp}
%In this section, address the approximability of \monosubmp by proving Theorems \ref{thm:mono-sub-mp-upper-bound} and \ref{thm:mono-sub-mp-lower-bound}. 
We recall that \monosubmp admits a $2$-approximation via the greedy splitting approach \cite{ZNI05} and via the convex relaxation approach \cite{CE}. There is a much simpler and faster $2$-approximation than either of these two approaches (see Proposition \ref{prop:monosubmp-easy-2-approx} below which leads to a linear-time algorithm). We mention this algorithm for the sake of completeness and to provide context. We subsequently improve the approximability to $4/3$ in Section \ref{sec:monosubmp-approx} (where we prove Theorem \ref{thm:mono-sub-mp-upper-bound}). We address the inapproximability in Section \ref{sec:monosubmp-inapprox} (where we prove Theorem \ref{thm:mono-sub-mp-lower-bound}). 

\begin{proposition}\label{prop:monosubmp-easy-2-approx}
Let $f:2^V\rightarrow \R_{\ge 0}$ be a monotone submodular function and $T:=\{t_1, t_2, \ldots, t_k\}\subseteq V$. Let $\opt$ be defined as
\[
\min\left\{\sum_{i=1}^k f(V_i): V_1, V_2, \ldots, V_k \text{ is a partition of } V \text{ such that } t_i \in V_i\ \forall\ i\in [k]\right\}. 
\]
Let $t_1, t_2, \ldots, t_k$ be an ordering of the terminals such that $f(\{t_1\})\le f(\{t_2\})\le \ldots \le f(\{t_k\})$. Then, 
\[
f\left(\left(V\setminus T\right)\cup\{t_k\}\right) + \sum_{i=1}^{k-1}f\left(\{t_i\}\right) 
\le \left(2-\frac{1}{k}\right)\opt. 
\]
\end{proposition}
\begin{proof}
    We have that 
    \begin{align}
        &f\left(\left(V\setminus T\right)\cup \{t_k\}\right) +\sum_{i=1}^{k-1}f\left(\{t_i\}\right)\notag\\
        &\quad \quad \le f(V) + (1-1/k)\sum_{i=1}^k f(\{t_i\})\notag\\
        &\quad \quad \le \left(2-\frac{1}{k}\right)\max\left\{f(V), \sum_{i=1}^k f\left(\{t_i\}\right)\right\}. \label{ineq:algo-upper-bound}
    \end{align} 
    Next, we lower bound the objective value of an optimum partition. 
    Let $V_1^*$, $V_2^*$, $\ldots$, $V_k^*$ be an optimum partition with $t_i\in V_i^*$ for every $i\in [k]$. Then, (1) $\sum_{i=1}^k f(V_i^*)\ge f(V)$ by non-negativity and submodularity of $f$ and (2) $\sum_{i=1}^k f(V_i^*)\ge \sum_{i=1}^k f(\{t_i\})$ since $f(V_i^*)\ge f(\{t_i\})$ by monotonicity of $f$ and the fact that $t_i\in V_i^*$ for every $i\in [k]$. Consequently, 
    \begin{align}
        \opt&\ge \max\left\{f(V), \sum_{i=1}^k f\left(\{t_i\}\right)\right\}. \label{ineq:opt-lower-bound}
    \end{align} 
    The proposition follows from inequalities \eqref{ineq:algo-upper-bound} and \eqref{ineq:opt-lower-bound}. 
\end{proof}
\begin{remark}
    We note that the factor of $(2-1/k)$ in Proposition \ref{prop:monosubmp-easy-2-approx} is tight: let $f$ be the rank function of the partition matroid on a $k$-partition $S_1, S_2, \ldots, S_k$ of the ground set where $|S_i|\ge 2$ for every $i\in [k]$. Suppose that the set $T=\{t_1,t_2, \ldots, t_k\}$ of terminals is such that $t_i\in S_i$ for every $i\in [k]$. Then, the LHS of the lemma is 
    %$\{t_1\}, \{t_2\}, \ldots, \{t_{k-1}\}, (V\setminus T)\cup \{t_k\}$ 
    is $2k-1$, whereas the partition $S_1, \ldots, S_k$ has objective value $k$.
\end{remark}

\input{mono-algorithm}
\input{mono-approx}

\input{mono-lower-bound}

%% file: mono-algorithm.tex
\subsection{Approximation Algorithm}\label{sec:monosubmp-approx}
In this section, we describe our algorithm for \monosubmp and analyze its approximation factor. Let $f: 2^V\rightarrow \R_{\ge 0}$ be the input monotone submodular function with terminals $t_1, \ldots, t_k\in V$. 
For $\vx \in [0, 1]^V$ and $\theta \in [0, 1]$, we define $\vx^{\theta}\in \{0,1\}^V$ as 
\[
(\vx^\theta)_v := 
\begin{cases} 
1 & \text{ if } \vx_v \geq \theta, \\ 
0 & \text{ otherwise. }  
\end{cases}
\]
%\rnote{I don't think we ever define $V = \{v_1, \ldots, v_n\}$}
We also define $f\left(\vx^\theta\right) := f\left(\{v\in V \colon (\vx^\theta)_v = 1\}\right)$. We define the Lov\'{a}sz extension $\hat{f}:[0,1]^V\rightarrow \R$ of $f$ as follows:  For all $\vx \in [0, 1]^V$, $\hat{f}(\vx) := \E_{\theta \in [0, 1]}\left[f(\vx^\theta)\right] = \int_0^1 f\left(\vx^\theta\right) d\theta$. We define $\vx_i = \left(x(v_1, i), \ldots, x(v_n, i)\right)$ and consider the convex program %\submprel 
below:
\begin{align}
    \text{minimize } \sum_{i = 1}^k \hat{f}\left(\vx_i\right) & \tag{Sub-MP-Rel} \label{CP:submp} \\
    %\text{subject to: } 
    \sum_{i = 1}^k x(v, i) = 1 & \quad \forall v \in V \notag\\
    x(t_i, i) = 1 & \quad \forall i \in [k] \notag\\
    x(v, i) \geq 0 & \quad \forall v \in V, i \in [k]. \notag
\end{align}
The convex program is identical to the convex programming relaxation of \submp 
introduced by Chekuri and Ene and it can be solved in polynomial time \cite{CE}. 
%It is a relaxation for \submp and 
%The convex program can be solved can be solved in polynomial time. 
We present a rounding algorithm for \monosubmp based on this relaxation. Our rounding algorithm is similar to the $\theta$-rounding algorithms of Chekuri and Ene for \submp and \symsubmp \cite{CE}, but the interval from which we choose $\theta$ is different. 
%Using the same interval as that of Chekuri and Ene for \submp or \symsubmp does not achieve a $4/3$-approximation (see Remark \ref{??}).  
Our choice for the interval of $\theta$ is designed to exploit monotonicity in order to achieve a better approximation factor. 
Below, we present the rounding algorithm for \monosubmp:
\begin{algorithm}[H] % H causes weird spacing right now
\caption{Mono-Sub-MP Rounding}\label{alg:mono}
\begin{algorithmic}
    \State Let $\vx$ be an optimal solution to \ref{CP:submp} %SUBMP-REL
    \State Choose $\theta \in [\frac{1}{4}, 1]$ uniformly at random
    \For{$i \in [k]$}
        \State $A(i, \theta) \gets \{v \in V \colon x(v, i) \geq \theta \}$
    \EndFor
    \State $U(\theta) \gets V - \bigcup_{i = 1}^k A(i, \theta)$
    \For{$v \in V$}
        \State $I_v \gets \{i \in [k] \colon v \in A(i, \theta) \}$
        \If{$|I_v| > 1$}
            \State delete $v$ from all but one of the sets $\{A(i, \theta)\}_{i \in I_v}$
        \EndIf
    \EndFor
    \State \Return $\left(A_1, \ldots, A_{k - 1}, A_k \right) = \left(A(1, \theta), \ldots, A(k - 1, \theta), A(k, \theta) \cup U(\theta) \right)$
\end{algorithmic}
\end{algorithm}

We observe that the rounding algorithm guarantees that each $v \in V$ is assigned to exactly one of the returned sets. Additionally, for each $i \in [k]$, we have that $t_i \in A(i, \theta)$ since $\vx(t_i,i)=1$. Hence, the algorithm returns a feasible multiway partition of $V$. Moreover, the algorithm can be implemented to run in polynomial time. %We now bound the approximation factor. 

%% file: mono-approx.tex
%\subsubsection{Approximation Factor Analysis}
%The following is the main theorem of this section. It bounds 
We now bound the approximation factor of Algorithm \ref{alg:mono} relative to the optimum value of \ref{CP:submp}. 

\begin{restatable}{theorem}{thmMonoSubMPRelCP} \label{thm:mono-sumbp-4/3}
 Let $\opt_{frac}$ be the optimum value of \ref{CP:submp} and $A_1, \ldots, A_k$ be the solution returned by Algorithm \ref{alg:mono}. Then, 
\[
    \E\left[\sum_{i=1}^k f(A_i)\right]\le \frac{4}{3}\opt_{frac}. 
\]

%Algorithm \ref{alg:mono} achieves an approximation factor of $4/3$ for \monosubmp.
\end{restatable}

The rest of this section is devoted to proving Theorem \ref{thm:mono-sumbp-4/3}. 
Our approximation factor analysis uses certain ingredients of Chekuri and Ene \cite{CE}, but differs from their analysis to exploit monotonicity and the different threshold range. We will point out the differences as we proceed with the analysis. 

We start with a lemma about the truncated mass of a discrete probability distribution. A vector $\vp\in \R^k_{\ge 0}$ with $\sum_{i\in [k]}p_i=1$ can be viewed as a probability distribution over the set $\{1,2,\ldots, k\}$ with $p_i$ being the probability of picking item $i$. For $\delta\in [0,1]$, the quantity $\sum_{i=1}^k \min\{\delta, p_i\}$ is the truncated probability mass. We show that if $\delta\ge 1/4$, then the truncated probability mass is at least $1-\max_{i\in [k]}p_i$. The choice of $\delta\in [1/4,1]$ in the rounding algorthm emanates from this lemma. This lemma does not appear in Chekuri and Ene \cite{CE} and is our contribution. 

\begin{lemma} \label{claim:1/4-ineq}
    Let $\vp \in \mathbb{R}^k_{\ge 0}$ such that $\sum_{i = 1}^k p_i = 1$ and $\delta \ge 1/4$. Then, 
    \[\sum_{i = 1}^k \min\left\{\delta, p_i\right\} \geq 1 - \max_{i \in [k]} p_i.\]
\end{lemma}
\begin{proof}
    % Proof of claim 3.1
    The LHS can be rewritten as follows: 
    \begin{align*}
        \sum_{i = 1}^{k} \min\left\{\delta, p_i\right\} &= \sum_{i = 1}^{k} p_i - \sum_{i \in [k]: p_i > \delta} (p_i - \delta)  
        = 1 - \sum_{i \in [k] : p_i > \delta} (p_i - \delta)
    \end{align*}
    % Hence, it suffices to show that $\sum_{i\in [k]: p_i>\delta}(p_i-\delta) \le \max_{i\in [k]}p_i$. 
    \iffalse
    \begin{align*}
        \sum_{i = 1}^k \min(\delta, p_i) \geq 1 - \max_{i = 1}^k p_i 
        &\iff 1 - \sum_{i : p_i > \delta} (p_i - \delta) \geq 1 - \max_{i = 1}^k p_i \\ 
        &\iff \sum_{i : p_i > \delta} (p_i - \delta) \leq \max_{i = 1}^k p_i
    \end{align*}
    \fi
    % \knote{
    Let $L:=\{i\in [k]: p_i>\delta\}$ and let $j\in [k]$ such that $p_j=\max_{i\in [k]}p_i$. 
    It suffices to show that $p_j \ge \sum_{i\in L}(p_i-\delta)$. 
    For the sake of contradiction, suppose that $p_j <\sum_{i\in L}(p_i-\delta)$. We case based on the size of $L-\{j\}$ and arrive at a contradiction in all cases.
    
    \noindent \textbf{Case 1.} Suppose $L-\{j\}=\emptyset$. If $j\in L$, then $L=\{j\}$, and hence, 
    \[
    p_j < \sum_{i\in L}(p_i-\delta) = p_j - \delta,
    \]
    a contradiction. If $j\not\in L$, then $L=\emptyset$, and hence,
    \[
    0\le p_j < \sum_{i\in L}(p_i-\delta) = 0,
    \]
    a contradiction. 
    
    \noindent \textbf{Case 2.} Suppose $|L-\{j\}|=1$. Let $L-\{j\}=\{\ell\}$. Then, by definition of $j, \ell,$ and $L$, we have that $L=\{\ell, j\}$. Consequently, we have that 
    \[
    p_j < \sum_{i\in L}(p_i-\delta) = p_j + p_{\ell} - 2\delta.
    \]
    Hence, we have that $p_{\ell}> 2\delta \ge 1/2$ since $\delta \ge 1/4$. Also, by definition of $j$, we have that $p_j\ge p_{\ell}$. Hence, we have that $p_j\ge p_{\ell}>1/2$ and hence, $\sum_{i=1}^k p_i \ge p_j + p_k >1$, a contradiction.

    \noindent \textbf{Case 3.} Suppose $|L-\{j\}|\ge 2$. 
    Then, by definition of $j$, we have that $j\in L$ and hence, $|L|\ge 3$ and $p_j> \delta$. Thus, we have that 
    \[
    \delta \le p_j <\sum_{i\in L}(p_i-\delta) \le \left(\sum_{i\in L}p_i\right) - 3\delta.
    \]
    The last inequality above is because $|L|\ge 3$. From the above inequalities, we have that $\sum_{i=1}^k p_i \ge \sum_{i\in L}p_i > 4\delta \ge 1$, a contradiction. 
\end{proof}

\begin{remark}\label{remark:tight-for-1/4-ineq}
    The conclusion of Lemma \ref{claim:1/4-ineq} fails for $\delta <1/4$. We give an example to illustrate this: let $k=2$ and let $x_1=1/2, x_2 = 1/2$. Now, consider $\delta <1/4$. Then, the LHS of the lemma is $\sum_{i=1}^k \min\{\delta, x_i\}=\min\{\delta, x_1\} + \min\{\delta, x_2\} = 2\delta$. On the other hand, the RHS of the lemma is $1-\max_{i\in [k]}x_i = 1- \max\{x_1, x_2\}=1/2$. Thus, the LHS is at least the RHS only for $\delta\ge 1/4$ and fails for $\delta<1/4$. 
    %$2\delta \ge 1/2$ fails to hold if $\delta < 1/4$. 
\end{remark}

We now set up the notation to analyze the approximation factor of our algorithm. This notation is similar to that of Chekuri and Ene \cite{CE}. 
\paragraph{Notation.}
Let $\vx$ be an optimum solution to \submprel. 
We assume that $V=\{v_1, \ldots, v_n\}$. 
For each $v_j\in V$, we define $\alpha_j = \max_{i = 1}^k x(v_j, i)$. We assume that the elements $v_1, v_2, \ldots, v_n$ of the ground set $V$ are ordered such that $0 \leq \alpha_1 \leq \alpha_2 \leq \ldots \leq \alpha_n \leq 1$. For convenience, we define $\alpha_0 := 0$ and $\alpha_{n + 1} := 1$. 
%For each $j \in [n]$, we let $\ell_j := \argmax_{i = 1}^k x(v_j, i)$. 
Let $\theta\in [0,1]$. 
For each $i \in [k]$, we define $A(i, \theta) = \{v \in V \colon x(v, i) \geq \theta\}$ as in the algorithm. Additionally, for each $j \in [n]$, we let $V_j := \{v_1, \ldots, v_j\}$, and for convenience, we define $V_0 := \emptyset$. Finally, we let $A_j(i, \theta) := A(i, \theta) \cap V_j$.

With the above notation, we show the following technical lemma based on monotonicity and submodularity of $f$ and by using Lemma \ref{claim:1/4-ineq}. Since we rely on monotonicity and Lemma \ref{claim:1/4-ineq}, the following lemma does not follow from the results in Chekuri and Ene \cite{CE}. 
\begin{lemma} \label{lem:telescope-sum}
    For every $\delta \in [\frac{1}{4}, 1]$ and $j \in [n]$, we have that
    \[
        \sum_{i = 1}^k \int_0^\delta \left( f\left(A_j(i, \theta)\right) - f\left(A_{j - 1}(i, \theta)\right)\right) d\theta \geq (1 - \alpha_j) \left(f(V_j) - f(V_{j - 1}) \right).
    \]
    % \knote{Remember to punctuate math modes. Put a period at the end of the above math mode.}
\end{lemma}
\begin{proof}
    % Let us fix $j \in [n]$ and $i \in [k]$ 
    Let $\delta\in [1/4, 1]$ and $j\in [n]$. We show the following two claims that will help in proving the lemma.
    \begin{claim} \label{claim:telescope-trunc}
    For every $i\in [k]$,
        \begin{align*}
            &\int_0^\delta \left( f\left(A_j(i, \theta)\right) - f\left(A_{j - 1}(i, \theta)\right)\right) d\theta \\
        &\quad \quad \quad \quad= \int_0^{\min\{\delta, x(v_j, i)\}} \left( f\left(A_j(i, \theta)\right) - f\left(A_{j - 1}(i, \theta)\right)\right) d\theta.
        \end{align*}
    \end{claim}
    \begin{proof}
        Let $i\in [k]$.
        We may assume that $\min\{\delta, \ x(v_j, i)\}=x(v_j, i)$ (otherwise, the claim holds trivially). 
        % Then, for all $\theta \in (x(v_j, i), \delta]$, we get $x(v_j, i) < \theta$. So, $v_j \notin A(i, \theta)$. Thus, $A_j(i, \theta) = A(i, \theta) \cap V_j = A(i, \theta) \cap V_{j-1} = A_{j-1}(i, \theta)$. 
        For every $\theta \in (x(v_j, i), \delta]$, we get $x(v_j, i) < \theta$ and hence, $v_j \notin A(i, \theta)$. Thus, for every $\theta \in (x(v_j, i), \delta]$, we have that $A_j(i, \theta) = A(i, \theta) \cap V_j = A(i, \theta) \cap V_{j-1} = A_{j-1}(i, \theta)$. Consequently, 
        % Hence, 
        \begin{align*}
            &\int_0^\delta \left( f\left(A_j(i, \theta)\right) - f\left(A_{j - 1}(i, \theta)\right)\right) d\theta\\
            & \quad \quad \quad \quad = \int_0^{x(v_j, i)} \left( f\left(A_j(i, \theta)\right) - f\left(A_{j - 1}(i, \theta)\right)\right) d\theta\\ 
            & \quad \quad \quad \quad \quad \quad + \int_{x(v_j, i)}^\delta \left( f\left(A_j(i, \theta)\right) - f\left(A_{j - 1}(i, \theta)\right)\right) d\theta \\ 
            &\quad \quad \quad \quad = \int_0^{x(v_j, i)} \left( f\left(A_j(i, \theta)\right) - f\left(A_{j - 1}(i, \theta)\right)\right) d\theta.
        \end{align*}
    \end{proof}
    \begin{claim} \label{claim:sub-telescope}
        For every $i\in [k]$, 
        \begin{align*}   
        &\int_0^{\min\{\delta, x(v_j, i)\}} \left( f\left(A_j(i, \theta)\right) - f\left(A_{j - 1}(i, \theta)\right)\right) d\theta \\
        & \quad \quad \quad \quad  \quad \quad \quad \quad \geq
        \min\{\delta, x(v_j, i)\} \left( f\left(V_j\right) - f\left(V_{j-1}\right)\right). 
        \end{align*}
    \end{claim}
    \begin{proof}
        Let $\theta \in [0, \min(\delta, x(v_j, i))]$. Then, $x(v_j, i) \geq \theta$ and thus, $A_j(i, \theta) = A_{j-1}(i, \theta) \cup \{v_j\}$. Additionally, $V_j = V_{j-1} \cup \{v_j\}$, $A_{j-1}(i, \theta) \subseteq V_{j-1}$, and $v_j \notin V_{j-1}$. Hence, we have that 
        \begin{align*}
            f\left(A_j(i, \theta)\right) - f\left(A_{j - 1}(i, \theta)\right) &= f\left(A_{j-1}(i, \theta) \cup \{v_j\}\right) - f\left(A_{j - 1}(i, \theta)\right) \\
            &\geq f\left(V_{j-1} \cup \{v_j\}\right) - f\left(V_{j-1}\right)\\
            & \quad \quad \quad \quad \quad \quad \text{(by submodularity of $f$)}\\
            &= f\left(V_j\right) - f\left(V_{j-1}\right). 
        \end{align*}
        Thus, 
        \begin{align*}
            &\int_0^{\min\{\delta, x(v_j, i)\}} \left( f\left(A_j(i, \theta)\right) - f\left(A_{j - 1}(i, \theta)\right)\right) d\theta \\
            &\quad \quad \geq
            \int_0^{\min\{\delta, x(v_j, i)\}} \left( f\left(V_j\right) - f\left(V_{j-1}\right)\right) d\theta\\
            &\quad \quad = \min\{\delta, x(v_j, i)\} \left( f\left(V_j\right) - f\left(V_{j-1}\right)\right).
        \end{align*}
        % \knote{Punctuate math mode above.}
    \end{proof}
    Now, we use the above claims to lower bound the LHS of the lemma 
    %Lemma \ref{lem:telescope-sum} 
    as follows:  
    \begin{align}
        &\sum_{i = 1}^k \int_0^\delta \left( f\left(A_j(i, \theta)\right) - f\left(A_{j - 1}(i, \theta)\right)\right) d\theta \notag\\
        &\quad \quad \quad \quad =
        \sum_{i = 1}^k \int_0^{\min\{\delta, x(v_j, i)\}} \left( f\left(A_j(i, \theta)\right) - f\left(A_{j - 1}(i, \theta)\right)\right) d\theta \quad \text{(by Claim \ref{claim:telescope-trunc})} \notag\\ 
        &\quad \quad \quad \quad \geq
        %\sum_{i = 1}^k \int_0^{\min(\delta, x(v_j, i))} \left( f\left(V_j\right) - f\left(V_{j-1}\right)\right) d\theta \\ &= 
        \sum_{i = 1}^k \min\{\delta, x(v_j, i)\} \left( f\left(V_j\right) - f\left(V_{j-1}\right)\right) \quad \text{(by Claim \ref{claim:sub-telescope})} \notag\\ 
        &\quad \quad \quad \quad = \left( f\left(V_j\right) - f\left(V_{j-1}\right)\right) \left( \sum_{i = 1}^k \min\{\delta, x(v_j, i)\} \right). \label{ineq:fA_j-difference}
    \end{align}
    Finally, we know that $x(v_j, i)\ge 0$ for every $i \in [k]$ and $\sum_{i=1}^{k} x(v_j, i) = 1$ since $x$ is a feasible solution to the convex program. Hence, using Lemma \ref{claim:1/4-ineq}, we get
    \begin{align}
    \sum_{i=1}^k \min\{\delta, x(v_j, i)\} &\ge 1- \max_{i\in [k]}x(v_j, i) = 1-\alpha_j. \label{ineq:truncated-mass}
    \end{align}
    Further, by monotonicity of $f$, we have that $f(V_j) - f(V_{j-1})\ge 0$. Hence, 
    \begin{align}
    \left(f\left(V_j\right) - f\left(V_{j-1}\right)\right) \left( \sum_{i = 1}^k \min\{\delta, x(v_j, i)\} \right) 
    \geq (1 - \alpha_j)  \left( f\left(V_j\right) - f\left(V_{j-1}\right)\right). \label{ineq:fV_j-difference}
    \end{align}
    
    % \begin{align*}
    %     \left( f\left(V_j\right) - f\left(V_{j-1}\right)\right) \left( \sum_{i = 1}^k \min(\delta, x(v_j, i)) \right) &\geq 
    %     \left( f\left(V_j\right) - f\left(V_{j-1}\right)\right) \left( 1 - \max_{i=1}^{k} x(v_j, i) \right) \\
    %     &= (1 - \alpha_j)  \left( f\left(V_j\right) - f\left(V_{j-1}\right)\right)
    % \end{align*}
    Thus, using inequalities (\ref{ineq:fA_j-difference}) and (\ref{ineq:fV_j-difference}), we get that 
    \begin{align*}
        \sum_{i = 1}^k \int_0^\delta \left( f\left(A_j(i, \theta)\right) - f\left(A_{j - 1}(i, \theta)\right)\right) d\theta \geq (1 - \alpha_j) \left(f(V_j) - f(V_{j - 1}) \right).
    \end{align*}
    % This finishes the proof of lemma \ref{lem:telescope-sum}. \knote{Remove last sentence. The QED symbol takes care of the sentence.} \knote{Punctuate math mode throughout the above proof. Add all periods. }
\end{proof}

We need the following Proposition from Chekuri and Ene \cite{CE}.
\begin{proposition} \cite{CE} \label{prop:U-int-telescope}
    Let $f$ be an arbitrary non-negative submodular function, $r \in [0,1]$, and let $h$ be the largest value of $j\in [n]$ such that $\alpha_j \leq r$. We have
    \iffalse
    \begin{align*}
        \int_0^r f\left(A(\theta)\right) d\theta = \sum_{j=1}^{h} \alpha_j \left(f(V - V_{j-1}) - f(V - V_j)\right) + rf(V - V_h),
    \end{align*}
    and 
    \fi
    \begin{align*}
        \int_0^r f\left(U(\theta)\right) d\theta = \sum_{j=1}^{h} \alpha_j \left(f(V_{j-1}) - f(V_j)\right) + rf(V_h).
    \end{align*}
\end{proposition}

Lemma \ref{lem:telescope-sum} and Proposition \ref{prop:U-int-telescope} lead to the following corollary. 
\begin{corollary}
\label{cor:at_least_U}
    For every $\delta \in [\frac{1}{4}, 1]$, 
    \[
        \sum_{j = 1}^n \left( \sum_{i = 1}^k \int_0^\delta \left(f(A_j(i, \theta) - f(A_{j - 1}(i, \theta) \right) d\theta \right) \geq \int_0^1 f(U(\theta)) d\theta - f(\emptyset).
    \]
\end{corollary}
\begin{proof}
    We use $r = 1$ and $h = n$ in Proposition \ref{prop:U-int-telescope} to get, 
    \begin{align} \label{eq:u-theta}
        \int_0^1 f\left(U(\theta)\right) d\theta = \sum_{j=1}^{n} \alpha_j \left(f(V_{j-1}) - f(V_j)\right) + f(V_n).
    \end{align}
    Now, by Lemma \ref{lem:telescope-sum}, we get 
    \begin{align*}
        &\sum_{j=1}^{n} \left(\sum_{i = 1}^k \int_0^\delta \left( f\left(A_j(i, \theta)\right) - f\left(A_{j - 1}(i, \theta)\right)\right) d\theta \right) \\
        &\quad \quad \geq 
        \sum_{j=1}^{n} (1 - \alpha_j) \left(f(V_j) - f(V_{j - 1}) \right) \\
        &\quad \quad = \sum_{j=1}^{n} \left(f(V_j) - f(V_{j - 1}) \right) + 
         \sum_{j=1}^{n} \alpha_j\left(f(V_{j-1}) - f(V_{j}) \right) \\
        &\quad \quad = f(V_n) - f(V_0) + \sum_{j=1}^{n} \alpha_j\left(f(V_{j-1}) - f(V_{j}) \right). 
    \end{align*}
    Finally, using $V_0 = \emptyset$ and equation (\ref{eq:u-theta}), we get, 
    \begin{align*}
        \sum_{j=1}^{n} \left(\sum_{i = 1}^k \int_0^\delta \left( f\left(A_j(i, \theta)\right) - f\left(A_{j - 1}(i, \theta)\right)\right) d\theta \right) 
        &\geq
        \int_0^1 f(U(\theta)) d\theta - f(\emptyset).
    \end{align*}
\end{proof}

We have the following theorem as a consequence of Corollary \ref{cor:at_least_U}. We observe that this theorem upper bounds the expected function value of the set of unassigned elements. 
\begin{theorem}
\label{thm:sum_of_A_at_least_U}
    For all $\delta \in [\frac{1}{4}, 1]$, 
    \[
        \sum_{i = 1}^k \int_0^\delta f(A(i, \theta)) d\theta \geq (k \delta - 1)f(\emptyset) + \int_0^1 f(U(\theta)) d\theta.
    \]
\end{theorem}
\begin{proof}
We have that 
    \begin{align*}
        &\sum_{i = 1}^k \int_0^\delta f(A(i, \theta)) d\theta \\
        & \quad \quad = \sum_{i = 1}^k \int_0^\delta f(A_n(i, \theta)) d\theta \\
        &\quad \quad = \sum_{i = 1}^k \int_0^\delta f(A_0(i, \theta)) d\theta + \sum_{j = 1}^n \left( \sum_{i = 1}^k \int_0^\delta f(A_j(i, \theta)) - \sum_{i = 1}^k \int_0^\delta f(A_{j - 1}(i, \theta)) \right) \\
        &\quad \quad = \sum_{i = 1}^k \int_0^\delta f(A_0(i, \theta)) d\theta + \sum_{j = 1}^n \left( \sum_{i = 1}^k \int_0^\delta \left(f(A_j(i, \theta)) - f(A_{j - 1}(i, \theta))\right) d\theta \right) \\
        &\quad \quad \geq \sum_{i = 1}^k \int_0^\delta f(A_0(i, \theta))d\theta + \int_0^1 f(U(\theta)) d\theta - f(\emptyset) \quad \quad \quad \text{(by Corollary \ref{cor:at_least_U})} \\
        &\quad \quad = (k\delta - 1)f(\emptyset) + \int_0^1 f(U(\theta)) d\theta.
    \end{align*}
\end{proof}

\begin{remark}\label{remark:tight-for-assigned-vs-unassigned-charging}
    Theorem \ref{thm:sum_of_A_at_least_U} fails to hold for $\delta<1/4$. We give an example to illustrate this: let $V:=\{t_1, t_2, u\}$, $T:=\{t_1, t_2\}$ (i.e., $n=3$ and $k=2$) and let $f(X):=1$ for every non-empty $X$ and $f(\emptyset)=0$. Setting $x(u, 1)=x(u, 2)=1/2$, $x(t_1, 1) = 1$, $x(t_2, 2)=1$ with the rest of the $x$ values being $0$ gives an optimum solution to the convex program \submprel. Now consider $\delta \le 1/4$. Then, we have that 
    $f(A(1, \theta))=1= f(A(2, \theta))$ for $\theta \in [0, 1/2]$, and hence, the LHS of the theorem is 
    $\sum_{i=1}^k \int_{0}^{\delta}f(A(i, \theta))d\theta = \int_{0}^{\delta} f(A(1, \theta))d\theta + \int_{0}^{\delta}f(A(1, \theta))d\theta= \delta + \delta = 2\delta$. On the other hand, we have that $f(U(\theta))=0$ for every $\theta \in [0,1/2)$ and $f(U(\theta))=1$ for every $\theta \in [1/2,1]$ and   hence, the RHS of the theorem is 
    $\int_{0}^1 f(U(\theta))d\theta = \int_{1/2}^1 f(U(\theta))d\theta = 1/2$. Thus, $2\delta \ge 1/2$ fails to hold if $\delta <1/4$. We note that this tight example for Theorem \ref{thm:sum_of_A_at_least_U} arises from the tight example for Lemma \ref{claim:1/4-ineq} mentioned in Remark \ref{remark:tight-for-1/4-ineq}. 
\end{remark}

We use Theorem \ref{thm:sum_of_A_at_least_U} and submodularity and monotonicity of $f$ to bound the approximation factor of the algorithm relative to the optimum value of the convex program. We restate and prove Theorem \ref{thm:mono-sumbp-4/3} below. 
\thmMonoSubMPRelCP*
\begin{proof}
    %We show that the algorithm Mono-Sub-MP Rounding gives a $\frac{4}{3}$-approximation for Mono-Sub-MP. 
    We note that $\opt_{frac} = \sum_{i = 1}^k \hat{f}(\vx_i) = \int_0^1 \sum_{i = 1}^k f(A(i, \theta)) d\theta$. By Theorem \ref{thm:sum_of_A_at_least_U}, 
    \begin{align}
        \sum_{i = 1}^k \int_0^\frac{1}{4} f(A(i, \theta)) d\theta \geq \left(\frac{1}{4}k - 1\right)f(\emptyset) + \int_0^1 f(U(\theta)) d\theta \geq \int_0^1 f(U(\theta)) d\theta.  \label{ineq:minus_A_plus_U}
    \end{align}
    The last inequality above is because $f(\emptyset)=0$. 
    %Rewriting the above inequality gives 
    %\begin{align}
    %    0 &\geq -\sum_{i = 1}^k \int_0^\frac{1}{4} f(A(i, \theta)) d\theta + \int_0^1 f(U(\theta)) d\theta \geq -\sum_{i = 1}^k \int_0^\frac{1}{4} f(A(i, \theta)) d\theta + \int_{\frac{1}{4}}^1f(U(\theta)) d\theta. \label{ineq:minus_A_plus_U}
    %\end{align}
Hence, the expected cost of the partition returned by Algorithm \ref{alg:mono} is
$\E_{\theta \in [\frac{1}{4}, 1]}\left[ \sum_{i = 1}^k f(A_i) \right]$
    \begin{align*}
        &\leq \E_{\theta \in [\frac{1}{4}, 1]} \left[ \sum_{i = 1}^{k - 1} f(A(i, \theta)) + f(A(k, \theta) \cup U(\theta))\right] \quad \quad \text{(since $f$ is monotone)} \\
        &\leq \E_{\theta \in [\frac{1}{4}, 1]}\left[ \sum_{i = 1}^k f(A(i, \theta)) + f(U(\theta)) \right] \text{(since $f$ is non-negative and submodular)} \\
        &= \frac{4}{3} \int_{\frac{1}{4}}^1 \left(\sum_{i = 1}^k f(A(i, \theta)) + f(U(\theta))\right) d\theta \\
        &= \frac{4}{3} \left(\opt_{frac} - \sum_{i = 1}^k \int_{0}^{\frac{1}{4}} f(A(i, \theta)) d\theta + \int_{\frac{1}{4}}^1 f(U(\theta)) d\theta \right) \\
        &\leq \frac{4}{3}\opt_{frac}. \quad \quad \text{(by inequality (\ref{ineq:minus_A_plus_U}))}
    \end{align*}
\end{proof}

\begin{remark}
    We emphasize that monotonicity of the function is used only in the proofs of Theorem \ref{thm:mono-sumbp-4/3} and Lemma \ref{lem:telescope-sum}. 
\end{remark}
\begin{remark}
    Our approximation factor analysis for Algorithm \ref{alg:mono} is tight. We give an example to show that Algorithm \ref{alg:mono} does not achieve $(4/3-\epsilon)$-approximation for any constant $\epsilon>0$. Let $V:= \{t_1, t_2, u_1 \cdots, t_{2k-1}, t_{2k}, u_k\}$ and $T:= \{t_1, t_2, \cdots, t_{2k}\}$. Let $E_i := \{t_{2i-1}, t_{2i}, u_i\}$ and $f(X) := |\{j : X \cap E_j \neq \emptyset\}|$.
    The partition $V_{2i} = \{t_{2i}, u_{i}\}$, $V_{2i-1} = \{t_{2i-1}\}$ is optimal for this instance of \monosubmp, with $\sum_{i=1}^{2k} f(V_i) = \sum_{i=1}^{k} \left(f(V_{2i}) + f(V_{2i-1})\right) = 2k$. 
    Now, setting $x(u_i, i) = x(u_i, i+1) = \frac{1}{2}$ with the rest of the $x$ values being $0$ gives an optimum solution to the convex program \submprel. For $\theta \in (\frac{1}{2}, 1]$, $A(i, \theta) = \{t_i\}$, $U(\theta) = \{u_1, \cdots, u_k\}$. Hence, the algorithm returns $V_i = \{t_i\}$ for $i \in [2k-1]$, $V_{2k} = \{u_1, \cdots, u_k, t_{2k}\}$. This gives $\sum_{i=1}^{2k} f(V_i) = 3k-1$. Further, for $\theta \in [\frac{1}{4}, \frac{1}{2}]$, $A(2i, \theta) = \{t_{2i}, u_i\}$ and $A(2i-1) = \{t_{2i-1}, u_i\}$. Hence, the algorithm returns $V_{2i} = \{t_i, u_i\}$, $V_{2i-1} = \{t_{2i-1}\}$, and $\sum_{i=1}^{2k} f(V_i) = 2k$. Hence, $\E_{\theta \in [\frac{1}{4}, 1]}[\sum_{i=1}^{2k} f(V_i)] = \frac{1}{3}(2k) + \frac{2}{3}(3k-1) = \frac{8k-2}{3}$. Hence, approximation factor for this example is $\frac{8k-2}{6k} = \frac{4}{3} - \frac{1}{3k}$. Setting $k$ large enough gives us the desired result.
    We mention that this example is obtained by taking multiple copies of the tight example for Theorem \ref{thm:sum_of_A_at_least_U} mentioned in Remark \ref{remark:tight-for-assigned-vs-unassigned-charging}.

\end{remark}

\begin{remark}
For threshold rounding where the threshold $\theta$ is picked uniformly at random from the interval $[1/2, 1]$ (i.e., Chekuri and Ene's rounding algorithm for \submp), the best possible approximation factor is at least $3/2$. The example in the previous remark illustrates this. 
%Moreover, for threshold rounding where the threshold $\theta$ is picked uniformly at random from the interval $[0, 1]$ and the sets are subsequently uncrossed via monotonicity (i.e., akin to Chekuri and Ene's rounding algorithm for \symsubmp), the best possible approximation factor is at least $2$: 
\end{remark}

%% file: mono-lower-bound.tex
\subsection{Inapproximability}\label{sec:monosubmp-inapprox}
%\donote{State the lower bound result. Prove it in subsections below.}
In this section, we prove Theorem \ref{thm:mono-sub-mp-lower-bound}. 
%\thmMonoSubMPLowerBound*
Our approach will be similar to the approaches in \cite{Vondrak, EVW}---we define a notion of \emph{symmetry gap} below; in 
%Section \ref{subsec:sym_gap_to_query-complexity}, 
Theorem \ref{thm:sym-gap-to-oracle-lower-bound}, 
we show that symmetry gap gives a lower bound on the approximation factor of algorithms with polynomial number of function evaluation queries; 
finally, in Theorem \ref{thm:large-symmetry-gap}, 
%and in Section \ref{subsec:large-symmetry-gap}, 
we construct an instance of \monosubmp that exhibits a symmetry gap of $10/9$. 
We mention that Dobzinksi and Vondr\'{a}k \cite{DV12} showed a general approach to convert query complexity based inapproximability results via the symmetry gap to computation based inapproximability results.

%We begin by defining symmetry gap. 
\begin{definition}(Symmetry Gap) 
\label{def:instance}
\iffalse
    Let $\phi_n, \phi_t \colon [k] \rightarrow \mathbb{R}$ be given functions. 
    We set $V = [k] \times [k]$, $R_i = \{(i, j) \colon j \in [k]\}$ and $C_j = \{(i, j) \colon i \in [k]\}$, and define 
    $f_{\phi_n, \phi_t} \colon 2^V \rightarrow \mathbb{R}$ as the function  
        \[
            f_{\phi_n, \phi_t}(S) = \sum_{i = 1}^k \left(\phi_n(S \cap R_i)\indicator{(i, i)\not\in S} + \phi_t(S \cap R_i)\indicator{(i, i)\in S} + \phi_n(S \cap C_i)\indicator{(i, i)\not\in S} + \phi_t(S \cap C_i)\indicator{(i, i)\in S}\right). 
        \] 
    We define the following: 
\fi
Let $V=[k]\times [k]$ and $f: 2^V\rightarrow \R$ be a set function defined over the set $V$. 
    \begin{enumerate}
    \item For a subset $A\subseteq V$, we denote $A^T:=\{(j,i): (i,j)\in A\}$. 
    \item The set function $f$ is a row-column-type function if there exists a function $g:2^{V}\rightarrow \R_{\ge 0}$ such that $f(A)=\sum_{i=1}^k (g(A\cap R_i) + g(A\cap C_i))$ for every $A\subseteq V$, where $R_i:=\{(i, j): j\in [k]\}$ and $C_i:=\{(j,i): j\in [k]\}$ for every $i\in [k]$.
    
    %\item The set function $f$ is \emph{transpose-invariant} if $f(A)=f(A^T)$ for every $A\subseteq V$. Generalizing, we say that a function $F:[0,1]^V\rightarrow R$ is \emph{transpose-invariant} if $F(\vx)=F(\vx^T)$ for every $\vx \in [0,1]^{[k]\times [k]}$, where $\vx^T$ is the transpose of the matrix $\vx$. 
    
    \item A partition $V_1, V_2, \ldots, V_k$ of $V$ is \emph{symmetric} if  
        %$(i, j) \in V_\ell$ if and only if $(j, i) \in V_\ell$ for all $(i, j) \in V$ and $\ell \in [k]$. 
        $V_{\ell}=V_{\ell}^T$ for every $\ell \in [k]$. 
        Generalizing, we say that $\vx \in [0, 1]^{V \times [k]}$ is \emph{symmetric} if $x(*, \ell) = x(*, \ell)^T$ for every $\ell \in [k]$. 
    \item Let $\opt(f) $ be 
    \[
    \min \left\{\sum_{i = 1}^k f(V_i) \colon V_1, \ldots, V_k \text{ is a partition of } V \text{ with } (i, i) \in V_i\ \forall i \in [k]\right\}. 
    \]
    \item Let $\symopt(f)$ be 
    \[
    \min \left\{\sum_{i = 1}^k f(V_i) \colon V_1, \ldots, V_k \text{ is a symmetric partition of } V \text{ with } (i, i) \in V_i\ \forall i \in [k]\right\}.
    \]
    \item Let $\symgap\left(f\right) := \frac{\symopt(f)}{\opt(f)}$. 
    \end{enumerate}
\end{definition}

Next, we prove that the approximation factor of an algorithm that makes polynomial number of function evaluation queries should be at least the symmetry gap. 

\begin{restatable}{theorem}{thmSymGapToOracleLowerBound}\label{thm:sym-gap-to-oracle-lower-bound}
    Suppose that there exists a row-column-type non-negative monotone submodular function $f\colon 2^V \rightarrow \mathbb{R}$ over the set $V=[k]\times [k]$ such that $\symgap(f)\ge \alpha$. 
    %Let $V=[k]\times [k]$ and $f\colon 2^V \rightarrow \mathbb{R}$ be a row-column-type %transpose-invariant 
    %non-negative monotone submodular function over the set $V$ such that $\symgap(f)\ge \alpha$. 
    Then, for every constant $\eps>0$, every algorithm for \monosubmp with an approximation factor of $\alpha-\epsilon$ makes $2^{\Omega(n)}$ function evaluation queries, where $n$ is the size of the ground set. 
\end{restatable}

Finally, we show that there exists a row-column-type non-negative monotone submodular function with large symmetry gap. 
\begin{restatable}{theorem}{thmLargeSymmetryGap}\label{thm:large-symmetry-gap}
    Let $V^{(k)} = [k] \times [k]$. Then, there exists a sequence $\left(f^{(k)} \colon 2^{V^{(k)}} \rightarrow \mathbb{R}_{\geq 0}\right)_{k \in \mathbb{N}}$ of row-column-type non-negative monotone submodular functions parameterized by $k$ such that 
    \[
        \lim_{k \rightarrow \infty} \symgap\left(f^{(k)}\right) = 10/9.
    \]
\end{restatable}

Theorem \ref{thm:mono-sub-mp-lower-bound} follows from Theorems \ref{thm:sym-gap-to-oracle-lower-bound} and \ref{thm:large-symmetry-gap}. We prove Theorem \ref{thm:sym-gap-to-oracle-lower-bound} in Section \ref{subsec:sym_gap_to_query-complexity} and Theorem \ref{thm:large-symmetry-gap} in Section \ref{subsec:large-symmetry-gap}. Our proof approach for Theorem \ref{thm:sym-gap-to-oracle-lower-bound} is similar to the proof approaches in \cite{Vondrak, EVW}---the result does not follow directly from the results in \cite{Vondrak, EVW} since we are interested in \monosubmp (which was not the focus of those two works), so we have to work out the details. We encourage first-time readers to skip Section \ref{thm:sym-gap-to-oracle-lower-bound}. Our main contribution is Theorem \ref{thm:large-symmetry-gap}, i.e., constructing an instance of \monosubmp with large symmetry gap. 

%\knote{Can we convert oracle hardness results to NP-hardness results as shown in \cite{EVW} and this paper: ``From Query Complexity to Computational Complexity'', Dobzinski and Vondrak?}

\input{mono-sym-gap-to-query-lower-bound}

\subsubsection{Symmetry Gap}
\label{subsec:large-symmetry-gap}
%\rnote{Might need to assume $k \geq 4$. (might need $k \geq \frac{k}{2} + 2$ for Corollary \ref{cor:final_structure})}
In this section, we construct an instance of \monosubmp with large symmetry gap, i.e., prove Theorem \ref{thm:large-symmetry-gap}. We begin by defining the function of interest. 

%We prove Theorem \ref{thm:large-symmetry-gap} in 
%\thmLargeSymmetryGap*

\begin{definition}
\label{def:def:f_for_10/9}
    Let $V = [k] \times [k]$ for $k\ge 4$. For each $i, j\in [k]$, let 
    \begin{align*}
        R_i &= \{(i, j) \colon j \in [k]\},\\
        \overrightarrow{R}_i &= \{(i, j) \colon j \in \{i + 1, i + 2, \ldots, k\} \}, \\
        \overleftarrow{R}_i &= \{(i, j) \colon j \in [i - 1]\}, \text{ and}\\
        C_j &= \{(i, j) \colon i \in [k]\}. 
    \end{align*}
    We define the following functions:
    \begin{enumerate}
        \item 
        $\phi_n \colon \mathbb{R} \rightarrow \mathbb{R}_{\geq 0}$ is defined as $\phi_n(a) = \min\left\{a, \frac{7}{8}k\right\}$;

        \item 
        $\phi_t \colon \mathbb{R} \rightarrow \mathbb{R}_{\geq 0}$ is defined as $\phi_t(a) = \min\left\{\frac{3}{8}ka, \frac{3}{8}k + a - 1, \frac{7}{8}k\right\}$;

        \item 
        $g \colon 2^V \rightarrow \mathbb{R}_{\geq 0}$ is defined as
        \[
            g(S) = \begin{cases} \phi_t\left(|S|\right) & \text{if } \left\{(\ell, \ell) \colon \ell \in [k] \right\} \cap S \neq \emptyset, \\ \phi_n\left(|S|\right) & \text{otherwise}; \end{cases}
        \]
        \item 
        $f \colon 2^V \rightarrow \mathbb{R}_{\geq 0}$ is defined as 
        \[
            f(S) = \sum_{i = 1}^k \left(g(R_i \cap S) + g(C_i \cap S) \right). 
        \]
    \end{enumerate}
    We say that $P_1, \ldots, P_k$ is a symmetric multiway partition of $V$ if $P_1, \ldots, P_k$ is symmetric partition of $V$ satisfying $(\ell, \ell) \in P_\ell$ for each $\ell \in [k]$. 
\end{definition}

In Lemma \ref{lem:function-is-monotone-submodular}, we show that the function $f$ defined above is row-column-type, non-negative, monotone, and submodular. 
In Lemma \ref{lem:opt-upper-bound}, we show that $\opt(f)\le (9k^2-k)/4$. 
In Lemma \ref{lem:sym-opt-lower-bound}, we show that $\opt(f)\ge (10k^2-22k-1)/4$ %\rnote{I think this should be $(10k^2-22k-1)/4$}. 
These three lemmas together imply Theorem \ref{thm:large-symmetry-gap}. We now prove these lemmas. 

\begin{lemma}\label{lem:function-is-monotone-submodular}
    Let $V = [k] \times [k]$ and $f \colon 2^V \rightarrow \mathbb{R}_{\geq 0}$ be defined as in Definition \ref{def:def:f_for_10/9}. Then, $f$ is a row-column-type non-negative monotone submodular function. 
\end{lemma}
\begin{proof}
    From the definition of $f$, we have that $f$ is row-column-type and non-negative. Moreover, $f$ is monotone since $\phi_n$ and $\phi_t$ are monotone and $\phi_t(a) \geq \phi_n(a)$ for all $a \in \{0, \ldots, k\}$. Now we prove that $f$ is submodular. We first prove the following observations in which $a$ and $c$ are assumed to take integer values:
    \begin{enumerate}
        \item For all $a \leq c$, $\phi_n(a + 1) - \phi_n(a) \geq \phi_n(c + 1) - \phi_n(c)$.

        \item For all $a \leq c$, $\phi_t(a + 1) - \phi_t(a) \geq \phi_t(c + 1) - \phi_t(c)$. 

        \item For all $0 \leq a \leq a + 1 \leq c$, $\phi_n(a + 1) - \phi_n(a) \geq \phi_t(c + 1) - \phi_t(c)$.

        \item For all $0 \leq a \leq c$, $\phi_t(a + 1) - \phi_n(a) \geq \phi_t(c + 1) - \phi_n(c)$.
    \end{enumerate}
    The first two observations follow from the concavity of $\phi_n$ and $\phi_t$, respectively. 
    
    To prove that if $0 \leq a$ and $a + 1 \leq c$, then $\phi_n(a + 1) - \phi_n(a) \geq \phi_t(c + 1) - \phi_t(c)$, we examine two cases: $a \leq \frac{7}{8}k - 1$ and $ a \geq \frac{7}{8}k$. If $a \leq \frac{7}{8}k$, then $\phi_n(a + 1) - \phi_n(a) = 1$. Additionally, because $0 \leq a$ and $a + 1 \leq c$, $c \geq 1$, it follows that $1 \geq \phi_t(c + 1) - \phi_t(c)$. If $a \geq \frac{7}{8}k$, then $\phi_n(a + 1) - \phi_n(a) = 0 = \phi_t(c + 1) - \phi_t(c)$ since $c \geq a + 1$. 

    To prove that if $0 \leq a \leq c$, $\phi_t(a + 1) - \phi_n(a) \geq \phi_t(c + 1) - \phi_n(c)$, we examine two cases: $c \leq \frac{7}{8}k$ and $c \geq \frac{7}{8}k + 1$. If $c \leq \frac{7}{8}k$, then $\phi_n(c) - \phi_n(a) = c - a$ while $\phi_t(c + 1) - \phi_t(a + 1) \leq c - a$ since $a + 1 \geq 1$. If $c \geq \frac{7}{8}k$, then $\phi_n(c) - \phi_n(a) = \max\left\{\frac{7}{8}k - a, 0\right\}$, and since $a + 1 \geq 1$, we have that $\phi_t(c + 1) - \phi_t(a + 1) = \max \left\{\frac{k}{2} - a + 1, 0\right\}$. Hence, $\phi_n(c) - \phi_n(a) \geq \phi_t(c + 1) - \phi_t(a + 1)$. Rearranging, we obtain that $\phi_t(a + 1) - \phi_n(a) \geq \phi_t(c + 1) - \phi_n(c)$. 

    We now prove submodularity armed with the above observations. 
    We will show that for all $X \subseteq Y \subseteq V$ and $(i, j) \in V \setminus Y$, we have that $f\left(X + (i, j)\right) - f\left(X\right) \geq f\left(Y + (i, j)\right) - f\left(Y\right)$. Let $X \subseteq Y \subseteq V$ and $(i, j) \in V \setminus Y$. We observe that $f\left(X + (i, j)\right) - f\left(X\right) = g\left(R_i \cap X + (i, j)\right) - g\left(R_i \cap X\right) + g\left(C_j \cap X + (i, j)\right) - g\left(C_j \cap X\right)$. Similarly, $f\left(Y + (i, j)\right) - f\left(Y\right) = g\left(R_i \cap Y + (i, j)\right) - g\left(R_i \cap Y\right) + g\left(C_j \cap Y + (i, j)\right) - g\left(C_j \cap Y\right)$. 

    First, we will consider the case in which $i \neq j$. In particular, we show that $g\left(R_i \cap X + (i, j)\right) - g\left(R_i \cap X\right) \geq g\left(R_i \cap Y + (i, j)\right) - g\left(R_i \cap Y\right)$. We let $|R_i \cap X| = a$ and $|R_i \cap Y| = c$. If $(i, i) \notin Y$, then $g\left(R_i \cap X + (i, j)\right) - g\left(R_i \cap X\right) = \phi_n(a + 1) - \phi_n(a) \geq \phi_n(c + 1) - \phi_n(c) = g\left(R_i \cap Y + (i, j)\right) - g\left(R_i \cap Y\right)$. If $(i, i) \in X$, then $g\left(R_i \cap X + (i, j)\right) - g\left(R_i \cap X\right) = \phi_t(a + 1) - \phi_t(a) \geq \phi_t(c + 1) - \phi_t(c) = g\left(R_i \cap Y + (i, j)\right) - g\left(R_i \cap Y\right)$. Finally, if $(i, i) \in Y \setminus X$, then $c \geq a + 1$ and $a \geq 0$, so $g\left(R_i \cap X + (i, j)\right) - g\left(R_i \cap X\right) = \phi_n(a + 1) - \phi_n(a) \geq \phi_t(c + 1) - \phi_t(c) = g\left(R_i \cap Y + (i, j)\right) - g\left(R_i \cap Y\right)$. We have shown that $g\left(R_i \cap X + (i, j)\right) - g\left(R_i \cap X\right) \geq g\left(R_i \cap Y + (i, j)\right) - g\left(R_i \cap Y\right)$ in all cases. Similar reasoning shows that $g\left(C_j \cap X + (i, j)\right) - g\left(C_j \cap X\right) \geq g\left(C_j \cap Y + (i, j)\right) - g\left(C_j \cap Y\right)$. Thus, if $i \neq j$, then $f\left(X + (i, j)\right) - f\left(X\right) \geq f\left(Y + (i, j)\right) - f\left(Y\right)$. 

    Now, we consider the case in which $i = j$. Again, we let $|R_i \cap X| = a$ and $|R_i \cap Y| = c$. Then $g\left(R_i \cap X + (i, i)\right) - g\left(R_i \cap X\right) = \phi_t(a + 1) - \phi_n(a) \geq \phi_t(c + 1) - \phi_n(c) = g\left(R_i \cap Y + (i, i)\right) - g\left(R_i \cap Y\right)$. Similar reasoning shows that $g\left(C_j \cap X + (i, i)\right) - g\left(C_j \cap X\right) \geq g\left(C_j \cap Y + (i, i)\right) - g\left(C_j \cap Y\right)$. Hence, $f\left(X + (i, i)\right) - f\left(X\right) \geq f\left(Y + (i, i)\right) - f\left(Y\right)$. 
\end{proof}

Next, we show that the function $f$ of interest to this section has small $\opt$. 
\begin{lemma}\label{lem:opt-upper-bound}
    Let $V = [k] \times [k]$ and $f \colon 2^V \rightarrow \mathbb{R}_{\geq 0}$ be defined as in Definition \ref{def:def:f_for_10/9}. Then $\opt\left(f\right) \leq \frac{9k^2 - 4k}{4}$. 
\end{lemma}
\begin{proof}
    We upper bound the optimum value by exhibiting a cheap feasible solution. Consider the partition $V_1, \ldots, V_k$ of $V$ defined by $V_i := R_i$ for all $i\in [k]$. We observe that $V_1, \ldots, V_k$ is a feasible multiway partition of $V$ with $(i,i)\in V_i$ for every $i\in [k]$. Hence, we have that
    \begin{align*}
        \opt(f)
        \le \sum_{\ell = 1}^k f(V_\ell) 
        &= \sum_{\ell = 1}^k \sum_{i = 1}^k \left(g(R_i \cap V_\ell) + g(C_i \cap V_\ell) \right) \\
        &= \sum_{\ell = 1}^k \sum_{i = 1}^k \left(g(R_i \cap R_\ell) + g(C_i \cap R_\ell)\right) \\
        &= \sum_{\ell = 1}^k \left( \phit{R_\ell \cap R_\ell} + \phit{\{(\ell, \ell)\}} + \sum_{i \neq \ell} \phin{\{(\ell, i)\}} \right) \\
        &= \sum_{\ell = 1}^k \left( \frac{7}{8}k + \frac{3}{8}k + k - 1 \right) \\
        &= k\left(\frac{18}{8}k - 1\right) \\
        &= \frac{9k^2 - 4k}{4}.
    \end{align*}
    %Hence, $\opt\left(f\right) \leq \frac{9k^2 - 4k}{4}$. 
\end{proof}

Finally, we show that the function $f$ of interest to this section has large $\symopt$. The proof of the following lemma is non-trivial and we devote Section \ref{sec:sym-opt-lower-bound} to prove it. 
\begin{lemma}\label{lem:sym-opt-lower-bound}
Let $V = [k] \times [k]$ and $f \colon 2^V \rightarrow \mathbb{R}_{\geq 0}$ be defined as in Definition \ref{def:def:f_for_10/9}. Let $Q_1, \ldots, Q_k$ be a symmetric multiway partition of $V$. Then, $\sum_{\ell = 1}^k f(Q_\ell) \geq \frac{10k^2 - 22k - 1}{4}$ %\rnote{I think this should be $(10k^2-22k-1)/4$}. 
\end{lemma}
%We devote Section \ref{sec:sym-opt-lower-bound} to prove Lemma \ref{lem:sym-opt-lower-bound}.  
\iffalse
\begin{proof}
    \donote{}
    
\end{proof}
\fi

\subsubsection{Proof of Lemma \ref{lem:sym-opt-lower-bound}}\label{sec:sym-opt-lower-bound}
In order to prove Lemma \ref{lem:sym-opt-lower-bound}, we first show a lower bound on the objective value of certain structured symmetric multiway partitions (see Lemma \ref{lem:structured-sym-opt-lower-bound}). Next, we show that every symmetric multiway partition can be converted to a structured symmetric multiway partition with a small additive loss (see Lemma \ref{lem:sym_part_structure}). See Figure \ref{fig:sym-part-structure} for an illustration of the structured symmetric multiway partition of interest that will be of interest to this section. 
For notational convenience, we define the objective value of a symmetric multiway partition $P_1, \ldots, P_k$ of $V$ as follows: 
\begin{align*}
    \obj{P_1, \ldots, P_k} &= \frac{1}{2} \sum_{\ell = 1}^k f(P_\ell).
\end{align*}
We observe that 
\begin{align*}
\obj{P_1, \ldots, P_k}
        &= \sum_{i = 1}^k \left(\phi_t\left(|R_i \cap P_i|\right) + \sum_{\ell \in [k] \setminus \{i\}} \phi_n\left(|R_i \cap P_\ell)\right)\right). 
\end{align*}

\begin{figure}[ht]
\centering
\begin{tikzpicture}
    \definecolor{brightube}{rgb}{0.82, 0.62, 0.91}
    
    % Define grid size
    \def\gridsize{9}
    \def\spacing{0.9} % Distance between points

    % Draw the grid of vertices
    \foreach \x in {1,...,\gridsize} {
        \foreach \y in {1,...,\gridsize} {
            % Draw regular black vertices
            \filldraw[black] (\x*\spacing, \y*\spacing) circle (2pt);
        }
    }
    % coloring L's
    
    \foreach \i/\c in {0/red, 1/orange, 2/yellow, 3/green, 4/blue}{
    \fill[\c, opacity=0.4] (\i*\spacing + 0.5*\spacing, \gridsize*\spacing-0.5*\spacing - \i*\spacing) rectangle (\gridsize*\spacing+0.5*\spacing, \gridsize*\spacing+0.5*\spacing - \i*\spacing);

    \fill[\c, opacity=0.4] (0.5*\spacing + \i*\spacing, 0.5*\spacing) rectangle (1*\spacing+0.5*\spacing + \i*\spacing, \gridsize*\spacing+0.5*\spacing - \spacing - \i*\spacing);
    }

    % remaining vertices
    
    \foreach \i/\c in {5/brightube, 6/gray, 7/brown, 8/black}{
        \foreach \j in {5, 6, 7, 8}{
                \ifthenelse{\i=\j}{\fill[\c, opacity = 0.4] (0.5*\spacing + \i*\spacing, \gridsize*\spacing - 0.5*\spacing - \j*\spacing) rectangle (0.5*\spacing + \i*\spacing +\spacing, \gridsize*\spacing - 0.5*\spacing - \j*\spacing + \spacing);}{\fill[red, opacity = 0.4] (0.5*\spacing + \i*\spacing, \gridsize*\spacing - 0.5*\spacing - \j*\spacing) rectangle (0.5*\spacing + \i*\spacing +\spacing, \gridsize*\spacing - 0.5*\spacing - \j*\spacing + \spacing);}
        }
    }
\end{tikzpicture}
\caption{A figure illustrating the structure that will be obtained in Lemma \ref{lem:sym_part_structure}. Each color corresponds to a part in the symmetric multiway partition. In particular, $i^*=5$ in this symmetric multiway partition.}
\label{fig:sym-part-structure}
\end{figure}
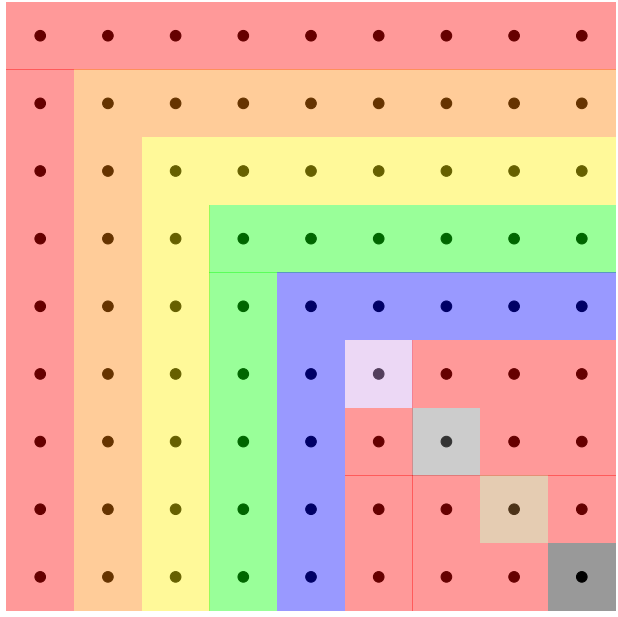

\begin{lemma}\label{lem:structured-sym-opt-lower-bound}
Let $V = [k] \times [k]$ and $f \colon 2^V \rightarrow \mathbb{R}_{\geq 0}$ be defined as in Definition \ref{def:def:f_for_10/9}. Let $Q_1, \ldots, Q_k$ be a symmetric multiway partition of $V$ for which there exists $i^* \in [k]$ such that $\overrightarrow{R}_i \subseteq Q_i$ for all $i \in [i^*]$ and $\overrightarrow{R}_i \subseteq Q_1$ for all $i \in [k] \setminus [i^*]$. Then, $\sum_{\ell = 1}^k f(Q_\ell) \geq \frac{10k^2 - 6k - 1}{4}$. 
\end{lemma}
\begin{proof}
    We have that 
    \begin{align}
        \obj{Q_1, \ldots, Q_k} = \sum_{i = 1}^k \left(\phit{R_i \cap Q_i} + \sum_{\ell \in [k] - i} \phin{R_i \cap Q_\ell}\right). \label{eq:obj-expression}
    \end{align}
    Let $i \in [i^*]$. Then, we have that $R_i\cap Q_i=\overrightarrow{R}_i\cup\{(i,i)\}$. Moreover, we have $R_i\cap Q_{\ell}=\{(i, \ell)\}$ for every $\ell\in [i-1]$ and $R_i\cap Q_{\ell}=\emptyset$ for every $\ell\in \{i,\ldots, k\}$. Hence, for every $i\in [i^*]$, we have that 
    \begin{align}
        \phit{R_i \cap Q_i} + \sum_{\ell \in [k] - i} \phin{R_i \cap Q_\ell} &= \phi_t(k + 1 - i) + (i - 1)\phi_n(1) \notag\\
        &= \phi_t(k + 1 - i) + (i - 1). \label{eq:smaller-than-i*}
    \end{align}
    Let $i \in \{i^* + 1, \ldots, k\}$. Then, we have that $R_i\cap Q_i=\{(i,i)\}$. Moreover, we have that $R_i\cap Q_{1}=\overrightarrow{R}_i$, $R_i\cap Q_{\ell}=\{(i, \ell)\}$ for every $\ell\in \{2, \ldots, i^*\}$ and $R_i\cap Q_{\ell}=\emptyset$ for every $ell\in \{i^*+1, \ldots, k\}$. Hence, for every $i\in \{i^*+1, \ldots, k\}$, we have that 
    \begin{align}
        \phit{R_i \cap Q_i} + \sum_{\ell \in [k] - i} \phin{R_i \cap Q_\ell} &= \phi_t(1) + \phi_n\left(k - i^*\right) + \left(i^* - 1\right)\phi_n(1) \notag\\
        &= \frac{3}{8}k + \phi_n\left(k - i^*\right) + \left(i^* - 1\right). \label{eq:larger-than-i*}
    \end{align}
    We consider three cases based on the value of $i^*$ and show the lower bound in all three cases. %(i) $i^* \leq \frac{1}{8}k$, $\frac{1}{8}k < i^* \leq \frac{k}{2}$, and $\frac{k}{2} < i^*$. 
    
    \noindent\textbf{Case 1.} Suppose that $i^* \leq \frac{1}{8}k$. Then, we have $\phi_t(k + 1 - i) = \frac{7}{8}k$ for each $i \in [i^*]$ and $\phi_n(k - i^*) = \frac{7}{8}k$. Hence, by \eqref{eq:obj-expression}, \eqref{eq:smaller-than-i*}, and \eqref{eq:larger-than-i*}, we have that 
    \begin{align}
        \obj{Q_1, \ldots, Q_k} &= \sum_{i = 1}^{i^*} \left(\frac{7}{8}k + i - 1\right) + \left(k - i^*\right)\left(\frac{3}{8}k + \frac{7}{8}k + i^* - 1\right) \notag\\
        &= ki^* - k - \left(i^*\right)^2 + i^* + \frac{1}{8}\left(4(i^*)^2-3ki^*-4i^* + 10k^2\right). \label{eq:obj-small-i*}
    \end{align}
    Taking the derivative with respect to $i^*$, we have 
    \[
        \frac{\partial\obj{Q_1, \ldots, Q_k}}{\partial i^*} = -i^* + \frac{5}{8}k + \frac{1}{2} \geq \frac{k}{2} + \frac{1}{2} \geq 0,
    \]
    where the first inequality is because $i^*\le \frac{1}{8}k$. 
    Hence, for $i^* \leq \frac{1}{8}k$, the minimum value of $\obj{Q_1, \ldots, Q_k}$ is is attained at $i^* = 1$. Substituting $i^*=1$ in \eqref{eq:obj-small-i*}, we obtain that 
    \[
        \obj{Q_1, \ldots, Q_k} \geq k - k - 1 + 1 + \frac{1}{8}\left(4 - 3k - 4 + 10k^2\right) = \frac{10k^2 - 3k}{8}.
    \]

    \noindent\textbf{Case 2.} Suppose that $\frac{1}{8}k < i^* \leq \frac{k}{2}$. Then, $\phi_t(k + 1 - i) = \frac{7}{8}k$ and $\phi_n(k - i^*) = k - i^*$. Hence, by \eqref{eq:obj-expression}, \eqref{eq:smaller-than-i*}, and \eqref{eq:larger-than-i*}, we have that 
    \begin{align}
        \obj{Q_1, \ldots, Q_k} &= \sum_{i = 1}^{i^*} \left(\frac{7}{8}k + i - 1\right) + \left(k - i^*\right)\left(\frac{3}{8}k + k - i^* + i^* - 1\right) \notag\\
        &= \sum_{i = 1}^{i^*} \left(\frac{7}{8}k + i - 1\right) + \left(k - i^*\right)\left(\frac{11}{8}k - 1\right) \notag\\
        &= \frac{1}{2}\left((i^*)^2 - i^*\right) + \frac{1}{8}\left(11k^2 - 4ki^*\right) - k + i^*. \label{eq:medium-i*}
    \end{align}
    Taking the derivative with respect to $i^*$, we have 
    \[
        \frac{\partial\obj{Q_1, \ldots, Q_k}}{\partial i^*} = i^* + \frac{1}{2} - \frac{k}{2}. 
    \]
    We observe that $\frac{\partial^2\obj{Q_1, \ldots, Q_k}}{\partial (i^*)^2} = 1$, so $\obj{Q_1, \ldots, Q_k}$ is a convex function of $i^*$. Hence, for $\frac{1}{8}k < i^* \leq \frac{k}{2}$, the minimum value for $\obj{Q_1, \ldots, Q_k}$ is attained at $i^* = \frac{k}{2} - \frac{1}{2}$. Thus, evaluating the expression \eqref{eq:medium-i*} at $i^* = \frac{k}{2} - \frac{1}{2}$, we obtain that
    \[
        \obj{Q_1, \ldots, Q_k} \geq \frac{10k^2 - 6k - 1}{8}.
    \]

    \noindent\textbf{Case 3.} Suppose that $\frac{k}{2} < i^*$. Then, 
    \begin{align*}
    \phi_t(k + 1 - i) &= \begin{cases} \frac{7}{8}k & i \leq \frac{k}{2} \\ \frac{11}{8}k - i & i > \frac{k}{2}\end{cases} \text{ and}\\
    \phi_n(k - i^*) &= k - i^*.
    \end{align*}
    Hence, by \eqref{eq:obj-expression}, \eqref{eq:smaller-than-i*}, and \eqref{eq:larger-than-i*}, we have that 
    \begin{align*}
        \obj{Q_1, \ldots, Q_k} &= \sum_{i = 1}^{\frac{k}{2}} \left(\frac{7}{8}k + i - 1\right) + \sum_{i = \frac{k}{2} + 1}^{i^*} \left(\frac{11}{8}k - i + i - 1 \right) + \left(k - i^*\right)\left(\frac{3}{8}k + k - i^* + i^* - 1\right) \\
        &= \sum_{i = 1}^{\frac{k}{2}} \left(\frac{7}{8}k + i - 1\right) + \left(i^* - \frac{k}{2}\right)\left(\frac{11}{8}k - 1\right) + \left(k - i^*\right)\left(\frac{11}{8}k - 1\right) \\
        &= \sum_{i = 1}^{\frac{k}{2}} \left(\frac{7}{8}k + i - 1\right) + \left(\frac{k}{2}\right)\left(\frac{11}{8}k - 1\right) \\
        &= \frac{10k^2 -6k}{8}. 
    \end{align*}
    \iffalse
    Hence, in all cases, $\obj{Q_1, \ldots, Q_k} \geq \frac{10k^2 - 6k - 1}{8}$. This demonstrates that 
    \[
        \sum_{\ell = 1}^k f(Q_\ell) = 2\obj{Q_1, \ldots, Q_k} \geq \frac{10k^2 - 6k - 1}{4}.
    \]
    \fi
\end{proof}

The following lemma shows that every symmetric multiway partition can be converted to a structured symmetric multiway partition with a small additive loss in the objective. 
\begin{restatable}{lemma}{lemSymPartStructure}
\label{lem:sym_part_structure}
    Let $V = [k] \times [k]$. Let $P_1, \ldots, P_k$ be a symmetric multiway partition of $V$. Then, there exists a symmetric multiway partition $Q_1, \ldots, Q_k$ of $V$ with the following properties:
    \begin{enumerate}
        \item 
        $\sum_{\ell = 1}^k f(Q_\ell) \leq \sum_{\ell = 1}^k f(P_\ell) + 4k$ % \rnote{I think this should be $\sum_{\ell = 1}^k f(Q_\ell) \leq \sum_{\ell = 1}^k f(P_\ell) + 4k$} and 

        \item 
        There exists $i^* \in [k]$ such that $\overrightarrow{R}_i \subseteq Q_i$ for all $i \in [i^*]$ and $\overrightarrow{R}_i \subseteq Q_1$ for all $i \in [k] \setminus [i^*]$.
    \end{enumerate}
\end{restatable}
Lemmas \ref{lem:structured-sym-opt-lower-bound} and \ref{lem:sym_part_structure} together complete the proof of Lemma \ref{lem:sym-opt-lower-bound}. The rest of this section is devoted to proving Lemma \ref{lem:sym_part_structure}. For the rest of this section, let $V:=[k]\times [k]$.
We will denote a symmetric multiway partition of $V$ satisfying the second condition of Lemma \ref{lem:sym_part_structure} as a \emph{structured symmetric partition} (see Figure \ref{fig:sym-part-structure} for an example). 
To prove Lemma \ref{lem:sym_part_structure}, we will prove several claims about how symmetric multiway partitions of $V = [k] \times [k]$ can be transformed into a structured symmetric multiway partitions without increasing the cost of the partition by more than $4k$. % \rnote{$4k$, but really $2(k + 1)$ is enough}.
% For notational convenience, we define the objective value of a symmetric multiway partition $P_1, \ldots, P_k$ of $V$ as follows: 
% \begin{align*}
%     \obj{P_1, \ldots, P_k} &= \frac{1}{2} \sum_{\ell = 1}^k f(P_\ell).
% \end{align*}
% We observe that 
% \begin{align*}
% \obj{P_1, \ldots, P_k}
%         &= \sum_{i = 1}^k \left(\phi_t\left(|R_i \cap P_i|\right) + \sum_{\ell \in [k] \setminus \{i\}} \phi_n\left(|R_i \cap P_\ell)\right)\right). 
% \end{align*}

The requirement for structured symmetric multiway partition for $i=1$ is that $R_1\subseteq Q_1$. 
We begin with Claim \ref{claim:move-first-row} below which shows that this property can be achieved for $i=1$ with an additive loss in the objective value. We emphasize that Claim \ref{claim:move-first-row} will be the only transformation that incurs a loss in the objective value. The rest of the transformations will transform \emph{partially} structured symmetric partition into a fully structured symmetric partition (as desired in Lemma \ref{lem:sym_part_structure}) without any loss in the objective value. 
\begin{claim}\label{claim:move-first-row}
Let $P_1, \ldots, P_k$ be a symmetric multiway partition of $V$. Then, there exists a symmetric multiway partition $Q_1, \ldots, Q_k$ of $V$ such that 
\begin{enumerate}
    \item $R_1\subseteq Q_1$ and
    \item $\obj{Q_1,\ldots, Q_k} \le \obj{P_1, \ldots, P_k} + 2k$ % \rnote{$\obj{Q_1,\ldots, Q_k} \le \obj{P_1, \ldots, P_k} + 2k$, although I think it is even true that $\obj{Q_1,\ldots, Q_k} \le \obj{P_1, \ldots, P_k} + k + 1$}. 
\end{enumerate}
\end{claim}
\begin{proof}
    Let use define $Q_1 = P_1 \cup (R_1 \cup C_1)$ and $Q_\ell = P_\ell \setminus (R_1 \cup C_1)$ for each $\ell \in [k] \setminus \{1\}$. We observe that $Q_1, \ldots, Q_k$ is a symmetric multiway partition with $R_1 \subseteq Q_1$. It remains to show that $\obj{Q_1, \ldots, Q_k} \leq \obj{P_1, \ldots, P_k} + 2k$. 

    We have 
    \[  
        \obj{Q_1, \ldots, Q_k} - \obj{P_1, \ldots, P_k} = \sum_{i = 1}^k \left(\phit{R_i \cap Q_i} - \phit{R_i \cap P_i}  + \sum_{\ell \in [k] \setminus \{i\}} \left(\phin{R_i \cap Q_\ell} - \phin{R_i \cap P_\ell}\right)\right). 
    \]

    First we consider $i = 1$. We have $R_1 \cap Q_1 = R_1$, so $\phit{R_1 \cap Q_1} = \frac{7}{8}k$. Additionally, $\phit{R_1 \cap P_1} \geq 0$ such that $\phit{R_1 \cap Q_1} - \phit{R_1 \cap P_1} \leq \frac{7}{8}k$. Since $R_1 \cap Q_\ell = \emptyset$ for each $\ell \in [k] \setminus \{1\}$, we have $\phin{R_1 \cap Q_\ell} = 0$ for each $\ell \in [k] \setminus \{1\}$. Also, we have $\phin{R_1 \cap P_\ell} \geq 0$ for each $\ell \in [k] \setminus \{1\}$. Therefore, for each $\ell \in [k] \setminus \{1\}$, we have $\phin{R_1 \cap Q_\ell} - \phin{R_1 \cap P_\ell} \leq 0$. Hence, 
    \[
        \phit{R_1 \cap Q_1} - \phit{R_1 \cap P_1} + \sum_{\ell \in [k] \setminus \{1\}} \left(\phin{R_1 \cap Q_\ell} - \phin{R_1 \cap P_\ell}\right) \leq \frac{7}{8}k \leq k.
    \]

    Now, we consider $i \in [k] \setminus \{1\}$. When transforming $P_1, \ldots, P_k$ into $Q_1, \ldots, Q_k$, the only way in which we may change the assignment in row $R_i$ is by moving element $(i, 1)$ into part $Q_1$. Hence, $R_i \cap Q_i \subseteq R_i \cap P_i$. Then, we have $\phit{R_i \cap Q_i} - \phit{R_i \cap P_i} \leq 0$. Likewise for each $\ell \in [k] \setminus \{1, i\}$, we have $R_i \cap Q_\ell \subseteq R_i \cap P_\ell$, so $\phin{R_i \cap Q_\ell} - \phin{R_i \cap P_\ell} \leq 0$. Finally, we have $|R_i \cap Q_1| \leq |R_i \cap P_1| + 1$, so $\phin{R_i \cap Q_1} - \phin{R_i \cap P_\ell} \leq 1$. Overall, this shows that for $i \in [k] \setminus \{1\}$, 
    \[
        \phit{R_i \cap Q_i} - \phit{R_i \cap P_i} + \sum_{\ell \in [k] \setminus \{i\}} \left(\phin{R_i \cap Q_\ell} - \phin{R_i \cap P_\ell}\right) \leq 1. 
    \]

    Therefore, 
    \[
        \obj{Q_1, \ldots, Q_k} - \obj{P_1, \ldots, P_k} \leq k + (k - 1)(1) \leq 2k.
    \]
\end{proof}

The rest of the proof of Lemma \ref{lem:sym_part_structure} involves transforming a partially structured symmetric partition into a structured symmetric partition without any loss in the objective value. For this, we introduce the notion of unhappy elements. 
    
\begin{definition}[unhappy elements]
\label{def:unhappy}
    Let $P_1, \ldots, P_k$ be a partition of $V$. An element $(i, j) \in V$ is \emph{unhappy under $P_1, \ldots, P_k$} if $(i, j) \in P_\ell$ where $\ell \notin \{i, j\}$. Moreover, let 
    \[
    U(P_1, \ldots, P_k) := \{(i, j) \in [k]\times [k] \colon  (i, j) \text{ is unhappy under } P_1, \ldots, P_k\}. 
    \]
\end{definition}

We now show that unhappy elements can be moved to the first part of a partially structured partition without increasing the objective value. 
%We note that $U(P_1, \ldots, P_k)\subseteq P_1$ implies that $R_1\subseteq P_1$ \rnote{I don't think this is true since e.g. $(1, 2) \in P_2$ would not be unhappy}, so the following claim can be viewed as a strengthening of Claim \ref{claim:move-first-row}. 
\begin{claim}
\label{claim:move-unhappy}
    %Let $V = [k] \times [k]$. 
    Let $P_1, \ldots, P_k$ be a symmetric multiway partition of $V$ such that $R_1 \subseteq P_1$. 
    Then, there exists a symmetric multiway partition $Q_1, \ldots, Q_k$ of $V$ such that 
    %For $U:=U(P_1, \ldots, P_k)$, let $Q_1 := P_1 \cup U$, and $Q_{\ell} := P_{\ell} \setminus U$ for each $\ell\in \{2,3,\ldots, k\}$. Then, 
    \begin{enumerate}
        \item $R_1 \subseteq Q_1$, 
        \item $U(Q_1, \ldots, Q_k)\subseteq Q_1$, and 
        \item $\obj{Q_1, \ldots, Q_k} \leq \obj{P_1, \ldots, P_k}$. 
    \end{enumerate}
\end{claim}
\begin{proof}
For $U:=U(P_1, \ldots, P_k)$, let $Q_1 := P_1 \cup U$, and $Q_{\ell} := P_{\ell} \setminus U$ for each $\ell\in \{2,3,\ldots, k\}$. Then, $Q_1, \ldots, Q_k$ is a symmetric multiway partition of $V$     The first two properties follow by definition of $Q_1, \ldots, Q_k$. It suffices to bound the objective value. 
We have that 
    \begin{align*}
        \obj{Q_1, \ldots, Q_k} - \obj{P_1, \ldots, P_k} 
        &= \sum_{i = 1}^k \sum_{\ell = 1}^k \left(g(R_i \cap Q_\ell) - g(R_i \cap P_\ell)\right).
        %&= \sum_{\ell = 1}^k \left(g(R_1 \cap Q_\ell) - g(R_1 \cap P_\ell)\right) 
         %\quad \quad \quad \quad 
         %+ \sum_{i = 2}^k \sum_{\ell = 1}^k \left(g(R_i \cap Q_\ell) - g(R_i \cap P_\ell)\right). 
    \end{align*}
In order to complete the proof, we show that  $\sum_{\ell = 1}^k \left(g(R_i \cap Q_\ell) - g(R_i \cap P_\ell)\right)\le 0$ for each $i\in [k]$. 

First, we consider $i = 1$. 
    \begin{align*}
        \sum_{\ell = 1}^k \left(g(R_1 \cap Q_\ell) - g(R_1 \cap P_\ell)\right) = \phit{R_1 \cap Q_1} - \phit{R_1 \cap P_1} + \sum_{\ell \neq 1} \left(\phin{R_1 \cap Q_\ell} - \phin{R_1 \cap P_\ell}\right) = 0,
    \end{align*}
    since $R_1 \subseteq P_1$ and $R_1 \subseteq Q_1$ implies that $|R_1 \cap Q_1| = |R_1 \cap P_1| = k$ and $|R_1 \cap Q_\ell| = |R_1 \cap P_\ell| = 0$ for each $\ell \in \{2, 3, \ldots, k\}$. 

Next, %we show that $\sum_{\ell = 1}^k \left(g(R_i \cap Q_\ell) - g(R_i \cap P_\ell)\right) \leq 0$ for each $i\in \{2, \ldots, k\}$. Let 
let $i \in \{2, \ldots, k\}$. We need to show that $\sum_{\ell = 1}^k \left(g(R_i \cap Q_\ell) - g(R_i \cap P_\ell)\right) \leq 0$. 
    We begin by showing that 
    \begin{align}
        |R_i\cap P_{\ell}|\le |R_i\cap Q_1| \text{ for each } \ell\in [k]\setminus \{1, i\}. \label{ineq:set-sizes}
    \end{align}
    We have that $R_i \cap P_i \cap U = \emptyset$, so $R_i \cap P_i = R_i \cap Q_i$. Moreover, $\sum_{\ell' = 1}^k \left(|R_i \cap Q_{\ell'}| - |R_i \cap P_{\ell'}|\right)=0$. Hence, we have that 
    \[
        \sum_{\ell' \neq i} \left(|R_i \cap Q_{\ell'}| - |R_i \cap P_{\ell'}|\right) = \sum_{\ell' = 1}^k \left(|R_i \cap Q_{\ell'}| - |R_i \cap P_{\ell'}|\right) = 0. 
    \]
    Therefore,
    \[
        |R_i \cap Q_1| - |R_i \cap P_1| + \sum_{\ell' \notin \{1, i\}} \left(|R_i \cap Q_{\ell'}| - |R_i \cap P_{\ell'}|\right) = 0. 
    \]
    Hence,  
    \[
        |R_i \cap Q_1| - |R_i \cap P_1| = \sum_{\ell' \notin \{1, i\}} \left(|R_i \cap P_{\ell'}| - |R_i \cap Q_{\ell'}|\right). 
    \]
    Since $R_i \cap Q_{\ell'} \subseteq R_i \cap P_{\ell'}$ for each $\ell' \in [k] \setminus \{1, i\}$, we have that 
    %$|R_i \cap P_{\ell'}| \geq |R_i \cap Q_{\ell'}|$ or 
    $|R_i \cap P_{\ell'}| - |R_i \cap Q_{\ell'}| \geq 0$ for each $\ell' \in [k] \setminus \{1, i\}$.
    Then, for each $\ell \in [k] \setminus \{1, i\}$,
    \[
        |R_i \cap P_\ell| - |R_i \cap Q_\ell| \leq \sum_{\ell' \notin \{1, i\}} \left(|R_i \cap P_{\ell'}| - |R_i \cap Q_{\ell'}|\right) = |R_i \cap Q_1| - |R_i \cap P_1|.
    \] 
    Hence, for each $\ell \in [k] \setminus \{1, i\}$,
    \begin{align*}
        |R_i \cap P_\ell| &\leq |R_i \cap Q_1| - |R_i \cap P_1| + |R_i \cap Q_\ell| \\
        &\leq |R_i \cap Q_1| - |R_i \cap P_1| + 1 \quad \text{(since $R_i \cap Q_\ell \subseteq \{(i, \ell)\}$)} \\
        &\leq |R_i \cap Q_1| - 1 + 1 \quad \text{(since $R_1 \subseteq P_1$, $(1, i) \in P_1$, so $(i, 1) \in P_1$)} \\
        &\leq |R_i \cap Q_1|.
    \end{align*}
    This completes the proof of inequality \eqref{ineq:set-sizes}. 

    Next, we observe that 
    \begin{align*}
        \sum_{\ell = 1}^k \left(g(R_i \cap Q_\ell) - g(R_i \cap P_\ell)\right) & = \phi_t\left(|R_i \cap Q_i|\right) - \phi_t\left(|R_i \cap P_i|\right) + \sum_{\ell \neq i} \left(\phi_n\left(|R_i \cap Q_\ell|\right) - \phi_n\left(|R_i \cap P_\ell|\right) \right) \\
        &= \sum_{\ell \neq i} \left(\phi_n\left(|R_i \cap Q_\ell|\right) - \phi_n\left(|R_i \cap P_\ell|\right) \right) \quad \quad \text{(since $R_i \cap P_i \cap U = \emptyset$)} \\ 
        &= \phi_n\left(|R_i \cap Q_1|\right) - \phi_n\left(|R_i \cap P_1|\right) + \sum_{\ell \notin \{1, i\}} \left(\phi_n\left(|R_i \cap Q_\ell|\right) - \phi_n\left(|R_i \cap P_\ell|\right)\right).
    \end{align*} 
    % \knote{Tag the relevant inequalities and cite them as needed.}
    We consider two cases based on $|R_i \cap Q_1|$. 
    
    \noindent \textbf{Case 1.} Suppose $|R_i \cap Q_1| \leq \frac{7}{8}k$. By inequality \eqref{ineq:set-sizes}, we have that $|R_i \cap P_\ell| \leq |R_i \cap Q_1| \leq \frac{7}{8}k$ for each $\ell \in [k] \setminus \{1, i\}$. Additionally, we have that $|R_i \cap P_1| \leq |R_i \cap Q_1| \leq \frac{7}{8}k$  since $P_1 \subseteq Q_1$. 
    Consequently, 
    \begin{align*}
        &\phin{R_i \cap Q_1} - \phin{R_i \cap P_1} + \sum_{\ell \notin \{1, i\}} \left(\phin{R_i \cap Q_\ell} - \phin{R_i \cap P_\ell}\right) \\
        & \quad = |R_i \cap Q_1| - |R_i \cap P_1| + \sum_{\ell \notin \{1, i\}} \left(|R_i \cap Q_\ell| - |R_i \cap P_\ell|\right) \\
        & \quad = 0. 
    \end{align*}

    \noindent \textbf{Case 2}. Suppose $|R_i \cap Q_1| > \frac{7}{8}k$. Since $|R_i \cap Q_1| > \frac{7}{8}k$, we have that $|R_i \cap Q_\ell| \leq \frac{7}{8}k$ for each $\ell \in [k] \setminus \{1, i\}$. Hence, 
    \[
    \phin{R_i \cap Q_\ell} = \begin{cases} \frac{7}{8}k & \ell = 1, \\ |R_i \cap Q_\ell| & \ell \notin \{1, i\}.\end{cases} 
    \]
    Therefore,
    \begin{align*}
        \sum_{\ell \neq i} \phin{R_i \cap Q_\ell} &= \frac{7}{8}k + \sum_{\ell \notin \{1, i\}} |R_i \cap Q_\ell| \\
        &= \frac{7}{8}k + \sum_{\ell \neq i} |R_i \cap Q_\ell| - |R_i \cap Q_1| \\
        &= \frac{7}{8}k + \sum_{\ell \neq i} |R_i \cap P_\ell| - |R_i \cap Q_1|.
    \end{align*}
    Now we consider $\sum_{\ell \neq i} \phin{R_i \cap P_\ell}$. Let $\ell^* = \argmax_{\ell \neq i} |R_i \cap P_\ell|$. Then, 
    \begin{align*}
        \sum_{\ell \neq i} \phin{R_i \cap P_\ell} &= \phin{R_i \cap P_{\ell^*}} + \sum_{\ell \notin \{\ell^*, i\}} |R_i \cap P_\ell| \quad \text{(since for each $\ell \notin \{\ell^*, i\}$, $|R_i \cap P_\ell| \leq \frac{7}{8}k$)} \\
        &= \phin{R_i \cap P_{\ell^*}} + \sum_{\ell \neq i} |R_i \cap P_\ell| - |R_i \cap P_{\ell^*}|.
    \end{align*}
    Therefore, 
    \[
        \sum_{\ell \neq i} \left(\phin{R_i \cap Q_\ell} - \phin{R_i \cap P_\ell}\right) = \frac{7}{8}k - \phin{R_i \cap P_{\ell^*}} - |R_i \cap Q_1| + |R_i \cap P_{\ell^*}|. 
    \]
    If $|R_i \cap P_{\ell^*}| \leq \frac{7}{8}k$, then 
    \begin{align*}
        \sum_{\ell \neq i} \left(\phin{R_i \cap Q_\ell} - \phin{R_i \cap P_\ell}\right) &= \frac{7}{8}k - |R_i \cap P_{\ell^*}| - |R_i \cap Q_1| + |R_i \cap P_{\ell^*}| \\
        &= \frac{7}{8}k - |R_i \cap Q_1| \\
        &< 0. \quad \text{(since $|R_i \cap Q_1| > \frac{7}{8}k$)}
    \end{align*}
    Otherwise, $|R_i \cap P_{\ell^*}| > \frac{7}{8}k$. Then,
    \begin{align*}
        \sum_{\ell \neq i} \left(\phin{R_i \cap Q_\ell} - \phin{R_i \cap P_\ell}\right) &= \frac{7}{8}k - \frac{7}{8}k - |R_i \cap Q_1| - |R_i \cap P_{\ell^*}| \\
        &= -|R_i \cap Q_1| + |R_i \cap P_{\ell^*}| \\
        &\leq 0,
    \end{align*}
    since $|R_i \cap P_\ell| \leq |R_i \cap Q_1|$ for each $\ell \in [k] \setminus \{1, i\}$ by inequality \eqref{ineq:set-sizes} and $|R_i \cap P_1| \leq |R_i \cap Q_1|$. 
    %Hence, in both cases, $\sum_{\ell = 1}^k \left(g(R_i \cap Q_\ell) - g(R_i \cap P_\ell)\right) \leq 0$. 
\end{proof}

We will use the following claim later to move certain elements to adjust sizes of certain parts in the symmetric partition without increasing the objective value. 
\begin{claim}
\label{claim:move_pair}
    %Let $V = [k] \times [k]$. 
    Let $P_1, \ldots, P_k$ be a symmetric multiway partition of $V$ %\knote{such that $R_1 \subseteq P_1$ and $U(P_1, \ldots, P_k)\subseteq P_1$. Let $(i, j) \in R_i \cap P_i$ where $i \neq j$. If $|R_j \cap P_j| \geq |R_i \cap P_i|$, then the symmetric partition $Q_i = P_i \setminus \{(i, j), (j, i)\}, Q_j = P_j \cup \{(i, j), (j, i)\}, Q_\ell = P_\ell$ for all $\ell \notin \{i, j\}$ satisfies $\obj{Q_1, \ldots, Q_k} \leq \obj{P_1, \ldots, P_k}$. } 
    with the following properties:
    \begin{enumerate}
        \item $R_1 \subseteq P_1$ and %\knote{Can't we remove this condition and just keep the second condition?}

        \item $U(P_1, \ldots, P_k)\subseteq P_1$. 
    \end{enumerate}
    Let $(i, j) \in R_i \cap P_i$ where $i \neq j$. If $|R_j \cap P_j| \geq |R_i \cap P_i|$, 
    %Consider $(i, j) \in R_i \cap P_i$ where $i \neq j$. If $|R_j \cap P_j| \geq |R_i \cap P_i|$, 
    then the symmetric multiway partition $Q_i := P_i \setminus \{(i, j), (j, i)\}, Q_j := P_j \cup \{(i, j), (j, i)\}, Q_\ell := P_\ell$ for all $\ell \notin \{i, j\}$ satisfies $\obj{Q_1, \ldots, Q_k} \leq \obj{P_1, \ldots, P_k}$. 
\end{claim}
\begin{proof}
    % \donote{From definition of unhappy vertices.}  
    We observe that $1 \notin \{i, j\}$ since $(j, i) \notin P_j$ implies that $k - 1 \geq |R_j \cap P_j| \geq |R_i \cap P_i|$ and $|R_1 \cap P_1| = k$. Let $a = |R_i \cap P_i|$ and $b = |R_j \cap P_j|$. 
    
    We have
    \begin{align*}
        &\phit{R_i \cap Q_i} - \phit{R_i \cap P_i} + \sum_{\ell \in [k] - i} \left(\phin{R_i \cap Q_\ell} - \phin{R_i \cap P_\ell}\right) \\
        &= \phit{R_i \cap Q_i} - \phit{R_i \cap P_i} + \phin{R_i \cap Q_j} - \phin{R_i \cap P_j} \quad \text{(since $P_\ell = Q_\ell$ for all $\ell \notin \{i, j\}$)} \\
        &= \phi_t(a - 1) - \phi_t(a) + 1 - 0 \quad \text{(since $P_j \cap U = \emptyset$ and $(i, j) \in P_i$ implies $|R_i \cap P_j| = 0$)} \\
        &= \phi_t(a - 1) - \phi_t(a) + 1.
    \end{align*}

    Additionally, 
    \begin{align*}
        &\phit{R_j \cap Q_j} - \phit{R_j \cap P_j} + \sum_{\ell \in [k] - j} \left(\phin{R_j \cap Q_\ell} - \phin{R_j \cap P_\ell}\right) \\
        &= \phit{R_j \cap Q_j} - \phit{R_j \cap P_j} + \phin{R_j \cap Q_i} - \phin{R_j \cap P_i} \quad \text{(since $P_\ell = Q_\ell$ for all $\ell \notin \{i, j\}$)} \\
        &= \phi_t(b + 1) - \phi_t(b) + 0 - 1 \quad \text{(since $P_i \cap U = \emptyset$ and $(j, i) \in P_i$ implies $|R_j \cap P_i| = 1$)} \\
        &= \phi_t(b + 1) - \phi_t(b) - 1.
    \end{align*}

    Then, we have 
    \begin{align*}
        &\obj{Q_1, \ldots, Q_k} - \obj{P_1, \ldots, P_k} \\
        &= \sum_{x = 1}^k \left(\phit{R_x \cap Q_x}  - \phit{R_x \cap P_x} + \sum_{\ell \in [k] - x} \left(\phin{R_x \cap Q_\ell} - \phin{R_x \cap P_\ell}\right)\right) \\
        &= \sum_{x \in \{i, j\}} \left(\phit{R_x \cap Q_x}  - \phit{R_x \cap P_x} + \sum_{\ell \in [k] - x} \left(\phin{R_x \cap Q_\ell} - \phin{R_x \cap P_\ell}\right)\right),
    \end{align*}
    since for all $x \notin \{i, j\}$ and $\ell \in [k]$, $R_x \cap Q_\ell = R_x \cap P_\ell$. 
    Thus, 
    \begin{align*}
        \obj{Q_1, \ldots, Q_k} - \obj{P_1, \ldots, P_k} &= \phi_t(a - 1) - \phi_t(a) + 1 + \phi_t(b + 1) - \phi_t(b) - 1 \\
        &= \phi_t(b + 1) - \phi_t(b) - (\phi_t(a) - \phi_t(a - 1)) \\
        &\leq 0 \quad \text{(since $\phi_t$ is concave and $a \leq b$)}. 
    \end{align*}
\end{proof}

For the rest of the transformations, we will use the following notion of \emph{swap} operation. 
\begin{definition}[swap operation]
\label{def:swap}
    For each $i\in \{2, 3, \ldots, k\}$, define $\pi_i:[k] \rightarrow [k]$ as follows:
    \[
        \pi_{i}(\ell) = \begin{cases} \ell & \text{if } \ell \notin \{i-1, i\}, \\ i-1 & \text{if } \ell = i, \\ i & \text{if } \ell = i-1. \end{cases}
    \]
    %Let $V = [k] \times [k]$. 
    For $i \in \{2, 3, \ldots, k\}$ and a symmetric multiway partition $P_1, \ldots, P_k$ of $V$, we define $\swapi{P_1, \ldots, P_k}$ as the symmetric multiway partition $Q_1, \ldots, Q_k$ of $V$ defined by 
    \begin{enumerate}
        \item $Q_{i - 1} = \left\{\left(\pi_{i}(x), \pi_{i}(y)\right) \colon (x, y) \in P_i \right\}$, 
        \item $Q_i = \left\{\left(\pi_{i}(x), \pi_{i}(y)\right) \colon (x, y) \in P_{i - 1} \right\}$, and 
        \item $Q_\ell = \left\{\left(\pi_{i}(x), \pi_{i}(y)\right) \colon (x, y) \in P_\ell \right\}$ for each $\ell \in [k]\setminus \{i-1, i\}$.
    \end{enumerate}
    Condensing this definition, we see that $Q_{\pi_i(\ell)} = \left\{ \left(\pi_i(x), \pi_i(y)\right) \colon (x, y) \in P_\ell \right\}$ for every $\ell\in [k]$. 
\end{definition}

\begin{figure}[ht]
\centering
\begin{minipage}[b]{0.49\textwidth}
    \centering
    \begin{tikzpicture}
    
        \definecolor{brightube}{rgb}{0.82, 0.62, 0.91}
        
        % Define grid size
        \def\gridsize{5}
        \def\spacing{0.9} % Distance between points
    
        % Draw the grid of vertices
        \foreach \x in {1,...,\gridsize} {
            \foreach \y in {1,...,\gridsize} {
                % Draw regular black vertices
                \filldraw[black] (\x*\spacing, \y*\spacing) circle (2pt);
            }
        }
        
        % P1
        \foreach \i in {0,...,4} {
            \foreach \j in {0} {
                \fill[red, opacity = 0.4] (0.5*\spacing + \i*\spacing, \gridsize*\spacing - 0.5*\spacing - \j*\spacing) rectangle (0.5*\spacing + \i*\spacing +\spacing, \gridsize*\spacing - 0.5*\spacing - \j*\spacing + \spacing);
            }
        }
    
        \foreach \i in {0} {
            \foreach \j in {1,...,4} {
                \fill[red, opacity = 0.4] (0.5*\spacing + \i*\spacing, \gridsize*\spacing - 0.5*\spacing - \j*\spacing) rectangle (0.5*\spacing + \i*\spacing +\spacing, \gridsize*\spacing - 0.5*\spacing - \j*\spacing + \spacing);
            }
        }
    
        % P2
        \foreach \i/\j in {1/1, 1/3, 3/1}{
            \fill[orange, opacity = 0.4] (0.5*\spacing + \i*\spacing, \gridsize*\spacing - 0.5*\spacing - \j*\spacing) rectangle (0.5*\spacing + \i*\spacing +\spacing, \gridsize*\spacing - 0.5*\spacing - \j*\spacing + \spacing);
        }
    
        % P3
        \foreach \i/\j in {1/2, 2/1, 2/2, 2/3, 2/4, 3/2, 4/2}{
            \fill[green, opacity = 0.4] (0.5*\spacing + \i*\spacing, \gridsize*\spacing - 0.5*\spacing - \j*\spacing) rectangle (0.5*\spacing + \i*\spacing +\spacing, \gridsize*\spacing - 0.5*\spacing - \j*\spacing + \spacing);
        }
    
        % P4
        \foreach \i/\j in {1/4, 3/3, 4/1}{
            \fill[blue, opacity = 0.4] (0.5*\spacing + \i*\spacing, \gridsize*\spacing - 0.5*\spacing - \j*\spacing) rectangle (0.5*\spacing + \i*\spacing +\spacing, \gridsize*\spacing - 0.5*\spacing - \j*\spacing + \spacing);
        }
    
        % P5
        \foreach \i/\j in {4/3, 4/4, 3/4}{
            \fill[brightube, opacity = 0.4] (0.5*\spacing + \i*\spacing, \gridsize*\spacing - 0.5*\spacing - \j*\spacing) rectangle (0.5*\spacing + \i*\spacing +\spacing, \gridsize*\spacing - 0.5*\spacing - \j*\spacing + \spacing);
        }
    \end{tikzpicture}
    \iffalse
    \captionsetup{singlelinecheck=off}
    \caption*{
        \begin{itemize}
            \item[] $P_1 = R_1 \cup C_1$

            \item[] $P_2 = \{(2, 2), (2, 4), (4, 2)\}$

            \item[] $P_3 = \{(2, 3), (3, 2), (3, 3), (3, 4), (3, 5), (4, 3), (5, 3)\}$

            \item[] $P_4 = \{(2, 5), (4, 4), (5, 2)\}$

            \item[] $P_5 = \{(4, 5), (5, 4), (5, 5)\}$
        \end{itemize}
    }
    \captionsetup{singlelinecheck=on}
    \fi
\end{minipage}
\begin{minipage}[b]{0.49\textwidth}
    \centering
    \begin{tikzpicture}
    
        \definecolor{brightube}{rgb}{0.82, 0.62, 0.91}
        
        % Define grid size
        \def\gridsize{5}
        \def\spacing{0.9} % Distance between points
    
        % Draw the grid of vertices
        \foreach \x in {1,...,\gridsize} {
            \foreach \y in {1,...,\gridsize} {
                % Draw regular black vertices
                \filldraw[black] (\x*\spacing, \y*\spacing) circle (2pt);
            }
        }
        
        % Q1
        \foreach \i in {0,...,4} {
            \foreach \j in {0} {
                \fill[red, opacity = 0.4] (0.5*\spacing + \i*\spacing, \gridsize*\spacing - 0.5*\spacing - \j*\spacing) rectangle (0.5*\spacing + \i*\spacing +\spacing, \gridsize*\spacing - 0.5*\spacing - \j*\spacing + \spacing);
            }
        }
    
        \foreach \i in {0} {
            \foreach \j in {1,...,4} {
                \fill[red, opacity = 0.4] (0.5*\spacing + \i*\spacing, \gridsize*\spacing - 0.5*\spacing - \j*\spacing) rectangle (0.5*\spacing + \i*\spacing +\spacing, \gridsize*\spacing - 0.5*\spacing - \j*\spacing + \spacing);
            }
        }
    
        % Q2
        \foreach \i/\j in {2/1, 1/2, 1/1, 1/3, 1/4, 3/1, 4/1}{
            \fill[green, opacity = 0.4] (0.5*\spacing + \i*\spacing, \gridsize*\spacing - 0.5*\spacing - \j*\spacing) rectangle (0.5*\spacing + \i*\spacing +\spacing, \gridsize*\spacing - 0.5*\spacing - \j*\spacing + \spacing);
        }
    
        % Q3
        \foreach \i/\j in {2/2, 2/3, 3/2}{
            \fill[orange, opacity = 0.4] (0.5*\spacing + \i*\spacing, \gridsize*\spacing - 0.5*\spacing - \j*\spacing) rectangle (0.5*\spacing + \i*\spacing +\spacing, \gridsize*\spacing - 0.5*\spacing - \j*\spacing + \spacing);
        }
    
        % Q4
        \foreach \i/\j in {2/4, 3/3, 4/2}{
            \fill[blue, opacity = 0.4] (0.5*\spacing + \i*\spacing, \gridsize*\spacing - 0.5*\spacing - \j*\spacing) rectangle (0.5*\spacing + \i*\spacing +\spacing, \gridsize*\spacing - 0.5*\spacing - \j*\spacing + \spacing);
        }
    
        % Q5
        \foreach \i/\j in {4/3, 4/4, 3/4}{
            \fill[brightube, opacity = 0.4] (0.5*\spacing + \i*\spacing, \gridsize*\spacing - 0.5*\spacing - \j*\spacing) rectangle (0.5*\spacing + \i*\spacing +\spacing, \gridsize*\spacing - 0.5*\spacing - \j*\spacing + \spacing);
        }
    \end{tikzpicture}
    \iffalse
    \captionsetup{singlelinecheck=off}
    \caption*{
        \begin{itemize}
            \item[] $Q_1 = R_1 \cup C_1$

            \item[] $Q_2 = \{(2, 2), (2, 3), (2, 4), (2, 5), (3, 2), (4, 2), (5, 2)\}$

            \item[] $Q_3 = \{(3, 3), (3, 4), (4, 3)\}$

            \item[] $Q_4 = \{(3, 5), (4, 4), (5, 3)\}$

            \item[] $Q_5 = \{(4, 5), (5, 4), (5, 5)\}$
        \end{itemize}
    }
    \captionsetup{singlelinecheck=on}
    \fi
\end{minipage}
\caption{An example of a swap in which $(Q_1, \ldots, Q_5) = \text{swap}_3(P_1, \ldots, P_5)$. On the left, $P_2$, $P_3$, and $P_4$ correspond to the orange, green, and blue part respectively. On the right, $Q_2$, $Q_3$, and $Q_4$ correspond to the green,  orange, and blue part respectively.}
\end{figure}
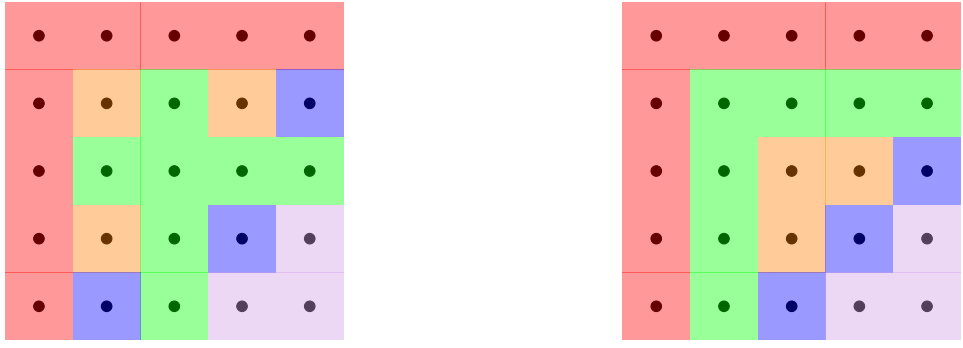

We now show that swap operations preserve certain structural properties without increasing the objective value. 
\begin{claim}
\label{claim:swap_properties}
    %Let $V = [k] \times [k]$. 
    Let $P_1, \ldots, P_k$ be a symmetric multiway partition of $V$. Let $Q_1, \ldots, Q_k = \swapi{P_1, \ldots, P_k}$ for some $i \in \{3,\ldots, k\}$. Then,
    \begin{enumerate}
        \item $Q_1, \ldots, Q_k$ is a symmetric multiway partition of $V$.
        \label{item:is_sym_mult_part}
        
        \item If $U(P_1, \ldots, P_k) \subseteq P_1$, then $U(Q_1, \ldots, Q_k) \subseteq Q_1$. 
        \label{item:swap_unhappy}
    
        \item 
        \label{item:swap_size}
        $\left|R_x \cap Q_\ell\right| = \left|R_{\pi_{i}(x)} \cap P_{\pi_{i}(\ell)}\right|$ for all $x, \ell \in [k]$. 

        \item $\obj{Q_1, \ldots, Q_k} = \obj{P_1, \ldots, P_k}$. 
        \label{item:swap_obj}

        \item If $\overleftarrow{R}_\ell \cap P_\ell = \emptyset$ and $\ell \in [k] \setminus \{i-1, i\}$, then $\overleftarrow{R}_\ell \cap Q_\ell = \emptyset$.
        \label{item:swap_preserve_empty}

        \item If $\overrightarrow{R}_\ell \subseteq P_1 \cup P_\ell$ and $\ell \in [k] \setminus \{i - 1, i\}$, then $\overrightarrow{R}_\ell \subseteq Q_1 \cup Q_\ell$. 
        \label{item:swap_preserve_subset}

    \end{enumerate}
\end{claim}
\begin{proof}
    First, we show property \ref{item:is_sym_mult_part}. We have $(x, y) \in Q_\ell$ if and only if $(\pi_i(x), \pi_i(y)) \in P_{\pi_i(\ell)}$. Using this fact, we have that $Q_1, \ldots, Q_k$ is a partition of $V$. We now show that $Q_1, \ldots, Q_k$ is symmetric: Suppose for contradiction that $Q_1, \ldots, Q_k$ is not symmetric. Then, there exists $\ell \in [k]$ and $(x, y) \in Q_\ell$ such that $(y, x) \notin Q_\ell$. Then, $(\pi_i(x), \pi_i(y)) \in P_{\pi_i(\ell)}$ and $(\pi_i(y), \pi_i(x)) \notin P_{\pi_i(\ell)}$, but this implies that $P_1, \ldots, P_k$ is not symmetric, which is a contradiction. Hence, $Q_1, \ldots, Q_k$ is symmetric. Finally, $Q_1, \ldots, Q_k$ satisfies $(\ell, \ell) \in Q_\ell$ for each $\ell \in [k]$: To see this, we observe that for each $\ell \in [k] \setminus \{i - 1, i\}$, we have $(\pi_i(\ell), \pi_i(\ell)) = (\ell, \ell) \in P_\ell = P_{\pi_i(\ell)}$, so $(\ell, \ell) \in Q_\ell$. Additionally, because $(i - 1, i - 1) \in P_{i - 1}$, we have $(i, i) = (\pi_i(i - 1), \pi_i(i - 1)) \in Q_{\pi_i(i - 1)} = Q_i$. Similar reasoning shows that $(i - 1, i - 1) \in Q_{i - 1}$. 
    
    Next, we show property \ref{item:swap_size}. We observe that $(x, y) \in Q_\ell$ if and only if $(\pi_i(x), \pi_i(y)) \in P_{\pi_i(\ell)}$. Hence, $|R_x \cap Q_\ell| = \sum_{y = 1}^k \mathbbm{1}\left[(x, y) \in Q_\ell\right] = \sum_{y = 1}^k \mathbbm{1}\left[\left(\pi_i(x), \pi_i(y)\right) \in P_{\pi_i(\ell)}\right] = \left|R_{\pi_i(x)} \cap P_{\pi_i(\ell)}\right|$ since $[k] = \{\pi_i(y)\}_{y \in [k]}$ as $\pi_i$ is a bijection. 

    Now, we show property \ref{item:swap_obj}. Consider $\phit{R_i \cap Q_i} + \sum_{\ell \neq x} \phin{R_x \cap Q_\ell}$. We have
    \begin{align*}
        &\phit{R_i \cap Q_i} + \sum_{\ell \neq i} \phin{R_x \cap Q_\ell} \\
        &= \phit{R_i \cap Q_i} + \phin{R_i \cap Q_{i -1}} + \sum_{\ell \notin \{i-1, i\}} \phin{R_i \cap Q_\ell} \\
        &= \phit{R_{i - 1} \cap P_{i - 1}} + \phin{R_{i - 1} \cap P_i} + \sum_{\ell \notin \{i - 1, i\}} \phin{R_{i - 1} \cap P_\ell} \quad \text{(by property \ref{item:swap_size})} \\
        &= \phit{R_{i - 1} \cap P_{i - 1}} + \sum_{\ell \neq i - 1} \phin{R_{i - 1} \cap P_\ell}.
    \end{align*}
    Similar reasoning shows that 
    \[ 
        \phit{R_{i - 1} \cap Q_{i - 1}} + \sum_{\ell \neq i - 1} \phin{R_{i - 1} \cap Q_\ell} = \phit{R_i \cap P_i} + \sum_{\ell \neq i}\phin{R_i \cap P_\ell}. 
    \]
    Finally, for $x \notin \{i - 1, i\}$, we have
    \begin{align*}
        &\phit{R_x \cap Q_x} + \sum_{\ell \neq x}\phin{R_x \cap Q_\ell} \\
        &= \phit{R_x \cap Q_x} + \phin{R_x \cap Q_{i - 1}} + \phin{R_x \cap Q_i} + \sum_{\ell \notin \{x, i - 1, i\}} \phin{R_x \cap Q_\ell} \\
        &= \phit{R_x \cap P_x} + \phin{R_x \cap P_i} + \phin{R_x \cap P_{i - 1}} + \sum_{\ell \notin \{x, i- 1, i\}} \phin{R_x \cap P_\ell} \quad \text{(by property \ref{item:swap_size})} \\
        &= \phit{R_x \cap P_x} + \sum_{\ell \neq x}\phin{R_x \cap P_\ell}. 
    \end{align*}
    Hence, 
    \begin{align*}
        \obj{Q_1, \ldots, Q_k} &= \sum_{x = 1}^k \left(\phit{R_x \cap Q_x} + \sum_{\ell \neq x} \phin{R_x \cap Q_\ell}\right) \\ 
        &= \sum_{x = 1}^k \left(\phit{R_x \cap P_x} + \sum_{\ell \neq x} \phin{R_x \cap P_\ell}\right) \\
        &= \obj{P_1, \ldots, P_k}. 
    \end{align*}

    Next, we prove property \ref{item:swap_unhappy}. For the sake of contradiction, suppose that property \ref{item:swap_unhappy} does not hold. Then, $U(P_1, \ldots, P_k) \subseteq P_1$ and there exists $(x, y) \in U(Q_1, \ldots, Q_k)$ such that $(x, y) \notin Q_1$. Then, $(x, y) \in Q_\ell$ for some $\ell \notin \{1, x, y\}$. Consequently, $\left(\pi_i(x), \pi_i(y)\right) \in P_{\pi_i(\ell)}$. If $\left(\pi_i(x), \pi_i(y)\right)$ is not unhappy under $P_1, \ldots, P_k$, then $\pi_i(\ell) = \pi_i(x)$ or $\pi_i(\ell) = \pi_i(y)$, which would imply that $x = \ell$ or $y = \ell$, a contradiction. Hence, $\left(\pi_i(x), \pi_i(y)\right)$ is unhappy under $P_1 \ldots, P_k$, so $\left(\pi_i(x), \pi_i(y)\right) \in P_1$. Then, $(x, y) \in Q_{\pi_i(1)} = Q_1$, which contradicts the assumption that $(x, y) \notin Q_1$. Thus, all unhappy vertices under $Q_1, \ldots, Q_k$ are in $Q_1$. 

    Next, we prove property \ref{item:swap_preserve_empty}. We consider two cases based on the value of $\ell$.
    \begin{enumerate}
        \item Suppose $\ell \in [i - 2]$. Suppose for contradiction that $\overleftarrow{R}_\ell \cap P_\ell = \emptyset$ and $\overleftarrow{R}_\ell \cap Q_\ell \neq \emptyset$. Then there exists $y \in [\ell - 1]$ such that $(\ell, y) \in Q_\ell$. Therefore, $(\ell, y) = (\pi_i(\ell), \pi_i(y)) \in P_{\pi_i(\ell)} = P_\ell$, which contradicts $\overleftarrow{R}_\ell \cap P_\ell = \emptyset$. 
        
        \item Suppose $\ell \in \{i + 1, \ldots, k\}$. Suppose for contradiction that $\overleftarrow{R}_\ell \cap P_\ell = \emptyset$ and $\overleftarrow{R}_\ell \cap Q_\ell \neq \emptyset$. Then there exists $y \in [\ell - 1]$ such that $(\ell, y) \in \overleftarrow{R}_\ell \cap Q_\ell$. Then, $(\ell, \pi_i(y)) = (\pi_i(\ell), \pi_i(y)) \in P_{\pi_i(\ell)} = P_\ell$. We observe that $\pi_i(y) = \begin{cases} i & y = i - 1 \\ i - 1 & y = i \\ y & y \notin\{i - 1, i\}\end{cases}$, and in all cases, $\pi_i(y) < \ell$. Hence, $(\ell, \pi_i(y)) \in \overleftarrow{R}_\ell \cap P_\ell$, which contradicts $\overleftarrow{R}_\ell \cap P_\ell = \emptyset$.
    \end{enumerate}

    Finally, we prove property \ref{item:swap_preserve_subset}. Again, we consider two cases based on the value of $\ell$. \begin{enumerate}
        \item Suppose $\ell \in [i - 2]$. Suppose for contradiction that $\overrightarrow{R}_\ell \subseteq P_1 \cup P_\ell$ and $\overrightarrow{R}_\ell \not\subseteq Q_1 \cup Q_\ell$. Then, there exists $y \in \{\ell + 1, \ldots, k\}$ and $\ell' \in [k] \setminus \{1, \ell\}$ such that $(\ell, y) \in Q_{\ell'}$. Then, $(\ell, \pi_i(y)) = (\pi_i(\ell), \pi_i(y)) \in P_{\pi_i(\ell')}$. We observe that $\pi_i(y) = \begin{cases} i - 1 & y = i \\ i & y = i - 1 \\ y & y \notin \{i - 1, i\}\end{cases}$, and $\pi_i(\ell') \in [k] \setminus \{1, \ell\}$. Hence, $(\ell, \pi_i(y)) \in \overrightarrow{R}_\ell \cap P_{\pi_i(\ell')}$ is a witnesses the fact that $\overrightarrow{R}_\ell \not\subseteq P_1 \cup P_\ell$, which is a contradiction. 
        
        \item Suppose $\ell \in \{i + 1, \ldots, k\}$. Suppose for contradiction that $\overrightarrow{R}_\ell \subseteq P_1 \cup P_\ell$ and $\overrightarrow{R}_\ell \not\subseteq Q_1 \cup Q_\ell$. Then there exists $y \in \{\ell + 1, \ldots, k\}$ and $\ell' \in [k] \setminus \{1, \ell\}$ such that $(\ell, y) \in Q_{\ell'}$. Then, $(\ell, y) = (\pi_i(\ell), \pi_i(y)) \in P_{\pi_i(\ell')}$ where $\pi_i(\ell') \in [k] \setminus \{1, \ell\}$. Then, $\overrightarrow{R}_\ell \not\subseteq P_1 \cup P_\ell$, which is a contradiction.
    \end{enumerate}
\end{proof}

\iffalse
\begin{claim}
    Let $V = [k] \times [k]$. Let $P_1, \ldots, P_k$ be a symmetric multiway partition of $V$ such that all unhappy vertices under $P_1, \ldots, P_k$ are in $P_1$. Let $Q_1, \ldots, Q_k = \swapi{P_1, \ldots, P_k}$. Then, all unhappy vertices under $Q_1, \ldots, Q_k$ are in $Q_1$. 
\end{claim}
\fi

Next, we show that swap operation can be used to enforce additional structure. 
\begin{claim}
\label{claim:swap_preserves_structure}
    %Let $V = [k] \times [k]$. 
    Let $P_1, \ldots, P_k$ be a symmetric multiway partition of $V$ and let $i\in \{3, \ldots, k\}$. 
    %\rnote{changed to $i \in \{3, \ldots, k\}$ since swaps only make sense if $i \leq k$}. 
    %Suppose that for some $i - 1 \geq 2$, 
    %$R_{i - 1}, R_i, \ldots, R_k$ satisfy 
    Suppose that the following properties hold for each $\ell \geq i - 1$:
    \begin{enumerate}
        \item $\overleftarrow{R}_\ell \cap P_\ell = \emptyset$ with the possible exception that $(i, i - 1) \in \overleftarrow{R}_i \cap P_i$;

        \item $\overrightarrow{R}_\ell \subseteq P_1 \cup P_\ell$ with the possible exception that $(i - 1, i) \in P_i$;

        \item $\{(i - 1, i), (i, i - 1)\} \subseteq P_i$ or $\{(i - 1, i), (i, i - 1)\} \subseteq P_1$.
    \end{enumerate}
    Let $(Q_1, \ldots, Q_k) = \swapi{P_1, \ldots, P_k}$. Then, we have that  $\overleftarrow{R}_\ell \cap Q_\ell = \emptyset$ and $\overrightarrow{R}_\ell \subseteq Q_1 \cup Q_\ell$
    for each $\ell \geq i - 1$. 
\end{claim}
\begin{proof}
    % \donote{From definition of swap operation.}  
    First, we show that for each $\ell \in \{i - 1, \ldots, k\}$, we have $\overleftarrow{R}_\ell \cap Q_\ell= \emptyset$. By Claim \ref{claim:swap_properties} property \ref{item:swap_preserve_empty}, this statement is true for each $\ell \in \{i + 1, \ldots, k\}$. Hence, suppose for contradiction there exists $j \in \{i - 1, i\}$ such that $\overleftarrow{R}_j \cap Q_j \neq \emptyset$. We case on the value $j$:
    \begin{enumerate}
        \item 
        If $j = i - 1$, then there exists $y \in [i - 2]$ such that $(i - 1, y) \in \overleftarrow{R}_{i - 1} \cap Q_{i - 1}$. Then $(i, \pi_i(y)) = (\pi_i(i - 1), \pi_i(y)) \in P_{\pi_i(i - 1)} = P_i$. This contradicts the fact that $\overleftarrow{R}_i \cap P_i \subseteq \{(i, i - 1)\}$. 

        \item 
        If $j = i$, then there exists $y \in [i - 1]$ such that $(i, y) \in \overleftarrow{R}_i \cap Q_i$. Since $(i - 1, i) \in P_i$ or $(i - i, i) \in P_1$, we have $(i, i - 1) = (\pi_i(i - 1), \pi_i(i)) \in Q_{\pi_i(i)} = Q_{i - 1}$ or $(i, i - 1) \in Q_{\pi_i(1)} = Q_1$. Hence, $y \neq i - 1$, so $y \in [i - 2]$. But then $(i - 1, y) = (\pi_i(i), \pi_i(y)) \in P_{\pi_i(i)} = P_{i - 1}$, which contradicts $\overleftarrow{R}_{i - 1} \cap P_{i - 1} = \emptyset$. 
    \end{enumerate}
    %Hence, we have that $\overleftarrow{R}_j \cap Q_j \neq \emptyset$ for each $j\in \{i-1, \ldots, k\}$. 
    %Overall, the assumption that there exists $j \in \{i - 1, \ldots, k\}$ such that $\overleftarrow{R}_j \cap Q_j \neq \emptyset$ is false. 

    Now, we show that for each $\ell \in \{i - 1, \ldots, k\}$, we have $\overrightarrow{R}_\ell \subseteq Q_1 \cup Q_\ell$. By Claim \ref{claim:swap_properties} property \ref{item:swap_preserve_subset}, this statement is true for each $\ell \in \{i + 1, \ldots, k\}$. Hence, suppose for contradiction there exists $j \in \{i - 1, i\}$ such that $\overrightarrow{R}_j \not\subseteq Q_1 \cup Q_j$. Again, we case on the value of $j$:
    \begin{enumerate}
        \item 
        If $j = i - 1$, then there exists $y \in \{i, \ldots, k\}$ and $\ell \in [k] \setminus \{1, i - 1\}$ such that $(i - 1, y) \in Q_\ell$. Since $(i, i - 1) \in P_i$ or $(i, i - 1) \in P_1$, we have $(i - 1, i) = (\pi_i(i), \pi_i(i - 1)) \in Q_{\pi_i(i - 1)} = Q_{i - 1}$ or $(i - 1, i) \in Q_{\pi_i(1)} = Q_1$. Therefore, $y \neq i$ and $y \in \{i + 1, \ldots, k\}$. Then, we have $(i, y) = (\pi_i(i - 1), \pi_i(y)) \in P_{\pi_i(\ell)}$ where $\pi_i(\ell) \in [k] \setminus \{1, i\}$. Since $y > i$, $\overrightarrow{R}_i \not\subseteq P_1 \cup P_i$, which is a contradiction. 

        \item 
        If $j = i$, then there exists $y \in \{i + 1, \ldots, k\}$ and $\ell \in [k] \setminus \{1, i\}$ such that $(i, y) \in Q_\ell$. Then, $(i - 1, y) = (\pi_i(i), \pi_i(y)) \in P_{\pi_i(\ell)}$ where $\pi_i(\ell) \in [k] \setminus \{1, i - 1\}$. But then since $y \geq i + 1 \geq i - 1$, we have a contradiction of $\overrightarrow{R}_{i - 1} \subseteq P_1 \cup P_{i - 1}$ with the possible exception that $(i - 1, i) \in P_i$. 
    \end{enumerate}
    %Hence, we have that for each $j \in \{i - 1, \ldots, k\}$, we have $\overrightarrow{R}_j \subseteq Q_1 \cup Q_j$.
\end{proof}

\begin{figure}[ht]
\centering
\begin{minipage}[b]{0.49\textwidth}
    \centering
    \begin{tikzpicture}
    
        \definecolor{brightube}{rgb}{0.82, 0.62, 0.91}
        
        % Define grid size
        \def\gridsize{6}
        \def\spacing{0.9} % Distance between points
    
        % Draw the grid of vertices
        \foreach \x in {1,...,\gridsize} {
            \foreach \y in {1,...,\gridsize} {
                % Draw regular black vertices
                \filldraw[black] (\x*\spacing, \y*\spacing) circle (2pt);
            }
        }
        
        % P1
        \foreach \i/\j in {0/0, 0/1, 0/2, 1/0, 1/3, 1/4, 1/5, 2/0, 2/5, 3/1, 3/4, 4/1, 4/3, 5/1, 5/2} {
            \fill[red, opacity = 0.4] (0.5*\spacing + \i*\spacing, \gridsize*\spacing - 0.5*\spacing - \j*\spacing) rectangle (0.5*\spacing + \i*\spacing +\spacing, \gridsize*\spacing - 0.5*\spacing - \j*\spacing + \spacing);
        }
    
        % P2
        \foreach \i/\j in {0/3, 0/4, 0/5, 1/1, 1/2, 2/1, 3/0, 4/0, 5/0}{
            \fill[orange, opacity = 0.4] (0.5*\spacing + \i*\spacing, \gridsize*\spacing - 0.5*\spacing - \j*\spacing) rectangle (0.5*\spacing + \i*\spacing +\spacing, \gridsize*\spacing - 0.5*\spacing - \j*\spacing + \spacing);
        }
    
        % P3
        \foreach \i/\j in {2/2, 2/4, 4/2}{
            \fill[green, opacity = 0.4] (0.5*\spacing + \i*\spacing, \gridsize*\spacing - 0.5*\spacing - \j*\spacing) rectangle (0.5*\spacing + \i*\spacing +\spacing, \gridsize*\spacing - 0.5*\spacing - \j*\spacing + \spacing);
        }
    
        % P4
        \foreach \i/\j in {2/3, 3/2, 3/3, 3/5, 5/3}{
            \fill[blue, opacity = 0.4] (0.5*\spacing + \i*\spacing, \gridsize*\spacing - 0.5*\spacing - \j*\spacing) rectangle (0.5*\spacing + \i*\spacing +\spacing, \gridsize*\spacing - 0.5*\spacing - \j*\spacing + \spacing);
        }
    
        % P5
        \foreach \i/\j in {4/4, 4/5, 5/4}{
            \fill[brightube, opacity = 0.4] (0.5*\spacing + \i*\spacing, \gridsize*\spacing - 0.5*\spacing - \j*\spacing) rectangle (0.5*\spacing + \i*\spacing +\spacing, \gridsize*\spacing - 0.5*\spacing - \j*\spacing + \spacing);
        }
    
        % P6
        \foreach \i/\j in {5/5}{
            \fill[gray, opacity = 0.4] (0.5*\spacing + \i*\spacing, \gridsize*\spacing - 0.5*\spacing - \j*\spacing) rectangle (0.5*\spacing + \i*\spacing +\spacing, \gridsize*\spacing - 0.5*\spacing - \j*\spacing + \spacing);
        }
    \end{tikzpicture}
    %\caption*{Visualization of $P_1, \ldots, P_6$.}
\end{minipage}
\begin{minipage}[b]{0.49\textwidth}
    \centering
    \begin{tikzpicture}
    
        \definecolor{brightube}{rgb}{0.82, 0.62, 0.91}
        
        % Define grid size
        \def\gridsize{6}
        \def\spacing{0.9} % Distance between points
    
        % Draw the grid of vertices
        \foreach \x in {1,...,\gridsize} {
            \foreach \y in {1,...,\gridsize} {
                % Draw regular black vertices
                \filldraw[black] (\x*\spacing, \y*\spacing) circle (2pt);
            }
        }
        
        % Q1
        \foreach \i/\j in {0/0, 0/1, 0/3, 1/0, 1/2, 1/4, 1/5, 3/0, 3/5, 2/1, 2/4, 4/1, 4/2, 5/1, 5/3} {
            \fill[red, opacity = 0.4] (0.5*\spacing + \i*\spacing, \gridsize*\spacing - 0.5*\spacing - \j*\spacing) rectangle (0.5*\spacing + \i*\spacing +\spacing, \gridsize*\spacing - 0.5*\spacing - \j*\spacing + \spacing);
        }
    
        % Q2
        \foreach \i/\j in {0/2, 0/4, 0/5, 1/1, 1/3, 3/1, 2/0, 4/0, 5/0}{
            \fill[orange, opacity = 0.4] (0.5*\spacing + \i*\spacing, \gridsize*\spacing - 0.5*\spacing - \j*\spacing) rectangle (0.5*\spacing + \i*\spacing +\spacing, \gridsize*\spacing - 0.5*\spacing - \j*\spacing + \spacing);
        }
    
        % Q4
        \foreach \i/\j in {3/3, 3/4, 4/3}{
            \fill[green, opacity = 0.4] (0.5*\spacing + \i*\spacing, \gridsize*\spacing - 0.5*\spacing - \j*\spacing) rectangle (0.5*\spacing + \i*\spacing +\spacing, \gridsize*\spacing - 0.5*\spacing - \j*\spacing + \spacing);
        }
    
        % Q3
        \foreach \i/\j in {3/2, 2/3, 2/2, 2/5, 5/2}{
            \fill[blue, opacity = 0.4] (0.5*\spacing + \i*\spacing, \gridsize*\spacing - 0.5*\spacing - \j*\spacing) rectangle (0.5*\spacing + \i*\spacing +\spacing, \gridsize*\spacing - 0.5*\spacing - \j*\spacing + \spacing);
        }
    
        % Q5
        \foreach \i/\j in {4/4, 4/5, 5/4}{
            \fill[brightube, opacity = 0.4] (0.5*\spacing + \i*\spacing, \gridsize*\spacing - 0.5*\spacing - \j*\spacing) rectangle (0.5*\spacing + \i*\spacing +\spacing, \gridsize*\spacing - 0.5*\spacing - \j*\spacing + \spacing);
        }
    
        % Q6
        \foreach \i/\j in {5/5}{
            \fill[gray, opacity = 0.4] (0.5*\spacing + \i*\spacing, \gridsize*\spacing - 0.5*\spacing - \j*\spacing) rectangle (0.5*\spacing + \i*\spacing +\spacing, \gridsize*\spacing - 0.5*\spacing - \j*\spacing + \spacing);
        }
    \end{tikzpicture}
    %\caption*{Visualization of $Q_1, \ldots, Q_6$.}
\end{minipage}
\caption{An example of applying Claim \ref{claim:swap_preserves_structure}. On the left is $P_1, \ldots, P_6$. On the right is $(Q_1, \ldots, Q_6) = \text{swap}_4(P_1, \ldots, P_6)$. On the left, $P_3$ and $P_4$ correspond to the green and blue part respectively. On the right, $Q_3$ and $Q_4$ correspond to the blue and green part respectively.}
\end{figure}
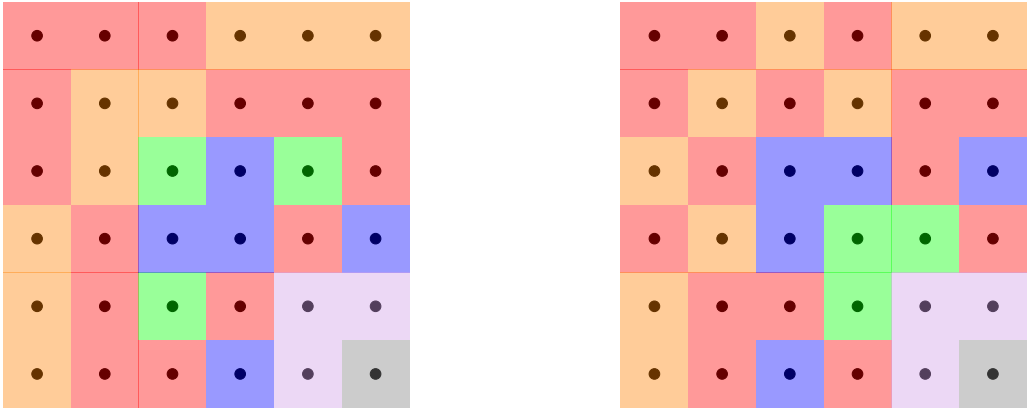

Our next claim will help in inductively achieving the required structure. We include an additional conclusion in the claim which will not be used beyond this claim, namely property 7. The purpose of this property is that it serves as a stronger induction hypothesis to prove the claim. 

\begin{claim}\label{claim:induction-size}
    %Let $V = [k] \times [k]$. 
    Let $P_1, \ldots, P_k$ be a symmetric multiway partition of $V$ and let $i \in \{3, \ldots, k + 1\}$. Suppose $P_1, \ldots, P_k$ additionally satisfies the following properties:
    \begin{enumerate}
        \item $R_1 \subseteq P_1$, 
        
        \item $U(P_1, \ldots, P_k) \subseteq P_1$, 

        \item $\overleftarrow{R}_\ell \cap P_\ell = \emptyset$ for each $\ell \in \{i - 1, i, \ldots, k\}$, 

        \item $\overrightarrow{R}_\ell \subseteq P_1 \cup P_\ell$ for each $\ell \in \{i - 1, i, \ldots, k\}$, and 

        \item $\left|R_i \cap P_i\right| \geq \left|R_{i + 1} \cap P_{i + 1}\right| \geq \ldots \geq \left|R_k \cap P_k\right|$. 
    \end{enumerate}
    Then, there exists a symmetric multiway partition $Q_1, \ldots, Q_k$ that additionally satisfies the following properties:
    \begin{enumerate}
        \item $R_1 \subseteq Q_1$, 
        
        \item $U(Q_1, \ldots, Q_k)\subseteq Q_1$, 

        \item $\overleftarrow{R}_\ell \cap Q_\ell = \emptyset$ for each $\ell \in \{i - 1, i, \ldots, k\}$, 

        \item $\overrightarrow{R}_\ell \subseteq Q_1 \cup Q_\ell$ for each $\ell \in \{i - 1, i, \ldots, k\}$, 

        \item $\left|R_{i - 1} \cap Q_{i - 1}\right| \geq \left|R_i \cap Q_i\right| \geq \ldots \geq \left|R_k \cap Q_k\right|$, 

        \item $\obj{Q_1, \ldots, Q_k} \leq \obj{P_1, \ldots, P_k}$, and

        \item $Q_1, \ldots, Q_k$ is obtained from $P_1, \ldots, P_k$ using only $\text{swap}_j$ operations and operations that modify the assignment of $\{(j, j - 1), (j - 1, j)\}$ where $j \in \{i \ldots, k\}$. 
        %\knote{Why do we need to state this property in the lemma statement? Where is it used?} \knote{Or is this statement needed as a stronger inductive hypothesis? } \rnote{Yes. This statement serves as an inductive hypothesis}
    \end{enumerate}
\end{claim}
\begin{proof}
    % \knote{Relies on Claims 2.8, 2.9, and 2.10.}
    If $i = k + 1$, then setting $Q_\ell := P_\ell$ for each $\ell \in [k]$ gives a symmetric multiway partition $Q_1, \ldots, Q_k$ of $V$ satisfying the required properties. We will henceforth assume that $i\le k$. Now, suppose for contradiction that the claim is false. Let $\mathcal{C}$ be the set of pairs $\left((P_1, \ldots, P_k), i\right)$, where $P_1, \ldots, P_k$ is a symmetric multiway partition of $V$ and $i \in \{3, \ldots, k + 1\}$, that serve as counterexamples to the claim. Let $t = \max\left\{i \colon \left((P_1, \ldots, P_k), i\right) \in \mathcal{C} \right\}$. Since the claim holds for $i = k + 1$, we have that $t \in \{3, \ldots, k\}$. Let 
    \[
        \left(P_1^{(0)}, \ldots, P_k^{(0)}\right) = \argmax_{(P_1, \ldots, P_k) \colon \left((P_1, \ldots, P_k), t\right) \in \mathcal{C}} |R_t \cap P_t|.
    \]
    In other words, among symmetric multiway partitions $P_1, \ldots, P_k$ that serve as a counterexample to the claim when paired with $t$, $\left|R_t \cap P_t^{(0)}\right|$ maximizes $|R_t \cap P_t|$. 

    Since $P_1^{(0)}, \ldots, P_k^{(0)}$ is a counterexample when paired with $t$, we must have 
    \begin{align}
        \left|R_{t - 1} \cap P_{t - 1}^{(0)} \right| < \left|R_t \cap P_t^{(0)} \right|. \label{ineq:row-intersection-sizes}   
    \end{align} 
    If not, then letting $Q_\ell := P_\ell^{(0)}$ for each $\ell \in [k]$ would satisfy the properties required of $Q_1, \ldots, Q_k$ in the claim. In the rest of this proof, we will perform a series of transformations to $P_1^{(0)}, \ldots, P_k^{(0)}$ to arrive at a partition $Q_1, \ldots, Q_k$ satisfying the conclusions of the claim. This contradicts the existence of a counterexample to the claim.
    
    If $\{(t - 1, t), (t, t - 1)\} \subseteq P_{t - 1}^{(0)}$, then we define $P_1^{(1)}, \ldots, P_k^{(1)}$ as $P_{t - 1}^{(1)} := P_{t - 1}^{(0)} \setminus \{(t - 1, t), (t, t - 1)\}$, $P_t^{(1)} := P_t^{(0)} \cup \{(t - 1, t), (t, t - 1)\}$, and $P_\ell^{(1)} := P_\ell^{(0)}$ for each $\ell \in [k] \setminus \{t - 1, t\}$. Otherwise, we let $P_\ell^{(1)} := P_\ell^{(0)}$ for each $\ell \in [k]$. By \eqref{ineq:row-intersection-sizes} and Claim \ref{claim:move_pair}, we have that $\obj{P_1^{(1)}, \ldots, P_k^{(1)}} \leq \obj{P_1^{(0)}, \ldots, P_k^{(0)}}$. 

    Next, let $\left(P_1^{(2)}, \ldots, P_k^{(2)}\right) = \text{swap}_t\left(P_1^{(1)}, \ldots, P_k^{(1)}\right)$. The symmetric multiway partition $P_1^{(2)}, \ldots, P_k^{(2)}$ of $V$ satisfies the following properties:
    \begin{enumerate}
        \item 
        We have $R_1 \subseteq P_1^{(2)}$ since $R_1 \subseteq P_1^{(1)}$ and by Claim \ref{claim:swap_properties} property \ref{item:swap_preserve_subset}. 

        \item 
        We have $U\left(P_1^{(2)}, \ldots, P_k^{(2)}\right) \subseteq P_1^{(2)}$ since $U\left(P_1^{(1)}, \ldots, P_k^{(1)}\right) \subseteq P_1^{(1)}$ and by Claim \ref{claim:swap_properties} property \ref{item:swap_unhappy}. 

        \item 
        For each $\ell \in \{t - 1, t, \ldots, k\}$, we have $\overleftarrow{R}_\ell \cap P_\ell^{(2)} = \emptyset$. This is because $P_1^{(1)}, \ldots, P_k^{(1)}$ and $i=t$ satisfy the hypotheses of Claim \ref{claim:swap_preserves_structure}. 

        \item 
        For each $\ell \in \{t - 1, t \ldots, k\}$, we have $\overrightarrow{R}_\ell \subseteq P_1^{(2)} \cup P_\ell^{(2)}$. Again, this is because $P_1^{(1)}, \ldots, P_k^{(1)}$ and $i=t$ satisfy the hypotheses of Claim \ref{claim:swap_preserves_structure}.

        \item 
        We have that $\left|R_{t + 1} \cap P_{t + 1}^{(2)}\right| \geq \left|R_{t + 2} \cap P_{t + 2}^{(2)}\right| \geq \ldots \geq \left|R_k \cap P_k^{(2)}\right|$: By Claim \ref{claim:swap_properties} property \ref{item:swap_size}, we have $\left|R_{\ell} \cap P_{\ell}^{(2)}\right| = \left|R_{\ell} \cap P_{\ell}^{(1)}\right|$ for each $\ell \in \{t + 1, \ldots, k\}$. Additionally, we have $\left|R_{\ell} \cap P_{\ell}^{(1)}\right| = \left|R_{\ell} \cap P_{\ell}^{(0)}\right|$ for each $\ell \in \{t + 1, \ldots, k\}$. Then the claimed property follows from the fact that $\left|R_{t + 1} \cap P_{t + 1}^{(0)} \right| \geq \left|R_{t + 2} \cap P_{t + 2}^{(0)}\right| \geq \ldots \geq \left|R_k \cap P_k^{(0)}\right|$ since $P_1^{(0)}, \ldots, P_k^{(0)}$ paired with $t$ is a counterexample to the claim. %Claim \ref{claim:induction-size}.

        \item 
        $\obj{P_1^{(2)}, \ldots, P_k^{(2)}} = \obj{P_1^{(1)}, \ldots, P_k^{(1)}} \leq \obj{P_1^{(0)}, \ldots, P_k^{(0)}}$: This is because of Claim \ref{claim:swap_properties} property \ref{item:swap_obj}.
    \end{enumerate}
    Since $t + 1 > t$, we have that $\left(\left(P_1^{(2)}, \ldots, P_k^{(2)}\right), t + 1\right) \notin \mathcal{C}$. Therefore, there exists a symmetric multiway partition $P^{(3)}_1, \ldots, P^{(3)}_k$ of $V$ that satisfies the following properties:
    \begin{enumerate}
        \item $R_1 \subseteq P^{(3)}_1$, 
        
        \item $U(P^{(3)}_1, \ldots, P^{(3)}_k)\subseteq P^{(3)}_1$, 

        \item $\overleftarrow{R}_\ell \cap P^{(3)}_\ell = \emptyset$ for each $\ell \in \{t, t + 1, \ldots, k\}$, 

        \item $\overrightarrow{R}_\ell \subseteq P^{(3)}_1 \cup P^{(3)}_\ell$ for each $\ell \in \{t, t + 1, \ldots, k\}$, 

        \item $\left|R_{t} \cap P^{(3)}_{t}\right| \geq \left|R_{t + 1} \cap P^{(3)}_{t + 1}\right| \geq \ldots \geq \left|R_k \cap P^{(3)}_k\right|$,

        \item $\obj{P^{(3)}_1, \ldots, P^{(3)}_k} \leq \obj{P_1^{(2)}, \ldots, P_k^{(2)}} \leq \obj{P_1^{(0)}, \ldots, P_k^{(0)}}$, and

        \item $P^{(3)}_1, \ldots, P^{(3)}_k$ can be obtained from $P_1^{(2)}, \ldots, P_k^{(2)}$ using only $\text{swap}_j$ operations and operations that modify the assignment of $\{(j, j - 1), (j - 1, j)\}$ where $j \in \{t + 1, \ldots, k\}$.
    \end{enumerate}

    Now, we show that $\overleftarrow{R}_{t - 1} \cap P^{(3)}_{t - 1} = \emptyset$:  We recall that $\overleftarrow{R}_{t - 1} \cap P_{t - 1}^{(2)} = \emptyset$. Additionally, $P^{(3)}_1, \ldots, P^{(3)}_k$ is obtained from $P_1^{(2)}, \ldots, P_k^{(2)}$ using only $\text{swap}_j$ operations and operations that change the part containing $\{(j, j - 1), (j - 1, j)\}$ where $j \in \{t + 1, \ldots, k\}$. Hence, we apply Claim \ref{claim:swap_properties} property \ref{item:swap_preserve_empty} in addition to the fact that modifying the assignment for $\{(j, j - 1), (j - 1, j)\}$ where $j \in \{t + 1, \ldots, k\}$ does not affect the assignment for elements in $R_{t - 1}$ to conclude that $\overleftarrow{R}_{t - 1} \cap P^{(3)}_{t - 1} = \emptyset$. 

    Next, using similar reasoning, we show that $\overrightarrow{R}_{t - 1} \subseteq P^{(3)}_1 \cup P^{(3)}_{t - 1}$: We recall that $\overrightarrow{R}_{t - 1} \subseteq P_1^{(2)} \cup P_{t - 1}^{(2)}$. Additionally, $P^{(3)}_1, \ldots, P^{(3)}_k$ is obtained from $P_1^{(2)}, \ldots, P_k^{(2)}$ using only $\text{swap}_j$ operations and operations that modify the assignment of $\{(j, j - 1), (j - 1, j)\}$ where $j \in \{t + 1, \ldots, k\}$. Hence, we apply Claim \ref{claim:swap_properties} property \ref{item:swap_preserve_subset} in addition to the fact that modifying the assignment for $\{(j, j - 1), (j - 1, j)\}$ where $j \in \{t + 1, \ldots, k\}$ does not affect the assignment for vertices in $R_{t - 1}$ to see that $\overrightarrow{R}_{t - 1} \subseteq P^{(3)}_1 \cup P^{(3)}_{t - 1}$. 

    Above, we have established that $P^{(3)}_1, \ldots, P^{(3)}_k$ and $t$ satisfy the hypotheses of Claim \ref{claim:induction-size}. Now, we demonstrate that $\left(\left(P^{(3)}_1, \ldots, P^{(3)}_k\right), t\right) \notin \mathcal{C}$. We have
    \begin{align}
        \left|R_t \cap P_t^{(0)} \right| &\leq \left|R_t \cap P_t^{(1)} \right| \quad \text{(since $P_t^{(0)} \subseteq P_t^{(1)}$)} \notag\\ 
        &= \left|R_{t - 1} \cap P_{t - 1}^{(2)} \right| \quad \text{(by Claim \ref{claim:swap_properties} property \ref{item:swap_size} applied to $\text{swap}_t$)} \notag \\
        &= \left|R_{t - 1} \cap P^{(3)}_{t - 1} \right|, \label{ineq:R_t-intersection-size}
    \end{align}
    where the last equality follows from repeated use of Claim \ref{claim:swap_properties} property \ref{item:swap_size} for $\text{swap}_j$ where $j \in \{t + 1, \ldots, k\}$ and the fact that changing the assignment for $\{(j, j - 1), (j - 1, j)\}$ where $j \in \{t + 1, \ldots, k\}$ does not affect the assignment for vertices in $R_{t - 1}$. Suppose for contradiction $\left(\left(P^{(3)}_1, \ldots, P^{(3)}_k\right), t\right) \in \mathcal{C}$. Then, $\left|R_{t - 1} \cap P^{(3)}_{t - 1}\right| < \left|R_t \cap P^{(3)}_t\right|$---otherwise, setting $Q_{\ell}:=P_{\ell}^{(3)}$ for each $\ell\in [k]$ would satisfy the properties required of $Q_1, \ldots, Q_k$ in the claim thereby implying that $\left(\left(P^{(3)}_1, \ldots, P^{(3)}_k\right), t\right) \not\in \mathcal{C}$. However, by \eqref{ineq:R_t-intersection-size}, we have that $\left|R_t \cap P_t^{(0)} \right| < \left|R_t \cap P^{(3)}_t\right|$, which contradicts the fact that
    \[
        \left|R_t \cap P_t^{(0)}\right| = \max_{(P_1, \ldots, P_k) \colon \left((P_1, \ldots, P_k), t\right) \in \mathcal{C}} |R_t \cap P_t|. 
    \]
    Hence, $\left(\left(P^{(3)}_1, \ldots, P^{(3)}_k\right), t\right) \notin \mathcal{C}$. 
    
    Furthermore, we observe that $P^{(3)}_1, \ldots, P^{(3)}_k$ was obtained from $P_1^{(0)}, \ldots, P_k^{(0)}$ using only $\text{swap}_j$ operations and operations that modify the assignment of $\{(j, j - 1), (j - 1, j)\}$ where $j \in \{t, \ldots, k\}$. 

    Finally, since $\left(\left(P^{(3)}_1, \ldots, P^{(3)}_k\right), t\right) \notin \mathcal{C}$ and $P^{(3)}_1, \ldots, P^{(3)}_k$ and $t$ satisfy the hypotheses of Claim \ref{claim:induction-size}, there exists a symmetric multiway partition $Q_1, \ldots, Q_k$ of $V$ such that 
    \begin{enumerate}
        \item $R_1 \subseteq Q_1$, 
        
        \item $U(Q_1, \ldots, Q_k)\subseteq Q_1$, 

        \item $\overleftarrow{R}_\ell \cap Q_\ell = \emptyset$ for each $\ell \in \{t - 1, t, \ldots, k\}$, 

        \item $\overrightarrow{R}_\ell \subseteq Q_1 \cup Q_\ell$ for each $\ell \in \{t - 1, t, \ldots, k\}$, 

        \item $\left|R_{t - 1} \cap Q_{t - 1}\right| \geq \left|R_t \cap Q_t\right| \geq \ldots \geq \left|R_k \cap Q_k\right|$, 

        \item $\obj{Q_1, \ldots, Q_k} \leq \obj{P_1^{(3)}, \ldots, P_k^{(3)}}$, and

        \item $Q_1, \ldots, Q_k$ is obtained from $P_1^{(3)}, \ldots, P_k^{(3)}$ using only $\text{swap}_j$ operations and operations that modify the assignment of $\{(j, j - 1), (j - 1, j)\}$ where $j \in \{t, \ldots, k\}$.
    \end{enumerate}
    Since, $\obj{P_1^{(3)}, \ldots, P_k^{(3)}} \leq \obj{P_1^{(0)}, \ldots, P_k^{(0)}}$, we have $\obj{Q_1, \ldots, Q_k} \leq \obj{P_1^{(0)}, \ldots, P_k^{(0)}}$. Additionally, $P^{(3)}_1, \ldots, P^{(3)}_k$ was obtained from $P_1^{(0)}, \ldots, P_k^{(0)}$ using only $\text{swap}_j$ operations and operations that modify the assignment of $\{(j, j - 1), (j - 1, j)\}$ where $j \in \{t, \ldots, k\}$. Therefore, $Q_1, \ldots, Q_k$ was obtained from $P_1^{(0)}, \ldots, P_k^{(0)}$ using only $\text{swap}_j$ operations and operations that change the part containing $\{(j, j - 1), (j - 1, j)\}$ where $j \in \{t, \ldots, k\}$. Hence, $Q_1, \ldots, Q_k$ satisfies the properties required by the claim. 
    %for Claim \ref{claim:induction-size}. 
    However, this shows that $P_1^{(0)}, \ldots, P_k^{(0)}$ paired with $t$ is not a counterexample to the claim. 
    %Claim \ref{claim:induction-size}, which contradicts $\left(\left(P_1^{(0)}, \ldots, P_k^{(0)}\right), t\right) \in \mathcal{C}$. Therefore, the original assumption that the Claim \ref{claim:induction-size} is false does not hold. 
\end{proof}

Our next claim will be applied immediately after the previous claim and will help in inductively achieving the required structure. 
\begin{claim}\label{claim:induction-containment}
    %Let $V = [k] \times [k]$. 
    Let $P_1, \ldots, P_k$ be a symmetric multiway partition of $V$ and let $i \in \{3, \ldots, k + 1\}$. Suppose $P_1, \ldots, P_k$ additionally satisfies the following properties:
    \begin{enumerate}
        \item $R_1 \subseteq P_1$, 
        
        \item $U(P_1, \ldots, P_k) \subseteq P_1$, 

        \item $\overleftarrow{R}_\ell \cap P_\ell = \emptyset$ for each $\ell \in \{i - 1, i, \ldots, k\}$, 

        \item $\overrightarrow{R}_\ell \subseteq P_1 \cup P_\ell$ for each $\ell \in \{i - 1, i, \ldots, k\}$, and 

        \item $\left|R_{i - 1} \cap P_{i - 1}\right| \geq \left|R_i \cap P_i\right| \geq \ldots \geq \left|R_k \cap P_k\right|$. 
    \end{enumerate}
    Then there exists a symmetric multiway partition $Q_1, \ldots, Q_k$ that additionally satisfies the following properties: \begin{enumerate}
        \item $R_1 \subseteq Q_1$;
        
        \item $U(Q_1, \ldots, Q_k) \subseteq Q_1$, 

        \item $\overleftarrow{R}_\ell \cap Q_\ell = \emptyset$ for each $\ell \in \{i - 2, i-1, \ldots, k\}$, 

        \item $\overrightarrow{R}_\ell \subseteq Q_1 \cup Q_\ell$ for each $\ell \in \{i - 2, i-1, \ldots, k\}$, 

        \item $\left|R_{i - 1} \cap Q_{i - 1}\right| \geq \left|R_i \cap Q_i\right| \geq \ldots \geq \left|R_k \cap Q_k\right|$, and 

        \item $\obj{Q_1, \ldots, Q_k} \leq \obj{P_1, \ldots, P_k}$.
    \end{enumerate}
\end{claim}
\begin{proof}
    % \knote{Relies on Claims 2.8 and 2.9?}
    If $i=3$, then the claim holds by setting $Q_\ell = P_\ell$ for each $\ell \in [k]$. We will henceforth assume that $i\ge 4$. 
    Let $j = \argmin_{\ell = 1}^{i - 2} |R_\ell \cap P_\ell|$. We observe that $j \geq 2$ since $R_1 \subseteq P_1$ implies $|R_1 \cap P_1| > |R_\ell \cap P_\ell|$ for all $\ell \in [k] - 1$. Let 
    \[ 
        (P_1', \ldots, P_k') = \text{swap}_{i - 2}\left(\text{swap}_{i - 3}\left(\ldots \text{swap}_{j + 1}\left(P_1, \ldots, P_k) \ldots \right)\right)\right). 
    \]
    Consider the bijection $\pi \colon [k] \rightarrow [k]$ defined as follows:
    \[
        \pi(t) = \begin{cases} t & t \leq j - 1 \text{ or } t \geq i - 1, \\
                    t - 1 & j + 1 \leq t \leq i - 2, \\
                    i - 2 & t = j.
                    \end{cases}
    \]
    We observe that $\pi = \pi_{i - 2} \circ \pi_{i - 3} \circ \ldots \circ \pi_{j + 1}$.
    By applying Claim \ref{claim:swap_properties} repeatedly, we observe that $P_1', \ldots, P_k'$ has the following properties:
    \begin{enumerate}
        \item $P_1', \ldots, P_k'$ is a symmetric multiway partition of $V$.

        \item $U(P_1', \ldots, P_k') \subseteq P_1'$.

        \item $\left|R_x \cap P_\ell'\right| = \left|R_{\pi^{-1}(x)} \cap P_{\pi^{-1}(\ell)}\right|$ for all $x, \ell\in [k]$.

        \item $\obj{P_1', \ldots, P_k'} \leq \obj{P_1, \ldots, P_k}$.

        \item $\overleftarrow{R}_\ell \cap P_\ell' = \emptyset$ for each $\ell \in \{i - 1, \ldots, k\}$. 

        \item $\overrightarrow{R}_\ell \subseteq P_1' \cup P_\ell'$ for each $\ell \in \{i - 1, \ldots, k\}$. 
    \end{enumerate} 
    We observe that $|R_1 \cap P_1'| = \left|R_{\pi^{-1}(1)} \cap P_{\pi^{-1}(1)}\right| = |R_1 \cap P_1| = k$, so $R_1 \subseteq P_1'$. Additionally, for each $\ell \in \{i - 1, \ldots, k\}$, we have $|R_\ell \cap P_\ell'| = \left|R_{\pi^{-1}(\ell)} \cap P_{\pi^{-1}(\ell)}\right| = |R_\ell \cap P_\ell|$. Hence, we have $|R_{i - 1} \cap P_{i - 1}'| \geq |R_i \cap P_i'| \geq \ldots \geq |R_k \cap P_k'|$. 

    We now show the following:  
    \begin{align}
        i - 2 = \argmin_{\ell = 1}^{i - 2} |R_\ell \cap P_\ell'|. \label{eq:i-2-smallest}
    \end{align}
        Since $\pi$ permutes the elements in $[i - 2]$, we have 
    \[  
        \left\{|R_\ell \cap P_\ell'| \colon \ell \in [i - 2]\right\} = \left\{\left|R_{\pi^{-1}(\ell)} \cap P_{\pi^{-1}(\ell)}\right| \colon \ell \in [i - 2]\right\} = \left\{|R_\ell \cap P_\ell| \colon \ell \in [i - 2]\right\}. 
    \]
    By the definition of $j$, we have $|R_j \cap P_j| = \min\left\{|R_\ell \cap P_\ell| \colon \ell \in [i - 2]\right\}$. Hence, 
    \begin{align*}
        \left|R_{i - 2} \cap P_{i - 2}'\right| &= |R_j \cap P_j| \\
        &= \min\left\{|R_\ell \cap P_\ell| \colon \ell \in [i - 2]\right\} \\
        &= \min\left\{|R_\ell \cap P_\ell'| \colon \ell \in [i - 2]\right\}, 
    \end{align*}
    demonstrating that $i - 2 = \argmin_{\ell = 1}^{i - 2} |R_\ell \cap P_\ell'|$. 

    Now, we obtain $Q_1, \ldots, Q_k$ using the following iterative procedure: For each $\ell \in [i - 3]$, if $(i - 2, \ell) \in P_{i - 2}'$, then 
    %\knote{define $Q_1, \ldots, Q_k$ as $Q_{i - 2} := P_{i - 2}' \setminus \{(i - 2, \ell), (\ell, i - 2)\}$, $Q_\ell := P_\ell' \cup \{(i - 2, \ell), (\ell, i - 2)\}$ and $Q_t:=P_t'$ for all $t \notin \{i - 2, \ell\}$. Otherwise, define $Q_t:=P_t'$ for all $t\in [k]$. Then, $Q_1, \ldots, Q_k$ is a symmetric multiway partition. We will show that $Q_1, \ldots, Q_k$ satisfies the properties required in the claim statement:}
    redefine $P_1', \ldots, P_k'$ such that $P_{i - 2}' := P_{i - 2}' \setminus \{(i - 2, \ell), (\ell, i - 2)\}$, $P_\ell' := P_\ell' \cup \{(i - 2, \ell), (\ell, i - 2)\}$ and $P_t'$ is unchanged for all $t \notin \{i - 2, \ell\}$. Let $Q_1, \ldots, Q_k$ be the parts obtained after the procedure is complete. Then, $Q_1, \ldots, Q_k$ is a symmetric multiway partition of $V$. We will show that $Q_1, \ldots, Q_k$ satisfies the properties mentioned in the claim statement:
    \begin{enumerate}
        \item 
        As established above, $R_1 \subseteq P_1'$. Since no vertices are removed from $P_1'$ during the procedure, we have $R_1 \subseteq Q_1$. 

        \item 
        Since $U(P_1', \ldots, P_k') \subseteq P_1'$ and no vertices in $U(P_1', \ldots, P_k')$ change parts during the procedure, we have $U(Q_1, \ldots, Q_k) \subseteq Q_1$.

        \item 
        Since $\overleftarrow{R}_\ell \cap P_\ell' = \emptyset$ for each $\ell \in \{i - 1, \ldots, k\}$ and the procedure only changes parts for vertices $(x, y)$ in which $x, y \in [i - 2]$, we have $\overleftarrow{R}_\ell \cap Q_\ell = \emptyset$ for each $\ell \in \{i - 1, \ldots, k\}$. Further, the procedure ensures that $\overleftarrow{R}_{i - 2} \cap Q_{i - 2} = \emptyset$. 

        \item 
        Since $\overrightarrow{R}_\ell \subseteq P_1' \cup P_\ell'$ for each $\ell \in \{i - 1, \ldots, k\}$ and the procedure only changes parts for vertices $(x, y)$ in which $x, y \in [i - 2]$, we have $\overrightarrow{R}_\ell \subseteq Q_1 \cup Q_\ell$ for each $\ell \in \{i - 1, \ldots, k\}$. Now, suppose for contradiction that $\overrightarrow{R}_{i - 2} \not\subseteq Q_1 \cup Q_{i - 2}$. Then, there exists $j \in \{i - 1, \ldots, k\}$ such that $(i - 2, j) \in Q_j$. However, $(j, i - 2) \in Q_j$ since $Q_1, \ldots, Q_k$ is symmetric, which contradicts $\overleftarrow{R}_j \cap Q_j = \emptyset$. 

        \item 
        As established above, we have $|R_{i - 1} \cap P_{i - 1}'| \geq |R_i \cap P_i'| \geq \ldots \geq |R_k \cap P_k'|$. Since the procedure only changes parts for vertices $(x, y)$ in which $x, y \in [i - 2]$, we have $R_\ell \cap P_\ell' = R_\ell \cap Q_\ell$ for each $\ell \in \{i - 1, \ldots, k\}$. Hence, $|R_{i - 1} \cap Q_{i - 1}| \geq |R_i \cap Q_i| \geq \ldots \geq |R_k \cap Q_k|$.

        \item We have that $\obj{P_1', \ldots, P_k'} \leq \obj{P_1, \ldots, P_k}$. Moreover, we have that $|R_{\ell}\cap P_{\ell}| \ge |R_{i-2}\cap P_{i-2}'|$ for every $\ell\in [i-3]$ by \eqref{eq:i-2-smallest} before the procedure and this inequality is preserved throughout the procedure. Hence, by Claim \ref{claim:move_pair}, we have that $\obj{P_1', \ldots, P_k'}$ does not increase during the procedure. Consequently, $\obj{Q_1 \ldots, Q_k} \leq \obj{P_1, \ldots, P_k}$. 
        %\knote{Claim \ref{claim:move_pair} requires the hypothesis $|R_j\cap P_j|\ge |R_i\cap P_i|$ for the objective to be non-increasing. Applying it here requires $|R_{\ell}\cap P_{\ell}'|\ge |R_{i-2}\cap P_{i-2}'|$ as the hypothesis. Why does this hypothesis hold?}
    \end{enumerate}
\end{proof}

The next claim follows by applying Claim \ref{claim:induction-size} followed by Claim \ref{claim:induction-containment}. 
\begin{claim}\label{claim:induction}
    %Let $V = [k] \times [k]$. 
    Let $P_1, \ldots, P_k$ be a symmetric multiway partition of $V$ and let $i \in \{3, \ldots, k + 1\}$. Suppose $P_1, \ldots, P_k$ additionally satisfies the following properties:
    \begin{enumerate}
        \item $R_1 \subseteq P_1$, 
        
        \item $U(P_1, \ldots, P_k)\subseteq P_1$, 

        \item $\overleftarrow{R}_\ell \cap P_\ell = \emptyset$ for each $\ell \in \{i - 1, i, \ldots, k\}$, 

        \item $\overrightarrow{R}_\ell \subseteq P_1 \cup P_\ell$ for each $\ell \in \{i - 1, i, \ldots, k\}$, and 

        \item $\left|R_i \cap P_i\right| \geq \left|R_{i + 1} \cap P_{i + 1}\right| \geq \ldots \geq \left|R_k \cap P_k\right|$. 
    \end{enumerate}
    Then there exists a symmetric multiway partition $Q_1, \ldots, Q_k$ of $V$ that additionally satisfies the following properties:
    \begin{enumerate}
        \item $R_1 \subseteq Q_1$, 
        
        \item $U(Q_1, \ldots, Q_k) \subseteq Q_1$, 

        \item $\overleftarrow{R}_\ell \cap Q_\ell = \emptyset$ for each $\ell \in \{i - 2, i-1, \ldots, k\}$, 

        \item $\overrightarrow{R}_\ell \subseteq Q_1 \cup Q_\ell$ for each $\ell \in \{i - 2, i-1, \ldots, k\}$, 

        \item $\left|R_{i - 1} \cap Q_{i - 1}\right| \geq \left|R_i \cap Q_i\right| \geq \ldots \geq \left|R_k \cap Q_k\right|$, and 

        \item $\obj{Q_1, \ldots, Q_k} \leq \obj{P_1, \ldots, P_k}$.
    \end{enumerate}
\end{claim}
%\begin{proof}
%    \knote{Relies on Claims 2.12, 2.13.}
%\end{proof}

We have the following corollary from Claim \ref{claim:induction} by induction. 
\begin{corollary}
\label{cor:partial_structure}
    %Let $V = [k] \times [k]$. 
    Let $P_1, \ldots, P_k$ be a symmetric multiway partition of $V$ such that $R_1 \subseteq P_1$ and $U(P_1, \ldots, P_k)\subseteq P_1$. Then, there exists a symmetric multiway partition $Q_1, \ldots, Q_k$ of $V$ such that the following properties hold:
    \begin{enumerate}
        \item $R_1 \subseteq Q_1$, 

        \item $\overrightarrow{R}_i \subseteq Q_1 \cup Q_i$ for each $i \in \{2, \ldots, k\}$, 

        \item $\left|R_1 \cap Q_1\right| \geq \left|R_2 \cap Q_2\right| \geq \ldots \geq \left|R_k \cap Q_k\right|$, and 

        \item $\obj{Q_1, \ldots, Q_k} \leq \obj{P_1, \ldots, P_k}$. 
    \end{enumerate}
\end{corollary}
%\begin{proof}
    %\knote{Relies on Claim 2.11.}
%\end{proof}

The next claim will help in inductively achieving the required structure. It helps achieve a certain structure for parts with large size without increasing the objective value. 
\begin{claim}
\label{claim:move_to_j}
Let $j\in [k]$ and $P_1, \ldots, P_k$ be a symmetric multiway partition of $V$ with the following properties:
\begin{enumerate}
    %\item $R_1 \subseteq P_1$, 
    \item $\overrightarrow{R}_i\subseteq P_i$ for each $i\in [j-1]$, 
    \label{item:suffix_contained_geq_1}
    \item $\overrightarrow{R}_i \subseteq P_1\cup P_i$ for each $i \in \{j, \ldots, k\}$, and 
    \label{item:suffix_contained_geq_2}
    \item $|R_j \cap P_j| \geq \frac{k}{2} + 2$. 
    \end{enumerate}
Then, there exists a symmetric multiway partition $Q_1, \ldots, Q_k$ of $V$ such that the following properties hold:
    \begin{enumerate}
        %\item $R_1\subseteq Q_1$, 
        \item $\overrightarrow{R}_i\subseteq Q_i$ for each $i\in [j]$, 
        \item $\overrightarrow{R}_i \subseteq Q_1\cup Q_i$ for each $i \in \{j+1, \ldots, k\}$, and 
        \item $\obj{Q_1, \ldots, Q_k} \leq \obj{P_1, \ldots, P_k}$
    \end{enumerate}
\end{claim}
\begin{proof}
% \knote{There is an issue with $Q_i$s not being symmetric since we describe movements of only rows and not columns. This issue persists for the rest of the section. We will edit it at the end with a  symmetrizing macro or some such alternative.}. 
    We define $\overrightarrow{C}_{t}:=\{(i,t): i\in \{t+1, t+2, \ldots, k\}\}$ for each $t\in [k]$. 
    We set $Q_1 := P_1 \setminus (\overrightarrow{R}_j \cup \overrightarrow{C}_j)$, $Q_j := P_j \cup (\overrightarrow{R}_j\cup \overrightarrow{C}_j)$, and $Q_\ell := P_\ell$ for all $\ell \in [k] \setminus \{1, j\}$. Then, $Q_1, \ldots, Q_k$ is a symmetric multiway partition of $V$ and the first two properties follow by definition. We now bound the objective value of $Q_1, \ldots, Q_k$. 

    We assume $j \neq 1$ since if $j = 1$, we have $Q_\ell = P_\ell$ for each $\ell \in [k]$ and the bound on the objective value follows immediately. We have 
    \begin{align*}
        &\obj{Q_1, \ldots, Q_k} - \obj{P_1, \ldots, P_k} \\
        &\quad \quad \quad \quad = \sum_{i = 1}^k \left(\phit{R_i \cap Q_i} - \phit{R_i \cap P_i} + \sum_{\ell \in [k] - i} \left(\phin{R_i \cap Q_\ell} - \phin{R_i \cap P_\ell}\right)\right) \\
        &\quad \quad \quad \quad = \sum_{i = j}^k \left(\phit{R_i \cap Q_i} - \phit{R_i \cap P_i} + \sum_{\ell \in [k] - i} \left(\phin{R_i \cap Q_\ell} - \phin{R_i \cap P_\ell}\right)\right),
    \end{align*}
    since for each $i \in [j - 1]$ and $\ell \in [k]$, $R_i \cap Q_\ell = R_i \cap P_\ell$. 

    We let $a = |\overrightarrow{R}_j \cap P_1|$, $b = |\overrightarrow{R}_j \cap P_j|$, $c = |\overleftarrow{R}_j \cap P_1|$, and $d = |\overleftarrow{R}_j \cap P_j|$. By assumption, we have $b + d + 1 \geq \frac{k}{2} + 2$.
    Then, 
    \begin{align*}
        &\phit{R_j \cap Q_j} - \phit{R_j \cap P_j} + \sum_{\ell \in [k] - j} \left(\phin{R_j \cap Q_\ell} - \phin{R_j \cap P_\ell}\right) \\
        &= \phit{R_j \cap Q_j} - \phit{R_j \cap P_j} + \phin{R_j \cap Q_1} - \phin{R_j \cap P_1} \quad \text{(since $Q_\ell = P_\ell$ for all $\ell \in [k] \setminus \{1, j\}$)} \\
        &= \phi_t(a + b + d + 1) - \phi_t(b + d + 1) + \phi_n(c) - \phi_n(a + c) \\
        &= \frac{7}{8}k - \frac{7}{8}k + \phi_n(c) - \phi_n(a + c) \\
        &= -a \quad \text{(since $|R_j \cap P_j| \geq \frac{k}{2} + 2$, we have $a + c = |R_j \cap P_1| \leq \frac{k}{2} - 2$)}. 
    \end{align*}
    
    Now fix $i \in \{j + 1, \ldots, k\}$. We have
    \begin{align*}
        &\phit{R_i \cap Q_i} - \phit{R_i \cap P_i} + \sum_{\ell \in [k] - i} \left(\phin{R_i \cap Q_\ell} - \phin{R_i \cap P_\ell}\right) \\
        &= \phin{R_i \cap Q_1} - \phin{R_i \cap P_1} + \phin{R_i \cap Q_j} - \phin{R_i \cap P_j} \\
        &= \mathbbm{1}\left[(i, j) \in P_1\right]\left(\phin{R_i \cap Q_1} - \phin{R_i \cap P_1} + \phin{R_i \cap Q_j} - \phin{R_i \cap P_j}\right),
    \end{align*}
    since if $(i, j) \in P_j$, $R_i \cap Q_\ell = R_i \cap P_\ell$ for each $\ell \in [k]$.  
    Therefore,
    \begin{align*}
        &\sum_{i \in \{j + 1, \ldots, k\}} \left(\phit{R_i \cap Q_i} - \phit{R_i \cap P_i} + \sum_{\ell \in [k] - i} \phin{R_i \cap Q_\ell} - \phin{R_i \cap P_\ell}\right) \\
        &= \sum_{i \in \{j + 1, \ldots, k\}} \mathbbm{1}\left[(i, j) \in P_1\right]\left(\phin{R_i \cap Q_1} - \phin{R_i \cap P_1} + \phin{R_i \cap Q_j} - \phin{R_i \cap P_j}\right) \\
        &= \sum_{i > j \colon (j, i) \in P_1} \left(\phin{R_i \cap Q_1} - \phin{R_i \cap P_1} + \phin{R_i \cap Q_j} - \phin{R_i \cap P_j}\right) \\
        &= \sum_{i > j \colon (j, i) \in P_1} \left(\phin{R_i \cap Q_1} - \phin{R_i \cap P_1} + 1 - 0\right)\\
        & \quad \quad \quad \quad \quad \quad \text{(by properties \ref{item:suffix_contained_geq_1} and \ref{item:suffix_contained_geq_2} and since $(i, j) \in P_1$, $R_i \cap P_j = \emptyset$)} \\
        &\leq \sum_{i > j \colon (j, i) \in P_1} 1 \quad \text{(since $\phi_n$ is monotonic and $R_i \cap Q_1 \subseteq R_i \cap P_1$)} \\
        &= a.
    \end{align*}
    Hence, $\obj{Q_1, \ldots, Q_k} - \obj{P_1, \ldots, P_k} \leq -a + a = 0$. 
\end{proof}

Our next claim builds on the previous claim and will help in inductively achieving the required structure. It helps achieve a certain structure for parts with small size without increasing the objective value. 
\begin{claim}
\label{claim:move_to_1}
Let $i^*, j\in [k]$ with $i^*<j$ and $P_1, \ldots, P_k$ be a symmetric multiway partition of $V$ with the following properties:
\begin{enumerate}
    %\item $R_1 \subseteq P_1$, 
    \item $\overrightarrow{R}_i\subseteq P_i$ for each $i\in [i^*]$, 
    \label{item:suffix_contained_leq_1}
    \item $\overrightarrow{R}_i\subseteq P_1$ for each $i\in \{i^* + 1, \ldots, j-1\}$, 
    \label{item:suffix_contained_leq_2}
    \item $\overrightarrow{R}_i \subseteq P_1\cup P_i$ for each $i \in \{j, \ldots, k\}$, and 
    \item $|R_j \cap P_j| \leq \frac{k}{2} + 1$. 
    \end{enumerate}
Then, there exists a symmetric multiway partition $Q_1, \ldots, Q_k$ of $V$ such that the following properties hold:
    \begin{enumerate}
        %\item $R_1\subseteq Q_1$, 
        \item $\overrightarrow{R}_i\subseteq Q_i$ for each $i\in [i^*]$, 
        \item $\overrightarrow{R}_i\subseteq Q_1$ for each $i\in \{i^* + 1, \ldots, j-1, j\}$, 
        \item $\overrightarrow{R}_i \subseteq Q_1\cup Q_i$ for each $i \in \{j+1, \ldots, k\}$, and 
        \item $\obj{Q_1, \ldots, Q_k} \leq \obj{P_1, \ldots, P_k}$
    \end{enumerate}
\end{claim}
\begin{proof}
We define $\overrightarrow{C}_{t}:=\{(i,t): i\in \{t+1, t+2, \ldots, k\}$ for each $t\in [k]$. 
We set $Q_1 := P_1 \cup (\overrightarrow{R}_j \cup \overrightarrow{C}_j)$, $Q_j := P_j \setminus (\overrightarrow{R}_j\cup \overrightarrow{C}_j)$, and $Q_\ell := P_\ell$ for all $\ell \in [k] \setminus \{1, j\}$. Then, $Q_1, \ldots, Q_k$ is a symmetric multiway partition of $V$ and the first three properties follow by definition. We now bound the objective value of $Q_1, \ldots, Q_k$. 

    We begin by observing that $j \neq 1$ since $R_1 \subseteq P_1$ implies that $|R_1 \cap P_1| \geq \frac{k}{2} + 2$. We have  
    \begin{align*}
        &\obj{Q_1, \ldots, Q_k} - \obj{P_1, \ldots, P_k} \\
        &\quad \quad \quad \quad = \sum_{i = 1}^k \left(\phit{R_i \cap Q_i} - \phit{R_i \cap P_i} + \sum_{\ell \in [k] - i} \left(\phin{R_i \cap Q_\ell} - \phin{R_i \cap P_\ell}\right)\right) \\
        &\quad \quad \quad \quad = \sum_{i = j}^k \left(\phit{R_i \cap Q_i} - \phit{R_i \cap P_i} + \sum_{\ell \in [k] - i} \left(\phin{R_i \cap Q_\ell} - \phin{R_i \cap P_\ell}\right)\right)
    \end{align*}
    since for each $i \in [j - 1]$ and $\ell \in [k]$, $R_i \cap Q_\ell = R_i \cap P_\ell$. 

    We let $a = |\overrightarrow{R}_j \cap P_1|$, $b = |\overrightarrow{R}_j \cap P_j|$, $c = |\overleftarrow{R}_j \cap P_1|$, and $d = |\overleftarrow{R}_j \cap P_j|$. By assumption, we have $b + d + 1 \leq \frac{k}{2} + 1$.
    Then,  
    \begin{align*}
        &\phit{R_j \cap Q_j} - \phit{R_j \cap P_j} + \sum_{\ell \in [k] - j} \left(\phin{R_j \cap Q_\ell} - \phin{R_j \cap P_\ell}\right) \\
        &= \phit{R_j \cap Q_j} - \phit{R_j \cap P_j} + \phin{R_j \cap Q_1} - \phin{R_j \cap P_1} \quad \text{(since $Q_\ell = P_\ell$ for all $\ell \in [k] \setminus \{1, j\}$)} \\
        &= \phi_t(d + 1) - \phi_t(b + d + 1) + \phi_n(a + b + c) - \phi_n(a + c) \\
        &\leq -b + b \quad \text{(since $b + d + 1 \leq \frac{k}{2} + 1$)} \\
        &= 0.
    \end{align*}

    Now, consider $i \in \{j + 1, \ldots, k\}$. If $(j, i) \in P_1$, then $R_i \cap Q_\ell = R_i \cap P_\ell$ for each $\ell \in [k]$, so $\phit{R_i \cap Q_i} - \phit{R_i \cap P_i} + \sum_{\ell \in [k] - i} \left(\phin{R_i \cap Q_\ell} - \phin{R_i \cap P_\ell}\right) = 0$. Hence, we assume $(j, i) \in P_j$. Then, 
    \begin{align*}
        &\phit{R_i \cap Q_i} - \phit{R_i \cap P_i} + \sum_{\ell \in [k] - i} \left(\phin{R_i \cap Q_\ell} - \phin{R_i \cap P_\ell}\right) \\
        &= \phin{R_i \cap Q_1} - \phin{R_i \cap P_1} + \phin{R_i \cap Q_j} - \phin{R_i \cap P_j} \\
        &\leq 1 + \phi_n(0) - \phi_n(1) \quad \text{(by properties \ref{item:suffix_contained_leq_1} and \ref{item:suffix_contained_leq_2}, $R_i \cap P_j = \{(i, j)\}$)} \\
        &= 1 - 1 \\
        &= 0.
    \end{align*}

    Hence, $\obj{Q_1, \ldots, Q_k} - \obj{P_1, \ldots, P_k} \leq 0$. 
\end{proof}

We have the following corollary from Claims \ref{claim:move_to_j} and \ref{claim:move_to_1}. 
\begin{corollary}
\label{cor:final_structure}
    %Let $V = [k] \times [k]$. 
    Let $P_1, \ldots, P_k$ be a symmetric multiway partition of $V$ with the following properties:
    \begin{enumerate}
        \item $R_1 \subseteq P_1$.

        \item $\overrightarrow{R}_i \subseteq P_1 \cup P_i$ for each $i \in \{2, \ldots, k\}$.

        \item $|R_1 \cap P_1| \geq |R_2 \cap P_2| \geq \ldots \geq |R_k \cap P_k|$.
    \end{enumerate}
    Then, there exists a symmetric multiway partition $Q_1, \ldots, Q_k$ of $V$ with the following properties:
    \begin{enumerate}
        \item 
        There exists $i^* \in [k]$ such that $\overrightarrow{R}_i \subseteq Q_i$ for each $i \in [i^*]$ and $\overrightarrow{R}_i \subseteq Q_1$ for each $i \in \{i^* + 1, \ldots, k\}$. 

        \item 
        $\obj{Q_1, \ldots, Q_k} \leq \obj{P_1, \ldots, P_k}$. 
    \end{enumerate}
\end{corollary}
\begin{proof}
    Since $R_1 \subseteq P_1$, we have that $k = |R_1 \cap P_1| \geq |R_2 \cap P_2| \geq \ldots \geq |R_k \cap P_k|\ge 1$. Since $k\ge 4$, there exists $i^* \in [k]$ such that $|R_i \cap P_i| \geq \frac{k}{2} + 2$ for each $i \in [i^*]$ and $|R_i \cap P_i| \leq \frac{k}{2} + 1$ for each $i \in \{i^* + 1, \ldots, k\}$. 
    Applying Claim \ref{claim:move_to_j} inductively for each $j\in [i^*]$ and then applying Claim \ref{claim:move_to_1} inductively for each $j\in \{i^*+1, \ldots, k\}$, we obtain the desired partition. 
    %We define the symmetric multiway partition $Q_1, \ldots, Q_k$ such that for each $i \in [i^*]$, $\overrightarrow{R}_i \subseteq Q_i$ and for each $i \in \{i^* + 1, \ldots, k\}$, $\overrightarrow{R}_i \subseteq Q_1$. By Claim \ref{claim:move_to_j} and Claim \ref{claim:move_to_1}, $\obj{Q_1, \ldots, Q_k} \leq \obj{P_1, \ldots, P_k}$. 
\end{proof}

We now restate and prove Lemma \ref{lem:sym_part_structure}. 
\lemSymPartStructure*
\begin{proof}
    % \donote{}
    Let us define $P_1^{(0)}, \ldots, P_k^{(0)}$ such that $P_\ell^{(0)} = P_\ell$ for each $\ell \in [k]$. We arrive at the required symmetric multiway partition by a series of transformations as described below.
    \begin{enumerate}
        \item 
        Apply Claim \ref{claim:move-first-row} to the symmetric multiway partition $P_1^{(0)}, \ldots, P_k^{(0)}$ of $V$ and obtain a symmetric multiway partition $P_1^{(1)}, \ldots, P_k^{(1)}$ of $V$. We have that 
        \[\obj{P_1^{(1)}, \ldots, P_k^{(1)}} \leq \obj{P_1^{(0)}, \ldots, P_k^{(0)}} + 2k = \obj{P_1, \ldots, P_k} + 2k.
        \]
        %\rnote{\[\obj{P_1^{(1)}, \ldots, P_k^{(1)}} \leq \obj{P_1^{(0)}, \ldots, P_k^{(0)}} + 2k = \obj{P_1, \ldots, P_k} + 2k.
        %\]}

        \item The symmetric multiway partition $P_1^{(1)}, \ldots, P_k^{(1)}$ of $V$ satisfies the hypothesis of Claim \ref{claim:move-unhappy}. 
        Apply Claim \ref{claim:move-unhappy} 
        to the symmetric multiway partition $P_1^{(1)}, \ldots, P_k^{(1)}$ of $V$ and obtain a symmetric multiway partition $P_1^{(2)}, \ldots, P_k^{(2)}$ of $V$ such that (1) $R_1 \subseteq P_1^{(2)}$ and (2) $U\left(P_1^{(2)}, \ldots, P_k^{(2)}\right) \subseteq P_1^{(2)}$. We also have that 
        \[\obj{P_1^{(2)}, \ldots, P_k^{(2)}} \leq \obj{P_1^{(1)}, \ldots, P_k^{(1)}} \leq \obj{P_1, \ldots, P_k} + 2k.\]

        \item The symmetric multiway partition $P_1^{(2)}, \ldots, P_k^{(2)}$ of $V$ satisfies the hypothesis of Corollary \ref{cor:partial_structure}. 
        Apply Corollary \ref{cor:partial_structure} 
        to the symmetric multiway partition $P_1^{(2)}, \ldots, P_k^{(2)}$ of $V$ and obtain a symmetric multiway partition $P_1^{(3)}, \ldots, P_k^{(3)}$ of $V$ such that 
        \begin{enumerate}[label=(\roman*)]
            \item $R_1 \subseteq P_1^{(3)}$, 

            \item $\overrightarrow{R}_i \subseteq P_1^{(3)} \cup P_i^{(3)}$ for each $i \in \{2, \ldots, k\}$, and 

            \item $\left|R_1 \cap P_1^{(3)}\right| \geq \left|R_2 \cap P_2^{(3)}\right| \geq \ldots \geq \left|R_k \cap P_k^{(3)}\right|$. 
        \end{enumerate}
        We also have that 
        \[\obj{P_1^{(3)}, \ldots, P_k^{(3)}} \leq \obj{P_1^{(2)}, \ldots, P_k^{(2)}} \leq \obj{P_1, \ldots, P_k} + 2k. \]

        \item The symmetric multiway partition $P_1^{(3)}, \ldots, P_k^{(3)}$ of $V$ satisfies the hypothesis of Corollary \ref{cor:final_structure}. 
        Apply Corollary \ref{cor:final_structure} 
        to the symmetric multiway partition $P_1^{(3)}, \ldots, P_k^{(3)}$ of $V$ and obtain a symmetric multiway partition $Q_1, \ldots, Q_k$ of $V$ such that 
        %obtain a symmetric multiway partition $Q_1, \ldots, Q_k$ such that 
        there exists $i^* \in [k]$ such that $\overrightarrow{R}_i \subseteq Q_i$ for each $i \in [i^*]$ and $\overrightarrow{R}_i \subseteq Q_1$ for each $i \in \{i^* + 1, \ldots, k\}$. We also have that 
        \[\obj{Q_1, \ldots, Q_k} \leq \obj{P_1^{(3)}, \ldots, P_k^{(3)}} \leq \obj{P_1, \ldots, P_k} + 2k. \]
    \end{enumerate}
    %Since $\obj{Q_1, \ldots, Q_k} = \frac{1}{2} \sum_{i = 1}^k f(Q_i)$ and $\obj{P_1, \ldots, P_k} = \frac{1}{2} \sum_{i = 1}^k f(P_i)$, we see that $Q_1, \ldots, Q_k$ satisfies the properties required for Lemma \ref{lem:sym_part_structure}.
\end{proof}

%% file: mono-sym-gap-to-query-lower-bound.tex
\subsubsection{Lower Bound from Symmetry Gap}
\label{subsec:sym_gap_to_query-complexity}

%Let $V=[k]\times [k]$ and $f: 2^V\rightarrow \R$ be a set function defined over the set $V$. For a set $A\subseteq V$, we denote $A^T:=\{(j,i): (i,j)\in A\}$. 
In this section, we %define the notion of symmetry gap and 
prove that the approximation factor of an algorithm that makes polynomial number of function evaluation queries should be at least the symmetry gap, i.e., prove Theorem \ref{thm:sym-gap-to-oracle-lower-bound}. Our proof approach is similar to the proof approaches in \cite{Vondrak, EVW}, but we have to work out the details since we are interested in \monosubmp which was not the focus of those two works. 
%which does not follow immediately from the results in \cite{Vondrak, EVW}. 
%We encourage first-time readers to skip this section. 
We begin by showing that \opt and \symopt can be obtained as optimum solutions to the multilinear relaxation. We define the multilinear extension and the multilinear relaxation of the multiway partitioning problem. 

\begin{definition}
Let $V=[k]\times [k]$ and $f: 2^V\rightarrow \R$ be a set function. 
%Let $U$ be a ground set and $h:2^U\rightarrow \R$ be a set function. 
    \begin{enumerate}
    \item The \emph{multilinear extension} of $f$ is the function $F \colon [0, 1]^V \rightarrow \mathbb{R}$ defined by
        \[
        F(\vx) = \sum_{S \subseteq V} f(S) \left(\prod_{v \in S} x_v \prod_{v \in V \setminus S} (1 - x_v)\right). 
        \]
    \iffalse
        \item The \emph{multilinear extension} of $h$ is the function $H \colon [0, 1]^V \rightarrow \mathbb{R}$ defined by
        \[
        H(\vx) = \sum_{S \subseteq V} h(S) \prod_{v \in S} x_v \prod_{v \in V \setminus S} (1 - x_v). 
        \]
        \fi
    \iffalse
    \item The \emph{multilinear relaxation} of the multiway partitioning problem over $V$ for a pairwise disjoint collection of terminal sets $T_1, \ldots, T_k\subseteq U$ is 
    \begin{align*}
        \text{min} \sum_{\ell = 1}^k H\left(\vx_\ell\right) & \\
        \text{subject to } \sum_{\ell = 1}^k x(u, \ell) = 1 & \quad \forall u \in U \\
        x(u, \ell) = 1 & \quad \forall u\in T_{\ell},\ \ell \in [k] \\
        x(u, \ell) \geq 0 & \quad \forall u \in U, \ell \in [k], 
    \end{align*}
    where $\vx_\ell = \left(x(u, \ell)\right)_{u \in U}$. We will say that a vector $\vx = \left(x(u, \ell) \right)_{u \in U, \ell \in [k]} \in [0, 1]^{U \times [k]}$ is a solution to the multilinear relaxation if it satisfies the constraints of the multilinear relaxation.
    \fi
    %\iffalse
    \item The \emph{multilinear relaxation} of the multiway partitioning problem over $V$ wrt $f$ for the set $\{(\ell,\ell): \ell\in [k]\}$ of terminals is 
    \begin{align*}
        \text{min} \sum_{\ell = 1}^k F\left(\vx_\ell\right) & \\
        \text{subject to } \sum_{\ell = 1}^k x((i, j), \ell) = 1 & \quad \forall (i, j) \in V \\
        x((\ell, \ell), \ell) = 1 & \quad \forall \ell \in [k] \\
        x((i, j), \ell) \geq 0 & \quad \forall (i, j) \in V, \ell \in [k], 
    \end{align*}
    where $\vx_\ell = \left(x((i, j), \ell)\right)_{(i, j) \in V}$. We will say that a vector $\vx = \left(x((i, j), \ell) \right)_{(i, j) \in V, \ell \in [k]} \in [0, 1]^{V \times [k]}$ is a solution to the multilinear relaxation if it satisfies the constraints of the multilinear relaxation.
    %\fi
    \end{enumerate}
\end{definition}

The following lemma shows that both \opt and \symopt can be expressed as optimum solutions to the multilinear relaxation and a symmetry constrained multilinear relaxation respectively. 
\begin{lemma}
\label{lem:OPT_equality}
    %Let $f_{\phi_n, \phi_t} \colon 2^V \rightarrow \mathbb{R}$ be defined as in Definition \ref{def:instance} and let $F_{\phi_n, \phi_t} \colon [0, 1]^V \rightarrow \mathbb{R}$ be the multilinear extension of $f_{\phi_n, \phi_t}$. Let $OPT$ and $\overline{OPT}$ be defined as in Definition \ref{def:OPT_def}. Then, 
    Let $V=[k]\times [k]$ and let $f:2^V\rightarrow \R$ be a row-column-type set function defined over the set $V$. 
    We have the following: 
    \begin{enumerate}
        \item $\opt(f) = \min\left\{\sum_{\ell = 1}^k F(\vx_\ell): \left(\vx_\ell \right)_{\ell \in [k]} \text{ is a feasible solution to the multilinear relaxation} \right\}$  and 
        \item %$\symopt(f) = \min\left\{\sum_{\ell = 1}^k F\left(\overline{\vx}_\ell\right): \left(\overline{\vx}_\ell \right)_{\ell \in [k]} \text{ is a symmetric feasible solution to the multilinear relaxation}\right\}$.
        $\symopt(f) = \min\left\{\sum_{\ell = 1}^k F\left(\vy_\ell\right): \left(\vy_\ell \right)_{\ell \in [k]} \text{ is a symmetric feasible solution to the multilinear relaxation}\right\}$.
    \end{enumerate}    
\end{lemma}
\begin{proof}
%\donote{Proof needs to be edited to adhere to the lemma statement.}  
For simplicity, we denote $\opt = \opt(f)$ and $\symopt = \symopt(f)$. 
\begin{enumerate}
    \item Let $\mlopt = \min\{\sum_{\ell = 1}^k F(\vx_\ell): \left(\vx_\ell \right)_{\ell \in [k]} \text{ is a feasble solution to the multilinear relaxation} \}$. We show that $\opt = \mlopt$. Let $V_1, \ldots, V_k$ be a partition of $V$ such that $(i, i) \in V_i$ for each $i \in [k]$ and $\sum_{i = 1}^k f(V_i) = OPT$. For every $i, j, \ell\in [k]$, we define 
    \[
    x((i, j), \ell) = \begin{cases} 1 & \text{ if } (i, j) \in V_\ell, \\ 0 & \text{ if } (i, j) \notin V_\ell. \end{cases}
    \]
    Then, $\left(x((i, j), \ell)\right)_{(i, j) \in V, \ell \in [k]}$ is a feasible solution to the multilinear relaxation and its objective value is $\sum_{\ell = 1}^k F\left(\left(x((i, j), \ell) \right)_{(i, j) \in V}\right) = \opt$. Consequently, $\mlopt\le \opt$. Next, let $\vx^*$ be an optimum solution to the multilinear relaxation. 
    We obtain a partition $V^*_1, \ldots, V^*_k$ of $V$ by assigning each element $(i,j)\in V$ to one of the parts as follows: consider an element $(i, j)\in V$; 
    we define $z_0:=0$ and $z_{\ell}:=\sum_{m=1}^{\ell}x^*((i,j),m)$ for every $\ell\in [k]$, pick $\theta\in [0,1]$ uniformly at random, and assign $(i,j)$ to $V_{\ell}$ such that $z_{\ell-1}\le \theta < z_{\ell}$. We note that each element $(i,j)$ is assigned to a part $V^*_{\ell}$ independently with probability $x^*((i,j), \ell)$. 
    %we assign it to $V^*_{\ell}$ independently such that the probability $(i, j)$ is assigned to $V^*_\ell$ is $x^*((i, j), \ell)$ \knote{Need to explain better}. 
    We observe that $V^*_1, \ldots, V^*_k$ is a partition of $V$ with $(i, i) \in V^*_i$ for every $i\in [k]$. Moreover, $\E \left[ f(V^*_\ell) \right] = \sum_{S \subseteq V} f(S) \prod_{(i, j) \in S} x^*((i, j), \ell) \prod_{(i, j) \in V \setminus S} (1 - x^*((i, j), \ell)) = F(\vx^*_\ell)$. Hence, $\E \left[ \sum_{\ell = 1}^k f(V^*_\ell) \right] = \sum_{\ell = 1}^k F(\vx^*_\ell) = \mlopt$. This implies that there exists a partition $V_1, \ldots, V_k$ with $(i, i) \in V_i$ for each $i \in [k]$ such that $\sum_{\ell = 1}^k f(V_\ell) \leq \mlopt$, so $\opt \le \mlopt$. 

    \item Let $\mlsymopt = \min\left\{\sum_{\ell = 1}^k F\left(\overline{\vx}_\ell\right): \left(\overline{\vx}_\ell \right)_{\ell \in [k]} \text{ is a symmetric solution to the multilinear relaxation}\right\}$. We show that $\symopt = \mlsymopt$. Using identical reasoning to the previous paragraph, we have that $\mlsymopt \leq \symopt$. Let $\vy$ be an optimum symmetric solution to the multilinear relaxation. We obtain a partition $V^*_1, \ldots, V^*_k$ of $V$ in the same manner as the previous paragraph from $\vy$. 
    By the reasoning of the above paragraph, we have $\E \left[\sum_{\ell = 1}^k f(V^*_\ell)\right] = \mlsymopt$. 
    We recall that $\vy$ is a symmetric feasible solution to the multilinear relaxation but the the partition $V^*_1, \ldots, V^*_k$ need not be symmetric. 
    In addition to the partition $V^*_1, \ldots, V^*_k$, we obtain a symmetric partition $\overline{V}_1, \ldots, \overline{V}_k$ of $V$ by assigning each pair $\{(i,j), (j,i)\}\subseteq V$ to one of the parts as follows: for $i,j\in [k]$ consider the pair of elements $\{(i,j), (j,i)\}\in V$; 
     we define $z_0:=0$ and $z_{\ell}:=\sum_{m=1}^{\ell}y((i,j),m)$ for every $\ell\in [k]$, pick $\theta\in [0,1]$ uniformly at random, and assign the pair of elements $\{(i,j), (j,i)\}$ to $V_{\ell}$ such that $z_{\ell-1}\le \theta < z_{\ell}$. We note that each element pair $\{(i,j), (j,i)\}$ is assigned to a part $\overline{V}_{\ell}$ independently with probability $y((i,j), \ell)$. 
    %by independently assigning each set $\{(i, j), (j, i)\}$ such that the probability $\{(i, j), (j, i)\} \subseteq \overline{V}_\ell$ is $\overline{x}((i, j), \ell) = \overline{x}((j, i), \ell)$. We observe that $V^*_1, \ldots, V^*_k$ is not necessarily a symmetric partition of $V$, but $\overline{V}_1, \ldots, \overline{V}_k$ is a symmetric partition of $V$. 
    We observe that $\overline{V}_1, \ldots, \overline{V}_k$ is a partition of $V$ with  $\overline{V}_i$ being symmetric and $(i, i)\in \overline{V}_i$ for every $i\in [k]$. We now show that $\E\left[f(\overline{V}_{\ell})\right]\le \E \left[ f(V^*_\ell)\right]$ for every $\ell\in [k]$. 
    
    %By the reasoning of the above paragraph, we have $\E \left[\sum_{\ell = 1}^k f(V^*_\ell)\right] = \mlsymopt$. 
    \begin{claim}\label{claim:symmetric-rounding-atmost-rounding}
        For every $\ell\in [k]$, we have that 
        \[\E\left[f(\overline{V}_{\ell})\right]\le \E \left[f(V^*_\ell)\right].\]
    \end{claim}
    \begin{proof}
    %\knote{Proof depends on row-column definition of $f$. Can we get rid of the dependence on this row-column definition and just rely on the fact that $f$ is transpose-invariant and $y_\ell$ is symmetric?}
    Let $\ell\in [k]$. 
    We recall that $V=[k]\times [k]$ and $f:2^V\rightarrow \R_{\ge 0}$ is a row-column-type function. %\knote{Define row-column type function.}. 
    Hence, there exists a function $g:2^{V}\rightarrow \R_{\ge 0}$ such that $f(A)=\sum_{i=1}^k (g(A\cap R_i) + g(A\cap C_i))$ for every $A\subseteq V$, where $R_i:=\{(i, j): j\in [k]\}$ and $C_i:=\{(j,i): j\in [k]\}$ for every $i\in [k]$. Hence, for every $i\in [k]$, we have that 
    \begin{align*}
        \E \left[ g\left(\overline{V}_\ell \cap R_i\right) \right] &= \sum_{S \subseteq R_i} g(S) \Pr \left[ \overline{V}_\ell \cap R_i = S\right] \\
        &= \sum_{S \subseteq R_i} g(S) \prod_{(i, j) \in S} y((i, j), \ell) \prod_{(i, j) \in R_i \setminus S} (1 - y((i, j), \ell)) \\%\quad \quad \text{(since each $(i, j) \in R_i$ is assigned independently)} \\
        &= \sum_{S \subseteq R_i} g(S) \Pr \left[ V^*_\ell \cap R_i = S\right] \\
        &= \E \left[ g\left(V^*_\ell \cap R_i\right) \right]. 
    \end{align*}
    Similar reasoning shows that $\E \left[ g\left(\overline{V}_\ell \cap C_i\right) \right] = \E \left[ g(V^*_\ell \cap C_i) \right]$ for every $i\in [k]$. 
    Hence, we have that 
    \begin{align*}
        \E \left[ f\left(\overline{V}_\ell\right) \right] &= \sum_{i = 1}^k \left(\E \left[ g\left(\overline{V}_\ell \cap R_i\right)\right] + \E \left[ g\left(\overline{V}_\ell \cap C_i\right) \right]\right)\\ 
        &= \sum_{i = 1}^k \left(\E \left[ g(V^*_\ell \cap R_i) \right] + \E \left[ g(V^*_\ell \cap C_i) \right]\right) \\
        &= \E \left[ f(V^*_\ell) \right].
    \end{align*}
    \end{proof}
    Claim \ref{claim:symmetric-rounding-atmost-rounding} implies that there exists a partition $\overline{V}_1, \ldots, \overline{V}_k$ of $V$ with  $\overline{V}_i$ being symmetric and $(i, i)\in \overline{V}_i$ for every $i\in [k]$ such that $\sum_{i=1}^{\ell}f(\overline{V}_{\ell})\le \E \left[\sum_{\ell = 1}^k f(V^*_\ell)\right] = \mlsymopt$ and hence, $\symopt\le \mlsymopt$. 
\end{enumerate}
\end{proof}

We will use the following three lemmas from \cite{Vondrak} to prove Theorem \ref{thm:sym-gap-to-oracle-lower-bound}. 

%\knote{Cite Vondrak's journal paper instead of EVW.}

\begin{lemma}[Lemma 3.1 from \cite{Vondrak}]
\label{lem:EVW_discrete}
    Let $H \colon [0, 1]^V \rightarrow \mathbb{R}$ be a function with absolutely continuous first partial derivatives. For an integer $n \geq 1$, let $N = [n]$ and define $h \colon 2^{N \times V} \rightarrow \mathbb{R}$ as $h(S) = H(\vx)$ where $x_i = \frac{1}{n}|S \cap (N \times \{i\})|$. Then, we have the following: 
    \begin{enumerate}
        \item If $\frac{\partial H}{\partial x_i} \geq 0$ everywhere for each $i$, then $h$ is monotone. 

        \item If $\frac{\partial^2 H}{\partial x_i \partial x_j} \leq 0$ almost everywhere for all $i, j$, then $h$ is submodular.
    \end{enumerate}
\end{lemma}

We will need the notion of transpose-invariant functions to state the next lemma. 
\begin{definition}
Let $V=[k]\times [k]$. A set function $f:2^V\rightarrow \R$ is \emph{transpose-invariant} if $f(A)=f(A^T)$ for every $A\subseteq V$. Generalizing, we say that a function $F:[0,1]^V\rightarrow R$ is \emph{transpose-invariant} if $F(\vx)=F(\vx^T)$ for every $\vx \in [0,1]^{[k]\times [k]}$, where $\vx^T$ is the transpose of the matrix $\vx$. 
\end{definition}
We note that a row-column-type set function is transpose-invariant by definition. 

\begin{lemma}[Lemma 3.2 from \cite{Vondrak}] 
\label{lem:EVW_continuous}
    Let $V=[k]\times [k]$, $f: 2^V\rightarrow \R_{\ge 0}$ be a transpose-invariant function over the set $V$, and $\eps>0$. 
    For $\vx\in [0,1]^{V}$, denote $\overline{\vx}=(\vx+\vx^T)/2$. 
    Then, there exists $\delta>0$ and functions $\hat{F}, \hat{G} \colon [0, 1]^V \rightarrow \mathbb{R}_{\ge 0}$ satisfying the following: 
    %Consider a function $f \colon 2^V \rightarrow \mathbb{R}$ that is invariant under a group of permutations $\mathcal{G}$ on the ground set $V$. Let $F$ be the multilinear extension of $f$, $\overline{\vx} = \E_{\sigma \in \mathcal{G}} \left[\sigma(\vx)\right]$, and fix an arbitrary $\epsilon > 0$. Then, there exists $\delta > 0$ and functions $\hat{F}, \hat{G} \colon [0, 1]^V \rightarrow \mathbb{R}_+$ (which are also symmetric with respect to $\mathcal{G}$) \knote{What does symmetric with respect to $\mathcal{G}$ mean?} \rnote{I think this means that $\hat{F}$ and $\hat{G}$ are invariant under $\mathcal{G}$}, satisfying:
    \begin{enumerate}[label=(\roman*)]
        \item Both $\hat{F}$ and $\hat{G}$ are transpose-invariant. %\knote{where is this conclusion used?}
        \item For all $\vx \in [0, 1]^V$, $\hat{G}(\vx) = \hat{F}(\overline{\vx})$.

        \item For all $\vx \in [0, 1]^V$, $|\hat{F}(\vx) - F(\vx)| \leq \epsilon$. 

        \item If $\|\vx - \overline{\vx}\|^2 \leq \delta$, then $\hat{F}(\vx) = \hat{G}(\vx)$ and the value of $\hat{F}(x)$ depends only on $\overline{\vx}$. 

        \item The first partial derivatives of $\hat{F}$ and $\hat{G}$ are absolutely continuous.

        \item If $f$ is monotone, then $\frac{\partial \hat{F}}{\partial x_i} \geq 0$ and $\frac{\partial \hat{G}}{\partial x_i} \geq 0$ everywhere. \label{label:monotone}

        \item If $f$ is submodular, then $\frac{\partial^2 \hat{F}}{\partial x_i \partial x_j} \leq 0$ and $\frac{\partial^2 \hat{G}}{\partial x_i \partial x_j} \leq 0$ almost everywhere. \label{submodular}
    \end{enumerate}
\end{lemma}

\begin{lemma}[Lemma 3.3 from \cite{Vondrak}]
\label{lem:EVW_distinguish} 
Let $V=[k]\times [k]$, $f: 2^V\rightarrow \R_{\ge 0}$ be a transpose-invariant function over the set $V$, and $\eps>0$. 
    Let $\hat{F}$, $\hat{G}$ be the two functions satisfying the conclusions of Lemma \ref{lem:EVW_continuous}. For an integer $n \ge 1$ and $N = [n]$, define two discrete functions $\hat{f}, \hat{g} \colon 2^{N \times V} \rightarrow \mathbb{R}_+$ as follows: 
    Let $\sigma_r\in \{0, 1\}$ be chosen uniformly at random for each $r\in N$. 
    %Let $\sigma^{(r)}$ be chosen uniformly at random in $\mathcal{G}$ for each $r \in N$. 
    For every set $S \subseteq N \times V$, we define a vector $\xi(S) \in [0, 1]^V$ by 
    \begin{align*}
    %:= \{r: \sigma^r = 0\}
    \xi(S)_{(i, j)} = \frac{1}{n} \left|\left\{r \in N \colon (r, (i, j)) \in S\text{ and } \sigma_r = 0\right\}\cup 
    \left\{r \in N \colon (r, (j, i)) \in S\text{ and } \sigma_r = 1\right\}
    \right|. 
        %\xi(S)_{(i, j)} = \frac{1}{n} |\{r \in N \colon (r, \sigma^{(r)}(i, j) \in S\}|. 
    \end{align*}
    and define $\hat{f}(S) = \hat{F}(\xi(S))$ and $\hat{g}(S) = \hat{G}(\xi(S))$. Then, deciding whether a function given by a value oracle is $\hat{f}$ or $\hat{g}$ (even using a randomized algorithm with a constant probability of success) requires $2^{\Omega(n)}$ function evaluation queries. 
\end{lemma}

We now prove Theorem \ref{thm:sym-gap-to-oracle-lower-bound}. In Theorem \ref{thm:sym-gap-to-oracle-lower-bound-terminal-sets} below, we prove that symmetry gap gives an inapproximability bound for a slight variant of \monosubmp where the terminals correspond to disjoint sets rather than singleton elements. 
We note that Theorem \ref{thm:sym-gap-to-oracle-lower-bound} follows from Theorem \ref{thm:sym-gap-to-oracle-lower-bound-terminal-sets}  (i.e., for the version of \monosubmp where the terminals are singleton elements) by a standard contraction  of terminal sets into a single terminal and appropriately modifying the associated submodular function. 
%\symGapToLowerBound*
% \begin{theorem}\label{thm:sym-gap-to-oracle-lower-bound-terminal-sets}
%     Let $V=[k]\times [k]$ and $f\colon 2^V \rightarrow \mathbb{R}_{\ge 0}$ be a row-column-type non-negative monotone submodular function over the set $V$. Then, for all $\gamma > 0$ and sufficiently large $n \in \mathbb{Z_+}$, there exists a non-negative monotone submodular function $\hat{f} \colon 2^{[n] \times V} \rightarrow \mathbb{R_+}$ and a collection of pairwise disjoint terminal sets $T_1, \ldots, T_k \subseteq  [n]\times V$ such that every algorithm that achieves an approximation factor of $\symgap(f)-\gamma$
%     %require an exponential number of oracle queries in $n$ to achieve an $\left(\frac{\overline{OPT}}{OPT} - \gamma\right)$-approximation 
%     for the following instance of \monosubmp requires $2^{\Omega(n)}$ function evaluation queries: 
%     \[
%         \min\left\{\sum_{i = 1}^k \hat{f}(V_i) \colon V_1, \ldots, V_k \text{ is a partition of } [n] \times V \text{ with } T_i \subseteq V_i \ \forall\ i \in [k]\right\}.  
%     \]
% \end{theorem}
\begin{theorem}\label{thm:sym-gap-to-oracle-lower-bound-terminal-sets}
    Let $V=[k]\times [k]$ and $f\colon 2^V \rightarrow \mathbb{R}_{\ge 0}$ be a row-column-type non-negative monotone submodular function over the set $V$. Then, for all $\gamma > 0$ and sufficiently large $n \in \mathbb{Z_+}$, there exist non-negative monotone submodular functions $\hat{f}, \hat{g} \colon 2^{[n] \times V} \rightarrow \mathbb{R_+}$ and a collection of pairwise disjoint terminal sets $T_1, \ldots, T_k \subseteq  [n]\times V$ such that there does not exist an algorithm that achieves an approximation factor of $\symgap(f) - \gamma$ for the following instances of \monosubmp while using $2^{o(n)}$ function evaluation queries on these instances:
    \begin{align*}
        \min&\left\{\sum_{i = 1}^k \hat{f}(V_i) \colon V_1, \ldots, V_k \text{ is a partition of } [n] \times V \text{ with } T_i \subseteq V_i \ \forall\ i \in [k]\right\} \text{ and} \\
        \min&\left\{\sum_{i = 1}^k \hat{g}(V_i) \colon V_1, \ldots, V_k \text{ is a partition of } [n] \times V \text{ with } T_i \subseteq V_i \ \forall\ i \in [k]\right\}.  
    \end{align*}
\end{theorem}
\begin{proof}

    Let $\gamma > 0$. We choose $\epsilon > 0$ such that $\frac{\symopt - \epsilon k}{\opt + \epsilon(k + 1)} > \frac{\symopt}{\opt} - \gamma$. 
    Consider the multilinear relaxation of multiway partitioning wrt the function $f: 2^V\rightarrow \R$ for the terminals $t_{\ell}=(\ell,\ell)$ for every $\ell\in [k]$. We  fix $n \in \mathbb{Z}_+$ large enough such that there exists a feasible solution $\tilde{\vx}$ to the multilinear relaxation such that the entries of $\tilde{\vx}$ are all rational numbers with $n$ in the denominator and $\sum_{\ell = 1}^k F(\tilde{\vx}_\ell) \leq \sum_{\ell = 1}^k F(\vx^*_\ell) + \epsilon$, where $\vx^*$ is an optimal solution to the multilinear relaxation (we note that such a solution $\tilde{\vx}$ exists since the objective of the multilinear relaxation is continuous). 
    For the chosen value of $\epsilon$, we use Lemma \ref{lem:EVW_continuous} to obtain $\hat{F}:[0,1]^V\rightarrow \R_+$ and $\hat{G}:[0,1]^V\rightarrow \R_+$ satisfying the conclusions of the lemma. For the chosen value of $n$, we let $N=[n]$ and $X=N\times V$. For each $r \in N$, we choose $\sigma_r \in \{0,1\}$ uniformly at random. Now, for $S \subseteq X$, we define the vector $\xi(S) \in [0, 1]^V$ such that 
    \[
        \xi(S)_{(i, j)} = \frac{1}{n} \left|\left\{r \in N \colon (r, (i, j)) \in S\text{ and } \sigma_r = 0\right\}\cup 
    \left\{r \in N \colon (r, (j, i)) \in S\text{ and } \sigma_r = 1\right\}
    \right|. 
    \]
    Next, we define $\hat{f}, \hat{g} \colon X \rightarrow \mathbb{R}_+$ as $\hat{f}(S) = \hat{F}(\xi(S))$ and $\hat{g}(S) = \hat{G}(\xi(S))$. By Lemma \ref{lem:EVW_continuous}(v)--(vii) and Lemma \ref{lem:EVW_discrete}, the functions $\hat{f}$ and $\hat{g}$ are monotone submodular. 
    %As described in \cite{Vondrak}, for the purposes of proving the monotonicity and the submodularity of $\hat{f}$ and $\hat{g}$, it suffices to assume that $\sigma^{(r)} = id$ for all $r \in N$ \knote{Why is this mentioned? Why doesn't it follow directly from the lemmas?}.
    %Fix $\epsilon > 0$ \knote{Is this the same $\eps$ that appears in the previous paragraph?}. 
    For each $\ell\in [k]$, let $T_{\ell}:=\{(r, (\ell, \ell): r\in N\}$. 
    We consider the following two instances of \monosubmp:
    \begin{align*}
        \min&\left\{\sum_{\ell = 1}^k \hat{f}(X_\ell) \colon X_1, \ldots, X_k \text{ is a partition of } X \text{ with } T_{\ell}\subseteq X_{\ell}\ \forall\ \ell\in[k]\right\} \text{ and}  \\
        \min&\left\{\sum_{\ell = 1}^k \hat{g}(X_\ell) \colon X_1, \ldots, X_k \text{ is a partition of } X \text{ with } T_{\ell}\subseteq X_{\ell}\ \forall\ \ell\in[k]\right\}. 
    \end{align*}
    %For each $\ell \in [k]$, we can contract $N \times (\ell, \ell)$ into a single vertex $t_\ell$ such that we specify exactly one terminal per part $X_\ell \in \{X_1, \ldots, X_k\}$ \knote{What is contraction?}. 
    Suppose for contradiction there exists an approximation algorithm $\mathcal{A}$ for \monosubmp that achieves an approximation factor of $\left(\symgap(f)- \gamma\right)$ for both instances of \monosubmp using $2^{o(n)}$ function evaluation queries. We will use such an algorithm to design a procedure that decides whether a function given by a value oracle is $\hat{f}$ or $\hat{g}$ using $2^{o(n)}$ function evaluation queries, thereby contradicting Lemma \ref{lem:EVW_distinguish}. 
    
    We now describe our procedure to distinguish between $\hat{f}$ and $\hat{g}$ using algorithm $\mathcal{A}$. Given an oracle for an unknown function $h \in \{\hat{f}, \hat{g}\}$, the procedure applies algorithm $\mathcal{A}$ to the following instance of \monosubmp:
    \[
        \min\left\{\sum_{\ell = 1}^k h(X_\ell) \colon X_1, \ldots, X_k \text{ is a partition of } X \text{ with } T_{\ell}\subseteq X_\ell\ \ \forall\ \ell \in [k] \right\}. 
    \]
    Let the value returned by the algorithm be $\beta$. If $\beta < \symopt - \epsilon k$, then the procedure declares that $h = \hat{f}$. Otherwise, %$\beta \geq \symopt - \epsilon k$ and 
    the procedure declares that $h = \hat{g}$. 
    
    We now show correctness of our procedure. It suffices to show that $h=\hat{f}$ if and only if $\beta<\symopt - \eps k$. 
    Suppose $h\neq \hat{f}$, i.e., $h=\hat{g}$. Then, by Claim \ref{claim:hat-g-opt-is-larger-than-sym-opt} shown below, $\beta \geq \symopt - \epsilon k$. 
    Next, suppose that $h = \hat{f}$. Then, by Claim \ref{claim:hat-f-opt-is-smaller-than-sym-opt} shown below, $\beta \leq (OPT + \epsilon(k + 1)) \left( \symgap(f) - \gamma\right) < (OPT + \epsilon(k + 1)) \left( \frac{\symopt - \epsilon k}{\opt + \epsilon(k + 1)}  \right) = \symopt - \epsilon k$. 
    %We correctly determine the identity of $h$ using fewer than an exponential number of oracle queries in $n$. 
    %However, this contradicts Lemma \ref{lem:EVW_distinguish}. We conclude that all algorithms achieving a $\left(\frac{\overline{OPT}}{OPT} - \gamma\right)$-approximation to the above instance of \monosubmp require exponentially many oracle queries in $|X|$.
\end{proof}

\begin{claim}\label{claim:hat-g-opt-is-larger-than-sym-opt}
        For every partition $X_1, \ldots, X_k$ of $X = N \times V$ with $T_{\ell} \subseteq X_{\ell}$ for every $\ell \in [k]$, we have that 
        \[
        \sum_{\ell = 1}^k \hat{g}(X_\ell) \geq \symopt - \epsilon k. 
        \]
    \end{claim}
    %We demonstrate that for all partitions $X_1, \ldots, X_k$ of $X = N \times V$ such that $(r, (\ell, \ell)) \in X_\ell$ for each $r \in N$ and $\ell \in [k]$, $\sum_{\ell = 1}^k \hat{g}(X_\ell) \geq \symopt - \epsilon k$. 
    \begin{proof}
    Let $X_1, \ldots, X_k$ be a partition of $X$ with $T_{\ell} \subseteq X_{\ell}$ for every $\ell \in [k]$.  We have that 
    \begin{align*}
        \sum_{\ell = 1}^k \hat{g}(X_\ell) &= \sum_{\ell = 1}^k \hat{G}(\xi(X_\ell)) \quad \quad \text{(by definition of $\hat{g}$)}\\
        &= \sum_{\ell = 1}^k \hat{F}(\overline{\xi(X_\ell)}) \quad \quad \text{(by Lemma \ref{lem:EVW_continuous}(ii))}\\
        &\geq \sum_{\ell = 1}^k (F(\overline{\xi(X_\ell)}) - \epsilon) \quad \quad \text{(by Lemma \ref{lem:EVW_continuous}(iii))}\\
        &= \sum_{\ell = 1}^k F(\overline{\xi(X_\ell)}) - \epsilon k\\
        &\ge \symopt-\epsilon k. \quad \quad \text{(by Lemma \ref{lem:OPT_equality})} 
    \end{align*}
    The last inequality is by Lemma \ref{lem:OPT_equality} since  $\left(\overline{\xi(X_\ell)}\right)_{\ell = 1}^k$ is a feasible symmetric solution to the multilinear relaxation of multiway partitioning wrt the function $f:2^V\rightarrow \R$ for the terminals $t_{\ell}=(\ell,\ell)$ for every $\ell\in [k]$. 
    \end{proof}

    \begin{claim}\label{claim:hat-f-opt-is-smaller-than-sym-opt}
        There exists a partition $X_1, \ldots, X_k$ of $X = N \times V$ with $T_{\ell}\subseteq X_{\ell}$ for every $\ell \in [k]$ such that 
        \[
        \sum_{\ell = 1}^k \hat{f}(X_\ell) \leq \opt + \epsilon (k + 1). 
        \]
    \end{claim}
    \begin{proof}
    %Now, we demonstrate that for sufficiently large $n$ there exists a partition $X_1, \ldots, X_k$ of $X = N \times V$ satisfying $(r, (\ell, \ell)) \in X_\ell$ for each $r \in N$ and $\ell \in [k]$ such that $\sum_{\ell = 1}^k \hat{f}(X_\ell) \leq \opt + \epsilon (k + 1)$. 
    %Let $\vx^*$ be an optimum solution to the multilinear relaxation with value $\opt$. %defined in Definition \ref{def:multilinear} 
    %Since the objective of the multilinear relaxation is continuous, we may choose $n$ large enough such that 
    We recall that there exists a feasible solution $\tilde{\vx}$ and an optimum solution $\vx^*$ to the multilinear relaxation of multiway partitioning wrt the function $f: 2^V\rightarrow \R$ for the terminals $t_{\ell}=(\ell,\ell)$ for every $\ell\in [k]$ with the following property: the entries of $\tilde{\vx}$ are all rational numbers with $n$ in the denominator and $\sum_{\ell = 1}^k F(\tilde{\vx}_\ell) \leq \sum_{\ell = 1}^k F(\vx^*_\ell) + \epsilon$. 
    We will round such a fractional solution $(\tilde{\vx}_{\ell})_{\ell\in [k]}$ to a partition $X_1, \ldots, X_k$ of $X=N\times V$ with $T_{\ell}\subseteq X_{\ell}$ for every $\ell\in [k]$ as follows: 
    %$\tilde{\vx} = \left(\xi(X_\ell)\right)_{\ell = 1}^k$ for the partition $X_1, \ldots, X_k$ defined as follows. Fix $i, j \in [k]$. Then for $1 \leq r \leq nx((i, j), 1)$, let $(r, \sigma^{(r)}(i, j))$ be assigned to $X_1$, for $nx((i, j), 1) + 1 \leq r \leq nx((i, j), 1) + nx((i, j), 2)$ let $(r, \sigma^{(r)}(i, j))$ be assigned to $X_2$, etc. 
    %We define a partition $X_1, \ldots, X_k$ by rounding $\tilde{\vx}$ as follows: 
    we define 
    \begin{align*}
        X_1^0&:=\{(r, (i,j)): r\in N, i, j\in [k], \sigma_r=0, r\in [n\tilde{x}((i,j), 1)]\}\},\\
        X_1^1&:=\{(r, (j,i)): r\in N, i, j\in [k], \sigma_r=1, r\in [n\tilde{x}((i,j), 1)]\}\}, 
    \end{align*}
    and for each $\ell\in \{2, \ldots, k\}$, 
    \begin{align*}
        X_{\ell}^0&:= \left\{(r, (i,j)): r\in N, i, j\in [k], \sigma_r=0, r\in [n\tilde{x}((i,j), {\ell-1})+1, n\tilde{x}((i,j), \ell-1)+n\tilde{x}((i,j), \ell)]\right\}\\
        X_{\ell}^1&:= \left\{(r, (j,i)): r\in N, i, j\in [k], \sigma_r=1, r\in [n\tilde{x}((i,j), {\ell-1})+1, n\tilde{x}((i,j), \ell-1)+n\tilde{x}((i,j), \ell)]\right\}, 
    \end{align*}
    and set
    $X_{\ell}:=X_{\ell}^0\cup X_{\ell}^1$ for every $\ell\in [k]$. We note that $X_1, \ldots, X_{\ell}$ is a partition of $X=N\times V$ with $T_{\ell}\subseteq X_{\ell}$ for every $\ell\in [k]$. Finally, we bound the cost of the solution as follows: 
    \begin{align*}
        \sum_{\ell = 1}^k \hat{f}(X_\ell) &= \sum_{\ell = 1}^k \hat{F}(\tilde{\vx}_\ell) \quad \quad \text{(by definition of $\hat{f}$ and $X_1, \ldots, X_k$)}\\
        &\leq \sum_{\ell = 1}^k \left(F(\tilde{\vx}_\ell) + \epsilon\right) \quad \quad \text{(by Lemma \ref{lem:EVW_continuous}(iii))}\\
        &= \sum_{\ell = 1}^k F(\tilde{\vx}_\ell) + \epsilon k \\
        &\leq \sum_{\ell = 1}^k F(\vx^*_\ell) + \epsilon(k + 1) \quad \quad \left(\text{since $\sum_{\ell = 1}^k F(\tilde{\vx}_\ell) \leq \sum_{\ell = 1}^k F(\vx^*_\ell) + \epsilon$}\right)\\
        &= \opt + \epsilon(k + 1). 
    \end{align*}
    The last inequality is by Lemma \ref{lem:OPT_equality} since  $\left(x^*_\ell\right)_{\ell = 1}^k$ is an optimum solution to the multilinear relaxation of multiway partitioning wrt the function $f:2^V\rightarrow \R$ for the terminals $t_{\ell}=(\ell,\ell)$ for every $\ell\in [k]$. 
    \end{proof}

%% file: graph-coverage-MP.tex
\section{Graph Coverage Multiway Partition}
\label{sec:graph-coverage-mp}
% \knote{Define the problem. Is it NP-hard? No PTAS? Explain these aspects. }
% \knote{State the approximation result.}

We recall \gcovmp. The input here is a graph $G = (V, E, w)$ where $w : E \to \R_+$ is the edge cost function,
% \knote{I prefer to use edge cost function $c:E\to \R_+$ - general convention is to use cost for minimization problems and weights for maximization problems---e.g., min-cost spanning tree and max-weight matching}
and a set of terminal vertices $T = \{t_1, \ldots, t_k\}$. The goal is to find a partition $S_1, \ldots, S_k$ of $V$ with $t_i\in S_i$ for every $i\in [k]$ in order to minimize $\sum_{i=1}^k f(S_i)$, 
\iffalse
\begin{align*}
    \min \left\{\sum_{i=1}^k f(S_i) : S_1, \cdots, S_k \text{ is a partition of $V$ }, t_i \in S_i \forall i \in [k]\right\},
\end{align*}
\fi 
where $b : 2^V \to \R_+$ is the graph coverage function defined as  
\begin{align*}
    b(S) := \sum_{e = \{u, v\} \in E : \{u, v\}\cap S\neq \emptyset} w_e\ \forall\ S\subseteq V.
\end{align*}
There is a close relationship between \gcovmp and \mwc that we recall. 
%We observe that the objective function counts an edge $e = \{u, v\}$ once if both its end-vertices lie inside the same partition and twice if its end-vertices lie in different partitions. 
In \mwc, the input is identical to \gcovmp while the goal is to find a partition $S_1, \ldots, S_k$ of $V$ with $t_i\in S_i$ for every $i\in [k]$ in order to minimize $(1/2)\sum_{i=1}^k d(S_i)$, where $d:2^V\to \R_+$ is the graph cut function defined as 
\begin{align*}
    d(S) := \sum_{e = \{u, v\} \in E : |\{u, v\}\cap S|=1} w_e\ \forall\ S\subseteq V.
\end{align*}
%For a partition $S_1, \ldots, S_k$ of $V$, let $\delta(S_1, \ldots, S_k):=\{e\in E: e\in \delta(S_i)\cap \delta(S_j)\text{ for distinct }i, j\in [k]\}$. 
Then, for a partition $S_1, \ldots, S_k$ of $V$, the objective in \gcovmp can be rewritten as, 
\begin{align}
    \sum_{i=1}^k b(S_i) &= 
    \frac{1}{2}\sum_{i=1}^k d(S_i) + \sum_{e\in E}w_e. \label{eq:cov-vs-cut-translation}
    %2\sum_{e \in \delta(S_1, \cdots, S_k)} w_e + \sum_{e \notin \delta(S_1, \cdots, S_k)}w_e 
    %= \sum_{e \in \delta(S_1, \cdots, S_k)} w_e + \sum_{e \in E} w_e.
\end{align}
Hence, the objective function of \gcovmp is a translation of the objective function of \mwc. This connection inspires the approaches underlying the results in this section. 
We begin with some easy consequences of the connection between \mwc and \gcovmp. 
Firstly, since \mwc is NP-hard, it follows that \gcovmp is also NP-hard. Secondly, the following proposition shows that an $\alpha$-approximation for \mwc immediately implies a $(1+\alpha)/2$-approximation for \gcovmp. 

\begin{proposition}\label{prop:mwc-to-gcovmp-approximation}
If we have an $\alpha$-approximation for \mwc, then we have a $(1+\alpha)/2$-approximation for \gcovmp. 
\end{proposition}
% \donote{Include the proof.}
\begin{proof}
    Given an instance of \gcovmp, we run the $\alpha$-approximation algorithm for \mwc on this instance to obtain a partition $S_1, \ldots, S_k$ of the vertex set and return the same partition. We now analyze the approximation factor of this approach. 
    Let $\opt_{\text{MWC}}$ and $\opt_{\text{G-COV-MP}}$ denote the optimum objective value of \mwc and \gcovmp respectively on the instance $I$. 
    We have that 
    \begin{align*}
        \frac{1}{2}\sum_{i=1}^k d(S_i) &\leq \alpha \cdot \opt_{\text{MWC}}, \text{ and}\\
        \opt_{\text{G-COV-MP}} &= \opt_{\text{MWC}} + \sum_{e\in E}w_e. 
    \end{align*}
    Since for every partition $T_1, \ldots, T_k$ of $V$, we have that $(1/2)\sum_{i=1}^k d(T_i)\le \sum_{e\in E}w_e$, we have that 
    \[
    \opt_{\text{MWC}} \le \sum_{e\in E}w_e. 
    \]
    Thus, the \gcovmp objective value of the partition $S_1, \ldots, S_k$ of $V$ returned by our algorithm is 
    \begin{align*}
    \sum_{i=1}^k b(S_i) 
    &= \frac{1}{2}\sum_{i=1}^k d(S_i) + \sum_{e\in E}w_e\\
    &\le \alpha \cdot \opt_{\text{MWC}} + \sum_{e\in E}w_e\\
    &\le \alpha \cdot \opt_{\text{MWC}} + \sum_{e\in E}w_e + \left(\frac{\alpha-1}{2}\right)\left(\sum_{e\in E}w_e - \opt_{\text{MWC}}\right)\\
    &= \left(\frac{\alpha+1}{2}\right)\left(\opt_{\text{MWC}}+ \sum_{e\in E}w_e\right)\\
    &= \left(\frac{\alpha+1}{2}\right)\opt_{\text{G-COV-MP}}. 
    \end{align*}
    \iffalse
    Given an instance $I$, let the $\alpha$-approximation algorithm for \mwc give a solution with cost $A(I)$, and let the optimal cost for that instance be $Opt(I)$. Then, 
    Since every multiway-cut is a subset of all edges of the graph, we get 
    \begin{align*}
        Opt(I) \leq A(I) \leq w(E).
    \end{align*}
    We now consider the same approximation algorithm for \gcovmp. Let the cost of the solution be $A'(I)$ and the corresponding optimal cost be $Opt'(I)$. Then we get, 
    \begin{align*}
        A'(I) &= w(E) + A(I) \leq w(E) + \alpha \cdot Opt(I) \\ 
        &\leq w(E) + \alpha \cdot Opt(I) + \left(\frac{\alpha-1}{2}\right) (w(E) - Opt(I))\\
        &= \left(\frac{1+\alpha}{2}\right)(w(E) + Opt(I)) = \left(\frac{1+\alpha}{2}\right) Opt'(I).
    \end{align*}
    \fi
\end{proof}

We note that the best possible approximation factor for $\mwc$ is $1.2965$. By Proposition \ref{prop:mwc-to-gcovmp-approximation}, this implies a $1.14825$-approximation for \gcovmp. 
In the next theorem, we improve on this approximability. 
\thmGcovMPApprox*

Since the reduction between \gcovmp and \mwc mentioned above is not approximation preserving, the inapproximability factor for \mwc does not give the same inapproximability factor for \gcovmp. With this issue in mind, we show the following inapproximability results for \gcovmp. 

%It was shown in \cite{cmplx-mwc} that \mwc does not admit a PTAS. Their reduction can be adapted to also conclude that \gcovmp is also APX-hard (see Theoerm \ref{thm:apx-hardness}). 
%We show that this can be extended to \gcovmp using an approach similar to \cite{cmplx-mwc}.

\thmGCovMPInapprox*

\subsection{Approximation Algorithm}
% \donote{Do the edits suggested in this section.}
In this section, we prove Theorem \ref{thm:gcovmp-approximation}. 
We will use the following LP-relaxation of \gcovmp: % (denoted \gcovrel): 
%\knote{I think the $1/2$ factor should not be in the objective. Confirm and remove it.}
\begin{align}
    \min \quad &\sum_{e \in E}w_e \left( \sum_{i=1}^{k} \max\{x_u^i, x_v^i\} \right) \tag{Graph-Cov-MP-Rel}\label{LP:gcovmp}\\ 
    &x_u \in \Delta_k \quad \quad \forall u \in V  \notag \\ 
    &x_{t_i} = e_i \quad \quad \forall i \in [k] \notag
\end{align}
where $\Delta_k$ denotes the $k$-simplex, i.e. $\Delta_k:=\{x\in \R^k: \sum_{i=1}^k x_i = 1, x\ge 0\}$ and $e_i$ denotes the unit vector in $\R^k$ with $1$ at the $i$-th index. 
% \knote{I would prefer to write the objective as $(1/2)\sum_{e=uv\in E}w_e\sum_{i=1}^k \max\{x_u^i, x_v^i\}$. Also, leave the analysis below as it is, but create a new copy and do the analysis wrt this objective function directly.}
Our rounding algorithm for the LP-relaxation is identical to the exponential clock based simplex partitioning given by Buchbinder, Naor, and Schwartz for \mwc \cite{smplx-mwc}. However, our analysis of the approximation factor is different since the objective of interest is different. We present our algorithm for \gcovmp below:
% \knote{State the LP for \gcovmp.}
\begin{algorithm}[H] 
\caption{\gcovmp Rounding}\label{alg:gcov}
\begin{algorithmic}
    \State Let $\vx$ be an optimal solution to \ref{LP:gcovmp} %\gcovrel
    \State Choose i.i.d random variables $Z_i \sim \exp(1)$ for $i \in [k]$
    \State For each $u\in V$, define $\ell(u):= \argmin\{\frac{Z_i}{(x_u)_i} : i \in [k]\}$
    %$\forall u \in V$ : $l(u) = \argmin\{\frac{Z_i}{(x_u)_i} : i \in [k]\}$
    \State For each $i\in [k]$, define $S_i:=\{u\in V: \ell(u)=i\}$. 
    \State \Return $S_1, \ldots, S_k$
\end{algorithmic}
\end{algorithm}
%\subsubsection{Approximation Factor Analysis}
We analyze the approximation factor of the algorithm in the rest of the section. 
We recall that $b: 2^V\rightarrow \R_{\ge 0}$ is the coverage function of the input $(G=(V, E), w: E\rightarrow \R_{\ge 0})$. 
The following result bounds the approximation factor of Algorithm \ref{alg:gcov} relative to the optimum value of \ref{LP:gcovmp}. 
\begin{restatable}{theorem}{thmGcovMPApproxReltoLP}\label{thm:gcovmp-approx-rel-to-lp}
    Let $\opt_{frac}$ be the optimum value of \ref{LP:gcovmp} and $S_1, \ldots, S_k$ be the solution returned by Algorithm \ref{alg:gcov}. Then, 
    \[
        \E\left[\sum_{i=1}^k b(S_i)\right]\le \frac{9}{8}\opt_{frac}. 
    \]
    %Algorithm \ref{alg:gcov} achieves an approximation factor of $\frac{9}{8}$ for \gcovmp.
\end{restatable}

The rest of this section is devoted to proving Theorem \ref{thm:gcovmp-approx-rel-to-lp}. 
Let $x$ be a feasible solution to \ref{LP:gcovmp}. For $u\in V$, we denote $x_u = (u_1, \cdots, u_k)$ and for every $u, v\in V$, we denote  $\epsilon_{uv} := \sum_{j=1}^{k} \max\{0, u_j - v_j\}$. We begin with the following proposition. 
\begin{proposition} \label{prop:u-v-sum}
For every pair $u, v\in V$, we have that 
\[
\epsilon_{uv} = \epsilon_{vu}=\frac{1}{2} \lVert x_u - x_v \rVert_1. 
\]

%    $\sum_{j=1}^k \max\{0, v_j - u_j\} = \sum_{j=1}^k \max\{0, u_j - v_j\}$.
\end{proposition}
\begin{proof}
    Since $x_u, x_v \in \Delta_k$, we have that $\sum_{j=1}^{k} u_j = \sum_{j=1}^{k} v_j = 1$. Hence we get, $\sum_{j=1}^{k} (u_j - v_j) = 0$. This gives us,
    \begin{align*}
        \sum_{j:v_j \geq u_j} (u_j - v_j) + \sum_{j:v_j < u_j} (u_j - v_j) = 0.
    \end{align*}
    Consequently, we have that 
    \begin{align*}
        \epsilon_{uv}=\sum_{j=1}^{k} \max\{0, v_j - u_j\} = \sum_{j:v_j \geq u_j} (v_j - u_j) = \sum_{j:v_j < u_j} (u_j - v_j) = \sum_{j=1}^{k} \max\{0, u_j - v_j\}=\epsilon_{vu}.
    \end{align*}
    By the same reasoning, we have that 
    \[
    \frac{1}{2} \lVert x_u - x_v \rVert_1 = \frac{1}{2}\left(\epsilon_{uv}+\epsilon_{vu}\right) = \eps_{uv}. 
    \]
    % Further we have,
    % \begin{align*}
    %     \sum_{j=1}^{k} |v_j - u_j| &= \sum_{j : v_j \geq u_j} (v_j - u_j) + \sum_{j : v_j < u_j} (u_j - v_j)\\ 
    %     &= 2\sum_{j:v_j \geq u_j} (v_j - u_j) + \sum_{j=1}^{k} (u_j - v_j) \\ 
    %     &= 2 \sum_{j=1}^{k} \max\{0, v_j - u_j\}. \quad \quad \text{(from equation (\ref{eq:u-v}))}
    % \end{align*}
    % Hence, 
    % \begin{align*}
    %     \epsilon_{uv} = \frac{1}{2} \lVert x_u - x_v \rVert_1 = \sum_{j=1}^{k} \max\{0, v_j - u_j\}.
    % \end{align*}
\end{proof}
%Henceforth, we will use $\epsilon_{uv} := \sum_{j=1}^{k} \max\{0, u_j - v_j\}$ for every $u, v\in V$. We note that $\epsilon_{uv} = \epsilon_{vu}$ by Proposition \ref{prop:u-v-sum}.
We will use the following properties of the exponential distribution. 

\begin{proposition} \label{prop:exp-dist}
    If $X \sim \exp(\lambda)$ and $c > 0$, then $\frac{X}{c} \sim \exp(\lambda c)$. Moreover, if we have independent random variables $X_1, \cdots, X_k$ with $X_i \sim \exp(\lambda_i)$, then 
    \begin{enumerate}
        \item $\min\{X_1, \cdots, X_k\} \sim \exp(\lambda_1 + \cdots + \lambda_k)$. 
        \item $Pr[X_i \leq \min\{X_j: j \in [k]\setminus \{i\}\}] = \frac{\lambda_i}{\lambda_1 + \ldots + \lambda_k}$.
    \end{enumerate}
\end{proposition}

We use Proposition \ref{prop:exp-dist} to bound the probability of distinct vertices receiving different labels. 
\begin{lemma}\label{lem:different-labels-prob}
    For each distinct $u, v\in V$, we have that $Pr(\ell(u) \neq \ell(v)) \leq \frac{2\epsilon_{uv}}{1 + \epsilon_{uv}}$.
\end{lemma}
\begin{proof}
    Let $A_i$ denote the event that $\ell(u) = \ell(v) = i$. 
    %\donote{This series of calculations uses a standard fact about exponential distribution. Write it as a proposition at the beginning of the section or just above the lemma. }\snote{Done} 
    Then we get, 
    \begin{align*}
        Pr[A_i] &= Pr\left[\frac{z_i}{u_i} \leq \min\left\{\frac{z_1}{u_1}, \cdots, \frac{z_{i-1}}{u_{i-1}}, \frac{z_{i+1}}{u_{i+1}}, \cdots \right\}, \frac{z_i}{v_i} \leq \min\left\{\frac{z_1}{v_1}, \cdots, \frac{z_{i-1}}{v_{i-1}}, \frac{z_{i+1}}{v_{i+1}}, \cdots \right\} \right] \\ 
        &\geq Pr\left[\frac{z_i}{\min\{u_i, v_i\}} \leq \min\left\{\frac{z_1}{\max\{u_1, v_1\}}, \cdots, \frac{z_{i-1}}{\max\{u_{i-1}, v_{i-1}\}}, \frac{z_{i+1}}{\max\{u_{i+1}, v_{i+1}\}}, \cdots\right\}\right] \\ 
        &= \frac{\min\{u_i, v_i\}}{\left(\sum_{j \neq i} \max\{u_j, v_j\}\right) + \min\{u_i, v_i\} } \quad\quad \text{(by Proposition \ref{prop:exp-dist})}\\
        &= \frac{u_i - \max\{0, u_i - v_i\}}{\left(\sum_{j \neq i} \left(u_j + \max\{0, v_j - u_j\}\right)\right) + u_i - \max\{0, u_i - v_i\} } \\ 
        &\geq \frac{u_i - \max\{0, u_i - v_i\}}{\sum_{j=1}^k \left(u_j + \max\{0, v_j - u_j\}\right) } \\ 
        &= \frac{u_i - \max\{0, u_i - v_i\}}{1 + \epsilon_{uv}}. 
    \end{align*}
    Since $A_i$ for $i \in [k]$ are mutually exclusive events, we have that 
    \iffalse
    \begin{align*}
        Pr[\ell(u) = \ell(v)] &= \sum_{i=1}^{k} Pr[A_i]. 
    \end{align*}
    Hence, by Proposition \ref{prop:u-v-sum}, we have that 
    \fi
    \begin{align*}
        Pr[\ell(u) \neq \ell(v)] &= 1 - \sum_{i=1}^k Pr[A_i] \\ 
        &\leq 1 - \left(\sum_{i=1}^{k} \frac{u_i - \max\{0, u_i - v_i\}}{1 + \epsilon_{uv}}\right) \quad \quad \text{(by Proposition \ref{prop:u-v-sum})} \\
        &= 1 - \left(\frac{1 - \epsilon_{uv}}{1+\epsilon_{uv}}\right)
        = \frac{2\epsilon_{uv}}{1+\epsilon_{uv}}.
    \end{align*}
\end{proof} 

We now restate and prove Theorem \ref{thm:gcovmp-approx-rel-to-lp}. 
\thmGcovMPApproxReltoLP*
\begin{proof}
    The expected cost of the solution returned by Algorithm \ref{alg:gcov} is 
    \begin{align} \label{eq:exp-cost}
        \E_{Z}\left[\sum_{i=1}^k b(S_i)\right] &= \sum_{e = \{u, v\} \in E} w_e \cdot (1 + Pr[\ell(u) \neq \ell(v)]) \quad \quad\text{(by linearity of expectation)} \notag\\ 
        &\leq \sum_{e = \{u, v\} \in E} w_e \cdot \left(1 + \frac{2\epsilon_{uv}}{1 + \epsilon_{uv}}\right) \quad \quad \text{(by Lemma \ref{lem:different-labels-prob})}\notag\\ 
        &= \sum_{e = \{u, v\} \in E} w_e \cdot \left(\frac{1+3\epsilon_{uv}}{1 + \epsilon_{uv}}\right).
    \end{align}
    Moreover, we have that, 
    \begin{align} \label{eq:opt-gcov}
        \opt_{frac} 
        &= \sum_{e=\{u,v\} \in E} w_e \left(\sum_{i=1}^k \max\{u_i, v_i\}\right) \notag\\ 
        &= \sum_{e=\{u,v\} \in E} w_e \left(\sum_{i=1}^k \left(v_i + \max\{0, u_i - v_i\}\right)\right) \notag \\
        &= \sum_{e=\{u,v\} \in E} w_e \cdot (1 + \epsilon_{uv}).
    \end{align}
    Now, using inequalities (\ref{eq:exp-cost}) and (\ref{eq:opt-gcov}) we get that
    \begin{align*}
        \frac{\E_{Z}\left[\sum_{i=1}^k b(S_i)\right]}{\opt_{frac}} &\leq \frac{\sum_{e = \{u, v\} \in E} w_e \cdot \left(\frac{1 + 3\epsilon_{uv}}{1 + \epsilon_{uv}}\right)}{\sum_{e = \{u,v\} \in E} w_e \cdot (1 + \epsilon_{uv})} \\ 
        &\leq \max_{e = \{u, v\} \in E} \left\{\frac{1 + 3\epsilon_{uv}}{(1 + \epsilon_{uv})^2}\right\}.
    \end{align*}
    Since $\epsilon_{uv} = \frac{1}{2} \lVert x_u - x_v \rVert_1$ by Proposition \ref{prop:u-v-sum} and $x_u, x_v \in \Delta_k$, we have that $\epsilon_{uv} \leq \frac{1}{2}(\lVert x_u \rVert_1 + \lVert x_v \rVert_1) = 1$. Hence, $\epsilon_{uv} \in [0, 1]$ for all $\{u, v\} \in E$.
    Now, we observe that the maximum value of the function $g(x) = \frac{1+3x}{(1+x)^2}$ for $x \in [0, 1]$ is attained at $x = \frac{1}{3}$ with $g(\frac{1}{3}) = \frac{9}{8}$. Hence, 
    \begin{align*}
        \frac{\E_{Z}\left[\sum_{i=1}^k b(S_i)\right]}{\opt_{frac}} &\leq \frac{9}{8}.
    \end{align*}
    
\end{proof}

%\knote{Discuss tightness of the analysis. Is there an example showing that the algorithm cannot achieve better than $9/8$? Even if we do not have a tight example, it is good to show some lower bound (via an example to illustrate that the algorithm cannot be better than some factor). }

\begin{remark}
The integrality gap instance constructed by C\u{a}linescu, Karloff, and Rabani \cite{CKR00} for the CKR-relaxation has an integrality gap of $46/45$ for \ref{LP:gcovmp}. 
The integrality gap instance constructed by Freund and Karloff \cite{freund-karloff} for the CKR-relaxation has an integrality gap of $22/21$ for \ref{LP:gcovmp}. We note that the integrality gap instances for $k=3$ constructed by \cite{KKSTY04} and \cite{CCT06} for the CKR-relaxation have integrality gap tending to $1$ for \ref{LP:gcovmp}. Thus, the integrality gap of \ref{LP:gcovmp} is at least $22/21=1.0476$. 
%this instance achieves the largest integrality gap for \ref{LP:gcovmp} among the well-known integrality gap instances for the CKR-relaxation \cite{CKR00, KKSTY04, CCT06}. 
\end{remark}

\subsection{Inapproximability}
In this section, we show inapproximability results for \gcovmp. Our proof approach is identical to the gadget-based reduction given by Dahlhaus, Johnson, Papadimitriou, Seymour, and Yannakakis \cite{cmplx-mwc} to show APX-hardness of \mwc. We adapt their reduction to \gcovmp in a straightforward manner. We reduce from \maxcut: the input here is a graph $G=(V, E)$ and the goal is to find a subset $U\subseteq V$ with maximum $d(U)$, where $d(U)$ is the number of edges crossing $U$. We show the following result: 

\begin{restatable}{theorem}{thmMaxCutToGCovMP} \label{thm:apx-cov-cut}
    If there exists an $\alpha$-approximation for \gcovmp, then there exists a $(164 - 163\alpha)$-approximation for \maxcut. 
\end{restatable}

Theorem \ref{thm:apx-cov-cut} implies Theorem \ref{thm:gcovmp-inapproximability} since \maxcut does not admit a $16/17$-approximation assuming P $\neq$ NP \cite{Has01} and does not admit a $0.878$-approximation assuming the unique games conjecture \cite{KKMO07}. 
For the reduction to prove Theorem \ref{thm:apx-cov-cut}, we will use the gadget graph shown in Figure \ref{fig:gadget_graph} with $9$ vertices and $3$ terminals. This is identical to the gadget graph in \cite{cmplx-mwc}. 

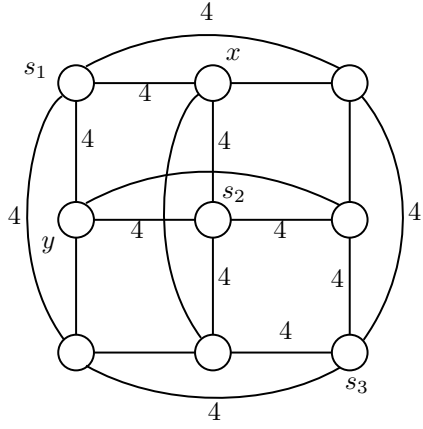
\begin{figure}[h]
\centering
\input{gadget-snp-hard}
\caption{Gadget graph with terminals $s_1, s_2, s_3$ and two specified vertices $x, y$. Twelve edges incident on terminals have weight $4$ and the remaining $6$ unlabeled edges have weight $1$.}
\label{fig:gadget_graph}
\end{figure}

Now, given an instance $G=(V, E)$ of \maxcut, we construct an instance $H$ of \gcovmp for $k=3$ as follows:  
\begin{enumerate}
    \item We begin with $3$ terminal vertices $s_1, s_2, s_3$ in $H$. 
    \item For each vertex $u\in V$, we introduce a vertex (with the same label) in $H$. 
    %The vertex set of $F$ is the vertex set of $G$ along with $3$ terminals $s_1, s_2, s_3$ and $4$ dummy vertices for each edge in $G$. 
    %Add the set of vertices of $G$ to $F$ along with $3$ terminals $s_1, s_2, s_3$. \knote{Unclear} \snote{"along" -> "along with"}
    \item For each edge $\{u, v\} \in E(G)$, add a gadget graph where $x = u, y = v$, the vertices $s_1, s_2, s_3$ are identified with the terminals, and the $4$ unlabeled vertices in the gadget are newly introduced dummy vertices. 
\end{enumerate}
Let $\opt_\maxcut$ and $\opt_{\gcovmp}$ denote the values of the optimum maxcut in $G$ and optimum coverage multiway-partitioning in $H$ respectively. 
We need the following lemma from \cite{cmplx-mwc}.
\begin{lemma}\label{lem:rel-mwc-cut} 
There exists a subset $U \subseteq V(G)$ with $d(U)\ge K$ if and only if there exists a partitioning $V_1, V_2, V_3$ of $V(H)$ with $s_i\in V_i$ such that $(1/2)\sum_{i=1}^3 d_H(V_i)\le 28|E(G)|-K$.  
%$G$ has a cut of size $K$ or greater if and only if $F$ has a $3$-way partitioning weight $B = 28|E(G)|-K$ or less.
\end{lemma}

\iffalse
The next corollary follows immediately from the lemma. 
\begin{corollary} \label{cor:rel-mwc-cut}
$OPT_{\mwc} = 28|E(G)| - OPT_{\maxcut}$
\end{corollary}
\begin{proof}
    Let $K$ be cost of optimal cut for $G$. Let there exist a 3-way-cut with cost $B < 28|E(G)| - K$. Then, from lemma \ref{lem:reduction-cov}, there exists a cut with cost $K' > 28|E(G)| - B > K$, giving a contradiction. Moreover, from lemma \ref{lem:reduction-cov} there exists a solution to $3$-way-\mwc with cost atmost $28|E(G)| - K$. 
\end{proof}
\fi

Lemma \ref{lem:rel-mwc-cut} leads to the following relation between $\opt_{\gcovmp}$ and $\opt_{\maxcut}$.
\begin{proposition} \label{prop:rel-1}
    $OPT_{\gcovmp}=82|E(G)| - OPT_{\maxcut} \leq 163 \cdot OPT_{\maxcut}$.
\end{proposition}

\begin{proof}
 Let $\opt_\mwc$ denote the objective value of an optimum muliway cut in the weighted graph $H$. 
    By Lemma \ref{lem:rel-mwc-cut}, we have that $OPT_{\mwc} = 28|E(G)| - OPT_{\maxcut}$. 
The weight of each gadget graph is $54$ and hence, $w(E(H))=54|E(G)|$. 
    %Using corollary \ref{cor:rel-mwc-cut}, we get, $OPT_{3-way-\mwc} = 28|E(G)| - OPT_{\maxcut}$. Now, in graph $F$ we have that the total weight of all edges is $|E(G)| \cdot w(E(Gadget))$, i.e., $w(E(F)) = 54|E(G)|$.
Moreover, we know that $OPT_{\maxcut} \geq \frac{|E(G)|}{2}$. 
    Therefore, 
\begin{align*} \label{eq:cut-term}
    OPT_{\gcovmp} &= OPT_{\mwc} + w(E(H)) \notag\\ 
    &= 54|E(G)| + 28|E(G)| - OPT_{\maxcut} \notag\\
    &= 82|E(G)| - OPT_{\maxcut}\\
%\end{align}
%We recall that $OPT_{\maxcut} \geq \frac{|E(G)|}{2}$. Hence, we have that 
%\begin{align*}
    %OPT_{\gcovmp} 
    &\leq 82(2 \cdot OPT_{\maxcut}) - OPT_{\maxcut} = 163 \cdot OPT_{\maxcut}.
\end{align*}

\end{proof}

Lemma \ref{lem:rel-mwc-cut} also implies that a low cost solution to \gcovmp can be used to construct a large-valued cut as shown below.  
\begin{proposition} \label{prop:rel-2}
    Let $V_1, V_2, V_3$ be a partitioning of $V(H)$ such that $s_i\in V_i$ for each $i\in [3]$. 
    %Let $U$ be a connected component of $G-\{\{u,v\}\in E(G):\{u,v\}\subsetneq V_i\ \forall\ i\in [3]\}$. 
    Then, there exists $U\subseteq V(G)$ such that 
    \[
    d_G(U)\ge \opt_{\maxcut} +\opt_{\gcovmp}- \sum_{i=1}^3 b_H(V_i). 
    %\opt_{\maxcut}-d_G(U) \le \sum_{i=1}^3 b_H(V_i) - \opt_{\gcovmp}. 
    \]
    
    %Let $y$ be cost of feasible solution to $F$ and $x$ be the cut of corresponding \maxcut obtained from $G$ by cutting $(u, v)$ iff copies of $u, v$ are assigned to different terminals in graph coverage instance, then 
    %$|x - OPT_{\maxcut}| \leq |y - OPT_{3-terminal-\gcovmp}|$
    %$OPT_{\maxcut}\le x + y - OPT_{3-terminal-\gcovmp}$. 
\end{proposition}

\begin{proof}
    By \eqref{eq:cov-vs-cut-translation}, we have that $(1/2)\sum_{i=1}^3 d_H(V_i) = \sum_{i=1}^3 b_H(V_i) - w(E(H))= \sum_{i=1}^3 b_H(V_i) - 54|E(G)|$. Hence, using Lemma \ref{lem:rel-mwc-cut}, we have a subset $U\subseteq V(G)$ with
    \begin{align*}
        d_G(U) 
        &\ge 28|E(G)|-\left(\sum_{i=1}^3 b_H(V_i) - 54|E(G)|\right) \\
        &= 82|E(G)| - \sum_{i=1}^3 b_H(V_i)\\
        &= 82|E(G)| - \sum_{i=1}^3 b_H(V_i) + \opt_{\maxcut} - \opt_{\maxcut}\\
        & \ge \opt_{\maxcut} + 82|E(G)| - \sum_{i=1}^3 b_H(V_i)- \opt_{\maxcut} \quad \quad \text{(By Proposition \ref{prop:rel-1})}\\
        &= \opt_{\maxcut} + \opt_{\gcovmp}-\sum_{i=1}^3 b_H(V_i).
    \end{align*}

\iffalse
    Given a solution with cost $y$ to $3$-terminal-\gcovmp, the corresponding $3$-way-\mwc cost is $y - w(E(F)) = y - 56|E(G)|$. Now using lemma \ref{lem:rel-mwc-cut}, we get a solution to \maxcut instance corresponding to $G$ with cost $x$ such that
\begin{align*}
    &x \geq 28|E(G)| - (y - 56|E(G)|)\\
    \implies &x \geq 82|E(G)| - y.
\end{align*}
Hence we get, 
\begin{align*}
    |x - OPT_{\maxcut}| &= OPT_{\maxcut} - x  \\
    &\leq OPT_{\maxcut} - 82|E(G)| + y \\ 
    &= y - OPT_{\gcovmp} \quad \quad \text{(from equation (\ref{eq:cut-term}))}\\
    &= |y - OPT_{\gcovmp}|.
\end{align*}
\fi
\end{proof}

% \begin{theorem}
%     For every $k \geq 3$, approximating $3-terminal-\gcovmp$
% \end{theorem}

We now restate and prove Theorem \ref{thm:apx-cov-cut}. 
\thmMaxCutToGCovMP*
\begin{proof}
Given an instance $G=(V, E)$ of \maxcut, we construct the instance $H$ of \gcovmp with $3$ terminals as discussed at the beginning of the section. We apply the $\alpha$-approximation algorithm for \gcovmp to $H$, to obtain a partition $V_1, V_2, V_3$ of $V(H)$ such that $\sum_{i=1}^{3} b_H (V_i) \leq \alpha\cdot OPT_{\gcovmp}$. By Proposition \ref{prop:rel-2}, there exists $U \subseteq V(G)$ such that, 
\begin{align*} 
d_G(U)
&\ge \opt_{\maxcut} + \opt_{\gcovmp} - \sum_{i=1}^3 b_H(V_i)\\
&\ge \opt_{\maxcut} - (\alpha-1)\opt_{\gcovmp}\\
&\ge (164-163\alpha) \opt_{\maxcut}.
\end{align*}
The last inequality above is by Proposition \ref{prop:rel-1} and $\alpha\ge 1$. 
\iffalse
\begin{align} \label{eq:num}
(\alpha-1)\opt_{\gcovmp}\ge
    \sum_{i=1}^3 b_H(V_i) - \opt_{\gcovmp} \ge \opt_{\maxcut} - d_G(U).
\end{align}
Now, by (\ref{eq:num}) and proposition \ref{prop:rel-1}, we get, 
\begin{align*}
    \alpha - 1 \ge \frac{\sum_{i=1}^{3} b_H(V_i) - OPT_{\gcovmp}}{OPT_{\gcovmp}} \geq \frac{OPT_{\maxcut} - d_G(U)}{163 \cdot OPT_{\maxcut}}. 
\end{align*}
Hence we get, 
\begin{align*}
    d_G(U) \ge (1 - 163(\alpha - 1))\cdot OPT_{\maxcut}.
\end{align*}
\fi
\end{proof}

%% file: gadget-snp-hard.tex
\tikzset{every picture/.style={line width=0.75pt}} %set default line width to 0.75pt        

\begin{tikzpicture}[x=0.75pt,y=0.75pt,yscale=-1,xscale=1]
%uncomment if require: \path (0,330); %set diagram left start at 0, and has height of 330

%Shape: Circle [id:dp8244262753257188] 
\draw  [color={rgb, 255:red, 0; green, 0; blue, 0 }  ,draw opacity=1 ] (142,60) .. controls (142,55.03) and (137.97,51) .. (133,51) .. controls (128.03,51) and (124,55.03) .. (124,60) .. controls (124,64.97) and (128.03,69) .. (133,69) .. controls (137.97,69) and (142,64.97) .. (142,60) -- cycle ;
%Shape: Circle [id:dp7469570606535232] 
\draw   (142,129) .. controls (142,124.03) and (137.97,120) .. (133,120) .. controls (128.03,120) and (124,124.03) .. (124,129) .. controls (124,133.97) and (128.03,138) .. (133,138) .. controls (137.97,138) and (142,133.97) .. (142,129) -- cycle ;
%Shape: Circle [id:dp8398505210251146] 
\draw   (142,196) .. controls (142,191.03) and (137.97,187) .. (133,187) .. controls (128.03,187) and (124,191.03) .. (124,196) .. controls (124,200.97) and (128.03,205) .. (133,205) .. controls (137.97,205) and (142,200.97) .. (142,196) -- cycle ;
%Straight Lines [id:da7945303739826284] 
\draw    (133,69) -- (133,120) ;
%Straight Lines [id:da5133867841459927] 
\draw    (133,138) -- (133,187) ;
%Shape: Circle [id:dp649822205948644] 
\draw  [color={rgb, 255:red, 0; green, 0; blue, 0 }  ,draw opacity=1 ] (211,60) .. controls (211,55.03) and (206.97,51) .. (202,51) .. controls (197.03,51) and (193,55.03) .. (193,60) .. controls (193,64.97) and (197.03,69) .. (202,69) .. controls (206.97,69) and (211,64.97) .. (211,60) -- cycle ;
%Shape: Circle [id:dp7592962482068544] 
\draw   (211,129) .. controls (211,124.03) and (206.97,120) .. (202,120) .. controls (197.03,120) and (193,124.03) .. (193,129) .. controls (193,133.97) and (197.03,138) .. (202,138) .. controls (206.97,138) and (211,133.97) .. (211,129) -- cycle ;
%Shape: Circle [id:dp11850625017152394] 
\draw   (211,196) .. controls (211,191.03) and (206.97,187) .. (202,187) .. controls (197.03,187) and (193,191.03) .. (193,196) .. controls (193,200.97) and (197.03,205) .. (202,205) .. controls (206.97,205) and (211,200.97) .. (211,196) -- cycle ;
%Straight Lines [id:da6488313051794505] 
\draw    (202,69) -- (202,120) ;
%Straight Lines [id:da1737428401640324] 
\draw    (202,138) -- (202,187) ;
%Shape: Circle [id:dp2922044869450455] 
\draw  [color={rgb, 255:red, 0; green, 0; blue, 0 }  ,draw opacity=1 ] (280,60) .. controls (280,55.03) and (275.97,51) .. (271,51) .. controls (266.03,51) and (262,55.03) .. (262,60) .. controls (262,64.97) and (266.03,69) .. (271,69) .. controls (275.97,69) and (280,64.97) .. (280,60) -- cycle ;
%Shape: Circle [id:dp40724211591326664] 
\draw   (280,129) .. controls (280,124.03) and (275.97,120) .. (271,120) .. controls (266.03,120) and (262,124.03) .. (262,129) .. controls (262,133.97) and (266.03,138) .. (271,138) .. controls (275.97,138) and (280,133.97) .. (280,129) -- cycle ;
%Shape: Circle [id:dp6370530731966637] 
\draw   (280,196) .. controls (280,191.03) and (275.97,187) .. (271,187) .. controls (266.03,187) and (262,191.03) .. (262,196) .. controls (262,200.97) and (266.03,205) .. (271,205) .. controls (275.97,205) and (280,200.97) .. (280,196) -- cycle ;
%Straight Lines [id:da5349484997112408] 
\draw    (271,69) -- (271,120) ;
%Straight Lines [id:da31878238659583236] 
\draw    (271,138) -- (271,187) ;
%Straight Lines [id:da14940877175339318] 
\draw    (142,60) -- (193,60) ;
%Straight Lines [id:da1547604298237384] 
\draw    (142,129) -- (193,129) ;
%Straight Lines [id:da7098759417890057] 
\draw    (142,196) -- (193,196) ;
%Straight Lines [id:da8240027892456354] 
\draw    (211,60) -- (262,60) ;
%Straight Lines [id:da45120186974743803] 
\draw    (211,129) -- (262,129) ;
%Straight Lines [id:da9696995216341129] 
\draw    (211,196) -- (262,196) ;
%Curve Lines [id:da6637579253420853] 
\draw    (138,51.6) .. controls (186,25.6) and (232,35.6) .. (266,52.6) ;
%Curve Lines [id:da2723689453971512] 
\draw    (126,66.6) .. controls (112,74.6) and (94,144.6) .. (127,189.6) ;
%Curve Lines [id:da1699799253676415] 
\draw    (195,65.6) .. controls (181,73.6) and (163,143.6) .. (196,188.6) ;
%Curve Lines [id:da8619981238969174] 
\draw    (277,66.6) .. controls (302,96.6) and (307,148.6) .. (278,189.6) ;
%Curve Lines [id:da8957456616719122] 
\draw    (138,120.6) .. controls (186,94.6) and (232,104.6) .. (266,121.6) ;
%Curve Lines [id:da8692879327212113] 
\draw    (138,202.6) .. controls (175,223.6) and (230,224.6) .. (266,203.6) ;

% Text Node
\draw (134,82) node [anchor=north west][inner sep=0.75pt]   [align=left] {$\displaystyle 4$};
% Text Node
\draw (163,59) node [anchor=north west][inner sep=0.75pt]   [align=left] {$\displaystyle 4$};
% Text Node
\draw (97,121) node [anchor=north west][inner sep=0.75pt]   [align=left] {$\displaystyle 4$};
% Text Node
\draw (194,18) node [anchor=north west][inner sep=0.75pt]   [align=left] {$\displaystyle 4$};
% Text Node
\draw (203,83) node [anchor=north west][inner sep=0.75pt]   [align=left] {$\displaystyle 4$};
% Text Node
\draw (159,128) node [anchor=north west][inner sep=0.75pt]   [align=left] {$\displaystyle 4$};
% Text Node
\draw (231,128) node [anchor=north west][inner sep=0.75pt]   [align=left] {$\displaystyle 4$};
% Text Node
\draw (260,153) node [anchor=north west][inner sep=0.75pt]   [align=left] {$\displaystyle 4$};
% Text Node
\draw (203,152) node [anchor=north west][inner sep=0.75pt]   [align=left] {$\displaystyle 4$};
% Text Node
\draw (234,179) node [anchor=north west][inner sep=0.75pt]   [align=left] {$\displaystyle 4$};
% Text Node
\draw (299,119) node [anchor=north west][inner sep=0.75pt]   [align=left] {$\displaystyle 4$};
% Text Node
\draw (198,220) node [anchor=north west][inner sep=0.75pt]   [align=left] {$\displaystyle 4$};
% Text Node
\draw (207,41) node [anchor=north west][inner sep=0.75pt]   [align=left] {$\displaystyle x$};
% Text Node
\draw (114,136) node [anchor=north west][inner sep=0.75pt]   [align=left] {$\displaystyle y$};
% Text Node
\draw (105,48) node [anchor=north west][inner sep=0.75pt]   [align=left] {$\displaystyle s_{1}$};
% Text Node
\draw (205,111) node [anchor=north west][inner sep=0.75pt]   [align=left] {$\displaystyle s_{2}$};
% Text Node
\draw (267,207) node [anchor=north west][inner sep=0.75pt]   [align=left] {$\displaystyle s_{3}$};

\end{tikzpicture}